\documentclass[11pt,oneside,english]{amsart}
\usepackage[T1]{fontenc}
\usepackage[latin9]{inputenc}
\usepackage[a4paper]{geometry}
\usepackage{amsbsy}
\usepackage{amstext}
\usepackage{amsthm}
\usepackage{amssymb}

\makeatletter

\newcommand{\noun}[1]{\textsc{#1}}

\numberwithin{equation}{section}
\numberwithin{figure}{section}
\theoremstyle{plain}
\newtheorem{thm}{\protect\theoremname}[section]
\theoremstyle{plain}
\newtheorem{lem}[thm]{\protect\lemmaname}
\theoremstyle{definition}
\newtheorem{example}[thm]{\protect\examplename}
\theoremstyle{definition}
\newtheorem{defn}[thm]{\protect\definitionname}
\theoremstyle{remark}
\newtheorem{rem}[thm]{\protect\remarkname}
\theoremstyle{plain}
\newtheorem{prop}[thm]{\protect\propositionname}
\theoremstyle{plain}
\newtheorem{cor}[thm]{\protect\corollaryname}

\def\makebbb#1{
    \expandafter\gdef\csname#1\endcsname{
        \ensuremath{\Bbb{#1}}}
}\makebbb{R}\makebbb{N}\makebbb{Z}\makebbb{C}\makebbb{H}\makebbb{E}\makebbb{H}\makebbb{P}\makebbb{B}\makebbb{Q}\makebbb{E}

\usepackage{babel}

\usepackage{babel}

\makeatother

\usepackage{babel}

\makeatother

\usepackage{babel}
\providecommand{\corollaryname}{Corollary}
\providecommand{\definitionname}{Definition}
\providecommand{\examplename}{Example}
\providecommand{\lemmaname}{Lemma}
\providecommand{\propositionname}{Proposition}
\providecommand{\remarkname}{Remark}
\providecommand{\theoremname}{Theorem}

\begin{document}
\title{Ensemble equivalence for mean field models and plurisubharmonicity }
\author{Robert J. Berman}
\begin{abstract}
We show that entropy is globally concave with respect to energy for
a rich class of mean field interactions, including regularizations
of  the point-vortex model in the plane, plasmas and self-gravitating
matter in 2D, as well as the higher dimensional logarithmic interactions
appearing in conformal geometry and power laws. The proofs are based
on a corresponding ``microscopic'' concavity result at finite $N,$
shown by leveraging an unexpected link to Kähler geometry and plurisubharmonic
functions. Under more restrictive homogeneity assumptions strict concavity
is obtained using a uniqueness result for free energy minimizers,
established in a companion paper. The results imply that thermodynamic
equivalence of ensembles holds for this class of mean field models.
As an application it is shown that the critical inverse negative temperatures
- in the macroscopic as well as the microscopic setting - coincide
with the asymptotic slope of the corresponding microcanonical entropies.
Along the way we also extend previous results on the thermodynamic
equivalence of ensembles for continuous weakly positive definite interactions,
concerning positive temperature states, to the general non-continuous
case. In particular, singular situations are exhibited where, somewhat
surprisingly, thermodynamic equivalence of ensembles fails at energy
levels sufficiently close to the minimum energy level.
\end{abstract}

\maketitle

\section{\label{sec:Introduction}Introduction}

\subsection{\label{subsec:General-setup intro}General Setup}

Let $X$ be a topological space, $W$ a symmetric lower semi-continuous
(lsc) function on $X^{2}$ (the \emph{pair interaction potential})
and $V$ a lsc function on $X$ (the \emph{exterior potential}), both
taking values in $]-\infty,\infty].$ The corresponding $N-$particle
\emph{mean field $N-$particle Hamiltonian} is defined by 

\begin{equation}
H^{(N)}(x_{1},...,x_{N}):=\frac{1}{2}\frac{1}{N}\sum_{i\neq j\leq N}W(x_{i},x_{j})+\sum_{i=1}^{N}V(x_{i}).\label{eq:Hamilt for W and V intro}
\end{equation}
 The self-interactions have, as usual, been excluded in order to render
$H^{(N)}$ generically finite in the case when $W$ is singular on
the diagonal. The corresponding \emph{(macroscopic) energy $E(\mu)$
}of a probability measure $\mu$ on $X,$ i.e. $\mu\in\mathcal{P}(X),$
is defined by 
\begin{equation}
E(\mu):=\frac{1}{2}\int_{X}W\mu\otimes\mu+\int_{X}V\mu\in]-\infty,\infty]\label{eq:def of E intro}
\end{equation}
when $X$ has compact support. The definition can be extended to non-compactly
supported measures (see Section\ref{subsec:Existence-when-X is noncompact}),
but for most purposes it will be enough to consider the restriction
of $E(\mu)$ to the space of all probability measures on $X$ with
compact support, denoted by $\mathcal{P}(X)_{0}.$ 

Now, fix also a measure $\mu_{0}$ on $X$ ( the ``prior''). Following
\cite{e-s,clmp2,e-h-t,c-d-r} the \emph{entropy} $S(e)$ (at energy
$e)$ is the function on $\R$ defined by $-\infty$ in the case that
$\{E(\mu)=e\}$ is empty and otherwise
\begin{equation}
S(e):=\sup_{\mu\in\mathcal{P}(X)_{0}}\left\{ S(\mu):\,\,E(\mu)=e\right\} ,\,\,\,\,\,\,S(\mu):=-\int_{X}\log(\mu/\mu_{0})\mu\label{eq:def of S e}
\end{equation}
 where $S(\mu)$ is the\emph{ entropy} of $\mu$ relative $\mu_{0},$
which, by definition, is equal to $-\infty$ if $\mu$ is not absolutely
continuous wrt $\mu_{0}.$ A measure $\mu^{e}$ maximizing $S(\mu)$
above is called a \emph{maximum entropy measure}. In the case when
$\mu_{0}$ is a probability measure, that we shall focus on,
\[
\max_{e\in\R}S(e)=S(e_{0})=0,\,\,\,e_{0}:=E(\mu_{0})
\]
 (since $S(\mu)\leq0$ with equality iff $\mu=\mu_{0}).$ This setup
is modeled on\emph{ }repulsive\emph{ }Hamiltonians (as in the case
of identical point vortices described below), but an equivalent setup
of ``attractive'' Hamiltonians is obtained by replacing $H^{(N)}$
with $-H^{(N)}$ and $e$ with $-e.$ 

\subsection{\label{subsec:Concavity-of-}Background: concavity of $S(e)$ and
thermodynamic equivalence of ensembles}

In the case when $E(\mu)$ is linear on $\mathcal{P}(X)$ (i.e. $W=0)$
it follows directly from the concavity of $S(\mu)$ on $\mathcal{P}(X)$
that the  entropy $S(e)$ is concave with respect to $e$ (see Section
\ref{subsec:Concavity-and monot}). This is the standard setup in
information theory and statistical inference, going back to Shannon
and Jaynes \cite{ja}, but here we will be concerned with the case
when $E(\mu)$ is quadratic, motivated by mean-field models in statistical
mechanics (see \ref{eq:W and V in disc for vortex} for a comparison
between the classical linear setup and the quadratic setup appearing
in the context of plasmas). General, non-linear $E(\mu)$ also appear
naturally in engineering optimization \cite{a-c-p}. The concavity
properties of the entropy $S(e)$ for mean-field models and other
systems with long-range interactions have been studied extensively
from various points of views in the last decades; theoretical as well
as experimental and numerical \cite{e-s,clmp2,c-d-r,e-t-t,d-r-a-w}.
As stressed in \cite{c-d-r} the unusual properties of these systems
stem from the lack of additivity (i.e. lack of linearity of $E(\mu)$).
In particular, the question whether $S(e)$ is globally concave is
crucial in connection to negative temperature states in Onsager's
point-vortex model for the large-time limit of turbulent incompressible
non-viscous 2D fluids; classical fluids \cite{e-s,clmp2}, as well
as quantum fluids \cite{gau,jo}. In the case of vortices of equal
circulation moving in the whole plane $\R^{2}$ the vortex-vortex
pair interaction potential $W(x,y)$ is proportional to the Green
function for the Laplacian in $\R^{2},$ 
\begin{equation}
W(x,y)=-\log(|x-y|).\label{eq:def of W vortex intro}
\end{equation}
As amphasized in \cite{e-h-t,t-e-t,c-d-r} the relevance of the global
concavity of $S(e)$ stems from the fact that it equivalently means
that $S(e)$ may (under appropriate regularity assumptions discussed
in Section \ref{sec:Thermodynamical-equivalence}) be expressed as
the Legendre-Fenchel transform of the \emph{Helmholtz (scaled) free
energy} $F(\beta)$ at inverse temperature $\beta:$ 
\begin{equation}
S(e)=\inf_{\beta\in\R}(-F(\beta)+\beta e),,\label{eq:S as Leg intro}
\end{equation}
 where $F(\beta)$ is defined as the infimum of the (scaled) free
energy functional $F_{\beta}(\mu):$ 
\begin{equation}
F(\beta)=\inf_{\mathcal{\mu\in P}(X)_{0}}F_{\beta}(\mu),\,\,\,F_{\beta}(\mu):=\beta E(\mu)-S(\mu),\label{eq:def of F beta of F beta mu intro}
\end{equation}
(where $F(\mu)$ is defined to be equal to $+\infty$ if $S(\mu)=-\infty).$
Accordingly, when $S(e)$ is globally concave\emph{ thermodynamic
equivalence of ensembles} is said to hold \cite{e-h-t,t-e-t,c-d-r}
(since it amounts to the equivalence between the \emph{microcanonical
ensemble }at a fixed energy $e$ and the \emph{canonical ensemble}
at a corresponding fixed inverse temperature $\beta,$ in the large
$N-$limit). More generally, this duality fits into primal-dual formulations
of non-linear optimization problems, where the free energy functional
appears as the augmented Lagrangian \cite{a-c-p}. If thermodynamic
equivalence of ensembles holds then $S(e)$ is differentiable at almost
any energy level $e$ and, by the concavity of $F(\beta),$ the infimum
over the $\beta$ in formula \ref{eq:S as Leg intro} is attained
precisely at the inverse temperature
\[
\beta(e)=\frac{dS(e)}{de},
\]
Remarkably, as stressed already by Onsager in the late 40s \cite{ons},
this means, since $S(e)$ is decreasing for $e>E(\mu_{0}),$ that
in the\emph{ ``high energy regime'' }
\begin{equation}
E(\mu_{0})<e\label{eq:high energy regime intro}
\end{equation}
an energy level $e$ should correspond to a\emph{ negative }inverse
temperature $\beta.$ As a consequence, the repulsive vortex interaction
should then become effectively attractive, resulting in the aggregation
of microscopic vortices of equal circulation into large-scale coherent
clusters (as observed in oceanic and atmospheric fluids, notably Jupiter's
famous great red spot). A few years after Onsager's prediction the
existence of negative temperature states was experimentally verified
in nuclear spin systems \cite{ra}, while the original prediction
was quantitatively experimentally demonstrated only very recently
in a 2D quantum superfluid (a Bose Einstein condensate \cite{gau,jo}).
Note that the high energy region only exists if 
\[
E(\mu_{0})<\infty.
\]
 This is automatically the case if $X$ is compact and $W$ and $V$
are locally integrable, but it also holds in many non-compact situations,
for example the vortex model in $\R^{2},$ when $\mu_{0}$ is taken
as a Gaussian probability measure (incorporating conservation of angular
momentum).

As shown in \cite{e-s}, if $W(x,y)$ defines a weakly positive definite
kernel, as in the point-vortex model, then  the concavity of $S(e)$
holds in the\emph{ ``low energy regime'':}
\[
e\leq E(\mu_{0}),
\]
 which corresponds to positive inverse temperature $\beta$ (more
precisely, in \cite{e-s} it is assumed that $W$ is continuous; the
general case is discussed in Section \ref{subsec:Concavity-and monot}).
However, the concavity may fail in the ``high energy regime'' and
thus the correspondence with negative temperature then breaks down
(leading to the peculiar phenomenon of negative heat-capacity \cite{lb,c-d-r}).
This is illustrated by the mean-field Blume-Emery-Griffiths spin model
in \cite[Section 4.2.4]{c-d-r}. In the case of the point-vortex model
the global concavity of $S(e)$ has been established when $X$ is
the unit-disc in $\R^{2}$ \cite{clmp2} (or a sufficiently small
deformation of the unit-disc) or all of $\R^{2}$ \cite{clmp2,c-k,k2},
while shown to fail for some domains $X$ (e.g. a sufficiently thin
rectangle). The proofs in \cite{clmp2,c-k} exploit that in the case
of the point-vortex model any minimizer $\mu_{\beta}$ of the free
energy functional $F_{\beta}(\mu)$ satisfies a second order PDE (the
Joyce-Montgomery mean field equation/the Liouville equation). This
opens the door for the application of various PDE-techniques ( uniqueness
results, concentration/compactness alternatives, symmetrization arguments,
...). 

\subsection{\label{subsec:Summary-of-the}Summary of the main results }

To the best of the authors knowledge there are, apart from a few special
cases - such as the BEG-model and the vortex model recalled above
- no general global concavity results for mean field Hamiltonians.
Even in the case of the regularized vortex model \cite{e-s} the question
of global concavity of $S(e)$ raised in \cite{e-s} appears to have
been left open. Similar questions have also been put forward in the
context of self-gravitating matter, where regularizations appear naturally
\cite{k0}. Allowing regularizations is also crucial when comparing
theoretical results with numerical simulations (such as \cite{si-r,d-g-c})
to ensure that the concavity of $S(e)$ is a robust feature of the
models in question. The main aim of the present work is to establish
the concavity of $S(e)$ for a rich class of potentials $W,V$ and
priors $\mu_{0},$ including the point-vortex model in $\R^{2},$
as well as its regularizations and regularized plasmas and self-gravitating
systems in 2D and power-laws. However, since neither explicit calculations,
nor PDE-techniques are available for such interactions $W$ we take
a different route. First it is shown that the \emph{(upper) microcanonical
entropy} (at energy $e)$.
\begin{equation}
S_{+}^{(N)}(e):=\frac{1}{N}\log\mu_{0}^{\otimes N}\{H^{(N)}/N>e\},\label{eq:def of S N plus intro}
\end{equation}
is concave for any finite $N$ (in the context of the point vortex
model this microcanonical entropy appears in \cite[Theorem 4.2]{clmp2}).
Then, letting $N\rightarrow\infty,$ and using the asymptotics for
$S_{+}^{(N)}(e)$ from \cite{e-s}, the concavity of the upper entropy
$S_{+}(e)$ is obtained (defined by replacing the condition $E(\mu)=e$
in the definition \ref{eq:def of S e} of $S(e)$ with the condition
$E(\mu)\geq e).$ Hence, the concavity of $S(e)$ in the high energy
region, $e>E(\mu_{0}),$ results from the observation that $S_{+}(e)=S(e)$
there. This derivation of the concavity of $S(e)$ is thus in the
spirit of statistical mechanics; the macroscopic property in question
emerges from a microscopic one. The proof of the concavity of $S_{+}^{(N)}(e)$
leverages some developments in Kähler geometry \cite{bern,b-b}, centered
around complex analogs of the Brunn-Minkowski inequality. Under more
restricted assumptions $S(e)$ is shown to be\emph{ strictly} concave,
using a different (macroscopic) approach - which is more in the spirit
of \cite{clmp2} - based on a uniqueness result for free energy minimizers
of independent interest. 

Before turning to a more precise description of the main present results
it may be worth emphasizing that the concavity of $S_{+}^{(N)}(e)$
is considerably stronger than the concavity of $S(e)$ and does not
require the mean field scaling (nor the permutation symmetry). Thus
it also applies to the microcanonical study of\emph{ small systems},
considered in the physics literature (see, for example, \cite{ru,gr,d-h}).
The relation to the setup in \cite{ru,d-h} becomes clearer in the
equivalent setup of ``attractive'' Hamiltonians obtained by replacing
the $H^{(N)}$ with the Hamiltonian $-H^{(N)}$ and $e$ with $-e.$
The concavity of $S_{+}^{(N)}(e)$ then translates into the concavity
of
\begin{equation}
S_{-}^{(N)}(e):=\frac{1}{N}\log\mu_{0}^{\otimes N}\{H^{(N)}/N<e\},\label{eq:def of S N minus}
\end{equation}
 called the bulk-entropy in \cite{ru} and the microcanonical Gibbs
entropy in \cite{d-h}. In recent years it has been debated whether
this microcanonical entropy is physically more revelant than the microcanonical
Boltzman entropy, obtained by replacing the volume $\mu_{0}^{\otimes N}\{H^{(N)}/N<e\}$
with its derivative with respect to $e$ (the surface area of the
level-set $H^{(N)}/N=e);$ see the discussion in \cite{h-p-d} and
references therein. The present results may, perhaps, be interpreted
as a case for bulk/Gibbs entropy as this entropy is shown to be concave
in our class of ``attractive'' Hamiltonians, while the Boltzmann
entropy is not always concave in this class (as discussed in connection
to Theorem \ref{thm:conv for N for two-point etc}). On the other
hand, in the case when $\mu_{0}$ is Lesbesgue measure on $\R^{2n}$
the bulk/Gibbs and the Boltzmann entropy coincide in the limit when
$N\rightarrow\infty$ (in the classical thermodynamical limit this
is discussed in \cite[Section 6.2]{hu} and in the present mean field
setup it can be shown that both limits coincide with $S(e).$ \cite{be2} 

Let now $X$ be a (possible non-compact) subset of $\R^{2n}$ end
let $\phi$ be a defining function for $X,$ i.e. a continuous function
such that 
\[
X=\{\phi\leq0\}.
\]
 Endow $X$ with a probability measure $\mu_{0}$ which is absolutely
continuous wrt Lebesgue measure $d\lambda:$ 
\[
\mu_{0}=e^{-\Psi_{0}}d\lambda.
\]
 on $\R^{2n}.$ We will identify $\R^{2n}$ with $\C^{n}$ in the
usual way and denote by $(z_{1},...,z_{n})$ the standard holomorphic
coordinates on $\C^{n}.$ 

\subsubsection*{Concavity of $S_{+}^{(N)}(e)$ and $S(e)$ in the high energy region
$e>e_{0}$}

The main results, saying that upper microcanonical entropy $S_{+}^{(N)}(e)$
and the entropy $S(e)$ are concave in the high-energy regime \ref{eq:high energy regime intro}
(Theorem \ref{thm:conv for N for two-point etc} and Theorem \ref{thm:concave S in text}),
are shown to hold under appropriate plurisubharmonicity and symmetry
properties of the data. Denoting by $PSH_{\boldsymbol{a}}$ the class
of all plurisubharmonic functions which are invariant under the action
\begin{equation}
(z_{1},...,z_{n})\mapsto(e^{ia_{1}\theta}z_{1},...,e^{ia_{n}\theta}z_{n})\label{eq:action by theta intro}
\end{equation}
 for any $\theta\in\R,$ for a given ``weight vector'' $\boldsymbol{a}\in]0,\infty[^{n},$
the main results hold under the following

$\medskip$

\paragraph*{\emph{\noun{Main Assumptions: }}$\phi,\Psi_{0},-V$ are in the class
$PSH_{\boldsymbol{a}}(\C^{n})$ and $-W$ is in $PSH_{\boldsymbol{a,a}}(\C^{n}\times\C^{n})$
for some $\boldsymbol{a}\in]0,\infty[^{n}$}

$\medskip$

The definition of plurisubharmonicity is recalled in Section \ref{subsec:Background-on-plurisubharmonicit}.
For the moment we just point out that the class $PSH_{\boldsymbol{a}}$
is very rich. For example, when the weights $a_{i}$ are positive
integers the class $PSH_{\boldsymbol{a}}$ contains the functions
\begin{equation}
\psi(z)=\log(\sum_{j=1}^{r}|P_{j}(z)|^{2})\label{eq:algebraic case intro}
\end{equation}
 where $P_{j}$ is a polynomial in $z_{1},..,z_{n},$ which is homogeneous
wrt  the scaling action by $\C^{*}$ on $\C^{n}$ with weights $\boldsymbol{a}.$
In particular, for any $\boldsymbol{a}$ the class $PSH_{\boldsymbol{a}}$
contains $\psi(z)=\log|z|$ as well as \textbf{$\Psi_{0}(z):=\sum_{i=1}^{n}\lambda_{i}|z_{i}|^{2},$}
for any positive $\lambda_{i}.$ Hence the Main Assumptions apply
to the corresponding Gaussian measures 
\begin{equation}
\mu_{0}=e^{-\sum_{i=1}^{n}\lambda_{i}|z_{i}|^{2}}d\lambda.\label{eq:gaussian measures}
\end{equation}
 In the case when the data is invariant under rotations of the $z_{i}-$variables
this is - from a physical point of view - the most natural choices
of priors, as they incorporate preservation of angular momentum in
the $z_{i}-$variables (see the discussion in Section \ref{subsec:Priors-versus-linear}).

An important general feature of the class $PSH_{\boldsymbol{a}}(\C^{n})$
is that is closed under scaling by positive numbers, taking sums and
maxima, as well as under composition with a complex linear map on
$\C^{n}$ or an increasing convex function, defined on the range of
a given $\psi\in PSH_{\boldsymbol{a}}(\C^{n}).$ This means, in particular,
that the Main Assumptions are stable under a range of different regularizations
of the data. For example, the Main Assumptions apply to the point-vortex
model in $X:=\R^{2}$ (formula \ref{eq:def of W vortex intro}) endowed
with a centered Gaussian measure. But the Main Assumptions also apply
to the standard continuous regularization and smooth regularization
of the point-vortex model where, for a given positive number $\delta,$
the pair interaction $W(x,y)$ is, in the continuous case, modified
so that it is constant on $|x-y|\leq\delta,$ while the smooth regularization
is defined by 
\[
W_{\delta}(x,y)=-\frac{1}{2\pi}\log(|x-y|+\delta).
\]
More generally, they apply to the regularizations obtained by convolution
of $-\log|x|$ with a positive sufficiently rapidly decreasing density
on $\R^{2}$, as used in the vortex blob model \cite[Section 6.2.1]{m-b}
(or more generally to the convolution of $-\log|x-y|$ with a smooth
density on $\R^{2}\times\R^{2}).$ An abundance of other examples
in $PSH_{\boldsymbol{a}}$ may be obtained by replacing $\psi$ in
formula with $\chi\circ\psi$ for any convex increasing function $\chi.$ 

Imposing translational and rotational symmetry the Main Assumptions
apply, in particular, under the

$\medskip$

\paragraph*{\emph{\noun{Homogeneous Assumptions:}}}
\begin{itemize}
\item $X$ is either a ball of radius $R$ centered at the origin in $\R^{2n}$
or equal to all of $\R^{2n}$ 
\item $W(x,y)=w(|x-y|),V(x)=v(|x|)$ and $\Psi_{0}(x)=\psi_{0}(|x|)$ with
$w(r),v(r)$ and $-\psi_{0}(r)$ \emph{concave }functions of\emph{
$\log r$ }(when $0<r\leq2R)$ and bounded from below as $r\rightarrow0.$ 
\end{itemize}
$\medskip$

In fact, the special assumptions imply that $w(r)$ is decreasing
in $r.$ In other words, the Homogeneous Assumptions equivalently
mean that the pair interaction $W(x,y)$ is repulsive and a concave
function of $\log|x-y|.$ The special assumptions, apply, for example,
to the continuous repulsive power-laws 
\begin{equation}
W_{\alpha}(x,y):=-|x-y|^{\alpha},\,\,\,\alpha>0.\label{eq:cont power law}
\end{equation}

Note that the Homogeneous Assumptions apply, in particular, to the
standard centered Gaussian probability measure $\mu_{0}$ on $\R^{2n}.$
However, one virtue of the Main Assumptions is that they, as pointed
out above, apply to the more general Gaussian measures \ref{eq:gaussian measures}
incorporating conservation of angular momentum in the $z_{i}-$variables
(as discussed in Section \ref{subsec:Priors-versus-linear}).

\subsubsection*{Global concavity of $S(e)$ and thermodynamic equivalence of ensembles}

In Section \ref{sec:Concavity-of-the} it is shown that if the assumption
that $W(x,y)$ be weakly positive definite is added to the Main Assumptions,
then $S(e)$ is globally concave, i.e. concave on all of $\R$ (Theorem
\ref{thm:Main-assumptions+weakly-positiv}) and finite on $]e_{min},e_{max}[.$
For example, as explained in Section \ref{subsec:Examples}, this
applies to the logarithmic interaction in $\R^{2n},$ as well as the
continuous power-laws \ref{eq:cont power law} when $a\in]0,2]$ and
to the exponential pair-potential
\[
W(x,y)=e^{-a|x-y|},\,\,\,a>0
\]
when $X$ is taken to be a disc centered at the origin with radius
at most $1/2a$ (known as the \emph{Born-Mayer potential} in chemistry).
It should be stressed that neither the power-laws with $a\in]0,1[,$
nor the exponential pair-potential, are concave wrt $(x,y)$ (otherwise
the concavity of $S_{\text{+}}^{(N)}(e)$ could also be deduced from
the ordinary Brunn-Minkowski inequality; compare Remark \ref{rem:B-M}). 

We then deduce that thermodynamic equivalence of ensembles holds for
any energy-level $e$ in $]e_{min},e_{max}[$ using a general result
(Theorem \ref{thm:convex cont iff energy appr}), saying that for
a general lower semi-continuous convex energy functionals $E(\mu)$
and prior $\mu_{0}$ thermodynamic equivalence of ensembles holds
in the low-energy region $]e_{min},e_{0}[$ iff $E(\mu)$ and $\mu_{0}$
satisfy a certain compatibility property (the ``energy approximation
property''). This property has previously appeared in connection
to the study of large deviation principles for the corresponding canonical
ensembles at positive inverse temperatures $\beta$ \cite{c-g-z,be1}. 

We also show that the global concavity of $S(e)$ holds for singular\emph{
}repulsive power-laws (Prop \ref{prop:cata}). However, in contrast
to the continuous power-laws \ref{eq:cont power law} (and the repulsive
logarithmic interaction) the singular power-laws do not satisfy the
Main Assumptions. In fact, in this case the global concavity of $S(e)$
in high-energy region $e\geq e_{0}$ holds for a bad reason: $S(e)\equiv S(e_{0})$
and, as a consequence, there are no maximum entropy measures $\mu^{e}$
when $e>e_{0}.$ This means that the equivalence of ensembles at the
level of\emph{ macrostates} then breaks down (see Section \ref{sec:Macrostate-equivalence-of}).
Similarly, regularized singular power-laws are expected to yield non-equivalent
ensembles and thus the corresponding entropies are expected to be
non-concave (as discussed in \cite[Page 252]{k0}).

\subsubsection*{Critical negative inverse temperatures and existence of maximum entropy
measures}

The singularity structure of a pair interaction $W(x,y)$ satisfying
the Main Assumptions can be very complicated, even if $W(x,y)$ is
taken to be translationally invariant, i.e.
\begin{equation}
W(x,y)=-\Psi(x-y)\label{eq:W transl invariant intro}
\end{equation}
 for a function $\Psi$ in the class $PSH_{\boldsymbol{a}}(\C^{n}).$
Still, as Shown in Section Section \ref{subsec:Bound-on-the non-isotr},
the singularities are mild enough to ensure that both the\emph{ microscopic
critical inverse temperature }
\[
\beta_{c,N}:=\left\{ \beta\in\R:\,Z_{N,\beta}:=\int_{X^{N}}e^{-\beta H^{(N)}}\mu_{0}^{\otimes N}<\infty\right\} 
\]
and the \emph{macroscopic critical inverse temperature} 
\[
\beta_{c}:=\inf:\left\{ \beta\in\R:\,\inf_{\mu}F_{\beta}(\mu)>-\infty\right\} 
\]
 are strictly negative. As a consequence we deduce that, when $X$
is compact, there exists a maximum entropy measure $\mu^{e}$ for
any $e\in]e_{min},e_{max}[.$ The concavity of $S_{+}^{(N)}(e)$ and
$S(e)$ is exploited to establish ``dual'' formulas for $\beta_{c,N}$
and $\beta_{c},$ which hold under the Main Assumptions (Corollary
\ref{cor:integrability threshold as slope} and Corollary \ref{cor:slope formula for beta c}):
\begin{equation}
\beta_{c,N}=\lim_{e\rightarrow\sup_{X^{N}}E_{N}}\frac{dS^{(N)}(e)}{de},\,\,\,\,\,\beta_{c}=\lim_{e\rightarrow\sup_{\mathcal{P}(X)}E(\mu)}\frac{dS(e)}{de}\label{eq:slope formulas for critical betas in intro}
\end{equation}
(which are decreasing limits when using either left or right derivatives).
The derivative $\frac{dS^{(N)}(e)}{de}$ corresponds to the \emph{inverse
Gibbs temperature} a\emph{t energy $e$ }in the context of small systems
\cite{d-h,h-p-d} (when $H^{(N)}$ is replaced by $-H^{(N)}$and $e$
with $-e$ so that $\frac{dS^{(N)}(e)}{de}$ is positive).

Applied to the regularized vortex model $W_{\delta}$ in $\R^{2}$
the second formula in \ref{eq:slope formulas for critical betas in intro}
confirms the expectations expressed in \cite[Page  855]{e-s}, concerning
the slope $dS_{\delta}(e)/de$ of the corresponding entropy: on the
one hand, as $e$ converges to the maximum (finite) value of the corresponding
regularized energy $E_{\delta}(\mu)$ the entropy $S_{\delta}(e)$
and its slope $dS_{\delta}(e)/de$ both converge towards $-\infty.$
On the other hand, for a \emph{fixed} $e$ the slope $dS_{\delta}(e)/de$
converges, as $\delta\rightarrow0,$ to the slope $dS_{0}(e)/de$
for the point-vortex model, which, in turn, is close to $-4$ for
large $e$ (with our normalizations). 

\subsection{\label{subsec:Outlook}Outlook}

In the companion papers \cite{be,be2,be3} elaborations of the main
results stated above are given, which may be summarized as follows.

\subsubsection*{Attractive classical Hamiltonians and self-gravitating matter in
2D}

The concavity results discussed above also apply if $H^{(N)}$ is
replaced by $-H^{(N)}$ if the energy level $e$ is replaced by $-e.$
In particular, when $V=0$ this amounts to replacing a ``repulsive''
pair-interaction $W(x,y)$ with an ``attractive'' pair-interaction
$U(x,y)\in PSH_{\boldsymbol{a,a}}(X\times X).$ It is then natural
to include momentum variables $p\in\R^{2n}$ and consider the corresponding
classical mean field Hamiltonian on phase space $(\R_{x}^{2n}\times\R_{p}^{2n})^{N}$
obtained by adding kinetic energy to the potential pair-interaction
energy $U(x,y):$
\[
H^{(N)}(x_{1},p_{1},...,x_{N},p_{N}):=\frac{1}{2}\frac{1}{N}\sum_{i\neq j\leq N}U(x_{i},x_{j})+\frac{1}{2}\sum_{i=1}^{N}|p_{i}|^{2}
\]
The phase space $X\times\R_{p}^{2n}$ is endowed with the standard
prior $\mu_{0}$ induced from Lebesgue measure $dxdp$ on $\R_{x}^{2n}\times\R_{p}^{2n}.$
Accordingly, the upper microcanonical entropy $S_{+}^{(N)}$ is now
replaced by $S_{-}^{(N)}(e)$ (formula \ref{eq:def of S N minus}).
One new feature in this setup is that the prior $\mu_{0}$ has infinite
mass. This leads to some technical difficulties appearing in the large
$N-$limit which are deferred to the companion paper \cite{be2}.
Here we we just briefly mention the main results shown in \cite{be2}.
First the global concavity of $S_{-}^{(N)}$ follows directly from
Theorem \ref{thm:conv for N for two-point etc}. As a consequence,
$S(e)$ is shown to also be globally concave under an additional stability
assumption, which ensures the existence of a maximum entropy measure
$\mu^{e}.$ The stability assumption in question is satisfied if,
for example, the assumption $U(x,y)\in PSH_{\boldsymbol{a,a}}(X\times X)$
is complemented with the assumption that $U(x,y)$ is, up to a bounded
term, translationally invariant. 

We will in particular specialize to the case when $U(x,y)$ is the
2D gravitational (Newtonian) pair interaction and its various regularizations,
generalizing a result in \cite{al}, concerning the unregularized
case. This case has been studied extensively in the astrophysics literature,
for example, as a model of galactic filaments \cite{d-g-c} . The
global concavity of $S(e)$ in for 2D-gravity should be contrasted
with the 3D case where $S(e)$ is identically equal to $\infty$ \cite{an}. 

\subsubsection*{Strict concavity of $S(e)$ and the Q-curvature equation}

In the final Section \ref{sec:strict} it is shown that if, under
the Homogeneous Assumptions, $v$ is moreover assumed strictly concave
(or, in the case when $X$ is the ball, $w$ is strictly decreasing),
then $S(e)$ is\emph{ strictly} concave. This follows from a uniqueness
result of free energy minimizers established in the companion paper
\cite{be}. Applications to conformal geometry are also given in \cite{be},
concerning the case when 
\[
W(x,y)=-\log|x-y|
\]
in $\R^{2n}.$ This pair potential is proportional to the Green kernel
of the $n$ th power $\Delta^{n}$ of the Laplacian, which, due to
its conformal invariance, plays a key role in conformal geometry and
mathematical physics \cite{ki2,chang}. As stressed in \cite{ki2}
the potential of the corresponding maximum entropy measures $\mu^{e}$
solve the \emph{Q-curvature equation }\cite{chang} with prescribed
Q-curvature proportional to $e^{-\Psi_{0}}:$ 
\[
-\Delta^{n}\phi=e^{-\beta\phi}e^{-\Psi_{0}}
\]
 The relevance of the present results to the Q-curvature equation
are discussed in the companion paper \cite{be}. 

\subsubsection*{A maximum entropy principle for Sasaki-Einstein metrics and AdS/CFT}

The present results arose as a ``spin-off effect'' of a microcanonical
approach to emergent Sasaki-Einstein geometry and the AdS/CFT (gauge/gravity)
correspondence between superconformal supersymmetric gauge theories
and supergravity. This is the subject of a separate publication \cite{be3},
but briefly the main results may be summarized as follows. Let $X$
be an $n-$dimensional complex algebraic subvariety of $\C^{m}$ which
is invariant under the action \ref{eq:action by theta intro} in $\C^{m}$
for some weight-vector $\boldsymbol{a}.$ Then $X$ is also invariant
under the scaling action on $\C^{m}$ by $\R_{>0}$ obtained by replacing
$e^{i\theta}$ with a scalar in $\R_{>0}.$ Denote by $x_{0}$ the
point in $X$ defined by the origin in $\C^{m}$ and assume that $X-\{x_{0}\}$
is non-singular. This means that the quotient $M:=(Y-\{y_{0}\})/\R_{>0}$
is a compact manifold, which is diffeomorphic to the intersection
of $Y$ with the unit-sphere in $\C^{m}.$ In the AdS/CFT correspondence
the classical moduli space of vacua of the rank $N$ gauge theory
is parameterized by $X^{N}$ and the corresponding supergravity vacuum
is encoded by a Sasaki-Einstein metric on the corresponding ``horizon''
$M;$ the $\R_{>0}-$action on $X$ represents the conformal symmetry
of the gauge theory. As observed in \cite{b-c-p} the rank $N$ gauge
theory admits a canonical fermionic $N-$particle BPS-state $\Psi_{\det},$
represented by a holomorphic polynomial on $X^{N}.$ Its self-information
\[
H^{(N)}(x_{1},...,x_{N}):=-\log|\Psi_{\det}(x_{1},...,x_{N})|^{2}
\]
 thus defines an effective symmetric Hamiltonian on $X^{N}.$ Dividing
$H^{(N)}$ by $N\lambda_{N}$ where $\lambda_{N}$ is the BPS R-charge
of $\Psi_{\det}$ (i.e. its scaling degree) one obtains, in the limit
when $N\rightarrow\infty,$ a macroscopic energy functional $E(\mu)$
on $\mathcal{P}(X),$ which is non-quadratic when $n>2.$ \footnote{More precisely, the convergence holds in the sense of $\Gamma-$convergence
when $X$ is quasi-regular, i.e. $\boldsymbol{a}$ has rational components,
but a variant of the argument applies in general, using an approximation
argument. }Assuming that the algebraic variety $X$ is Gorenstein there is also
a canonical $\R_{>0}-$equivariant volume form $dV_{X}$ on $X.$
This volume form has infinite total volume, but ``truncating'' it
appropriately one obtains a measure $\mu_{0}$ on $Y$ of finite total
volume. For example, one can multiply $dV_{X}$ with a Gaussian factor
$e^{-\epsilon r^{2}},$ where $r$ denotes the radial coordinate on
$\C^{m},$ or replace $X$ with its intersection with a ball of a
fixed radius $R$ in $\C^{m}.$ Then the corresponding entropy function
$S(e)$ is defined as in formula \ref{eq:def of S e}. It is shown
in \cite{be3} that $S(e)$ is is strictly concave and for any fixed
energy-level $e$ there exists a unique maximum entropy measure $\mu^{e}$
on $X.$ Moreover, as $e\rightarrow\infty$ $M$ the push-forward
$\nu^{e}$ of $\mu^{e}$ to $M,$ under the natural quotient projection
from $X$ to $M,$ converges to the volume form of Sasaki-Einstein
metric on $M$ - if such a metric exists. Otherwise, the reduced energy
of $\nu^{e}$ blows up as $e\rightarrow\infty.$ The convergence as
$e\rightarrow\infty$ is obtained by establishing a macroscopic equivalence
of ensembles in this setup, which shows that $\mu^{e}$ minimizes
the free energy functional introduced in \cite{b-c-p}, for an appropriate
$\beta.$

\subsection{Acknowledgments}

Thanks to Bo Berndtsson for many stimulating discussions on the topic
of \cite{b-b}. This work was supported by grants from the Knut and
Alice Wallenberg foundation, the Göran Gustafsson foundation and the
Swedish Research Council.

\subsection{Organization}

We start in Section \ref{sec:Preliminaries-and-notation} by introducing
a very general setup and provide some background on concavity and
on plurisubharmonic functions ( appearing in the Main Assumptions).
In Section \ref{sec:Thermodynamical-equivalence} general properties
of the entropy $S(e)$ are studied. In particular, finiteness and
monotonicity properties of $S(e)$ are established and relations to
the notion of\emph{ thermodynamic equivalence of ensembles} are explored.
In the following Section \ref{sec:Macrostate-equivalence-of} the
notion of \emph{macrostate equivalence of ensembles} is discussed
and existence results for maximum entropy measures are provided. Then,
in Section \ref{sec:Concavity-of-S low-energy} we consider the case
when $E(\mu)$ is convex and show that thermodynamic equivalence of
ensembles holds in the low-energy region $\{e>e_{0}\}$ iff the\emph{
energy approximation property} holds. In the remaining sections we
specialize to the Main Assumptions. First in Section \ref{sec:Concavity-of-the}
we deduce the concavity of the upper microcanonical entropy $S_{+}^{(N)}(e)$
(Theorem \ref{thm:conv for N for two-point etc}) from a complex analog
of the Brunn-Minkowski inequality. Then, letting $N\rightarrow\infty$
the concavity of the entropy $S(e)$ in the high energy region $\{e>e_{0}\}$
(Theorem \ref{thm:concave S in text}) is deduced. In the following
Section \ref{sec:global con} this is shown to yield global concavity
of $S(e)$ when the Main Assumptions are complemented with weak positive
definiteness and some examples are exhibited. In Section \ref{sec:Critical-inverse-temperatures}
applications to slope formulas of \emph{critical inverse temperatures
}are given and some connections to algebraic geometry are explained.
In the final Section \ref{sec:strict} a strict concavity result for
$S(e)$ is deduced under the Homogeneous Assumptions from a uniqueness
result for free energy minimizers, established in the companion paper
\cite{be}.

\section{\label{sec:Preliminaries-and-notation}Setup and preliminaries }

\subsection{\label{subsec:Very-general-setup}Very General Setup and notation}

A very general formulation of the setup that we shall consider, henceforth
called the Very General Setup may be formulated as follows. Let $X$
be a topological space endowed with a probability measure $\mu_{0}$
and $E(\mu)$ a lsc functional $E(\mu)$ on the space $\mathcal{P}(X)$
of all probability measures on $X.$ We then define the corresponding
 entropy $S(e)$ and free energy $F(\beta)$ as in formula \ref{eq:def of S e}
and formula \ref{eq:def of F beta of F beta mu intro}, respectively.
Occasionally, when specializing to the General Setup introduced in
Section \ref{subsec:General-setup intro} the notation $E_{W,V}(\mu)$
will designate an energy functional $E(\mu)$ of the particular form
\ref{eq:def of E intro}. 

We set 
\[
e_{min}:=\inf_{\mathcal{P}_{0}(X)}E(\mu),\,\,\,e_{0}:=E(\mu_{0}),\,\,\,e_{max}:=\sup_{\mathcal{P}_{0}(X)}E(\mu)
\]
(recall that $\mathcal{P}_{0}(X)$ denotes the space of all probability
measures on $X$ with compact support). 

We will mainly consider the case when $X\Subset\R^{2n}$ and the Main
Assumptions (or the Homogeneous Assumptions) introduced in Section
\ref{subsec:Summary-of-the} hold. These assumptions will be recalled
in Section \ref{subsec:The-Main-and}, but we first provide some preliminaries
on concavity and plurisubharmonicity.

\subsection{Concave preliminaries}

We we will be discussing concavity properties of the entropy $S(e)$
we provide some general preliminaries on concave functions. First
recall that a function $\phi$ on a convex subset $C$ of $\R^{d}$
taking values in $]-\infty,\infty]$ is said to be \emph{convex} on
$C$ if for any given two points $x_{0}$ and $x_{1}$ and $t\in]0,1[$
\[
\phi(tx_{0}+(1-t)x_{1})\leq t\phi(x_{0})+(1-t)\phi(x_{1})
\]
 and\emph{ strictly convex} on $C$ if the inequality above is strict
for any $t\in]0,1[.$ A function $f$ on $C$ is \emph{(strictly)
concave} if $-f$ is (strictly) convex. Here we will be mainly concerned
with the case when $d=1.$ In this case, if $f$ is concave and finite
on a closed interval $C\subset\R$, but not strictly convex, then
there exist two points $x_{0}$ and $x_{1}$ in $C$ such that $f$
is affine on $[x_{0},x_{1}].$ In Sections \ref{sec:Thermodynamical-equivalence},
\ref{sec:global con} we will use some standard properties of convex
functions recalled below, translated into the setup of concave functions
(for further background see \cite{r} and \cite[Section 2.1.3]{v}).
If $\phi$ is a convex function on $\R^{d}$ then its\emph{ subdifferential}
$(\partial\phi)$ at a point $x_{0}\in\R^{d}$ is defined as the convex
set
\begin{equation}
(\partial\phi)(x_{0}):=\left\{ y_{0}:\,\phi(x_{0})+y_{0}\cdot(x-x_{0})\leq\phi(x)\,\,\,\forall x\in\R^{d}\right\} \label{eq:sub gradient prop}
\end{equation}
In particular, if $\phi(x_{0})=\infty,$ then $(\partial\phi)(x_{0})$
is empty. Similarly, if $f$ is concave on $\R^{d}$ then its \emph{superdifferential}
$(\partial f)(x_{0})$ is defined as above, but reversing the inequality.
In other words, $(\partial f)(x_{0}):=$ $-(\partial(-f)(x_{0}).$
In the case when $f$ is concave on $\R$ and finite in a neighborhood
of $x_{0}$ 
\[
(\partial f)(x)=[f'(x+),f'(x-)],
\]
 where $f'(x+)$ and $f'(x-)$ denote the right and left derivatives
of $f$ at $x,$ respectively. In particular, $f$ is differentiable
at $x$ iff $(\partial f)(x)$ consists of a single point. If $f$
is a function on $\R^{d}$ taking values in $[-\infty,\infty]$ its
(concave) \emph{Legendre-Fenchel transform }is the usc and concave
function on $\R^{d}$ (taking values in $[-\infty,\infty[$ ) defined
by

\[
f^{*}(y):=\inf_{x\in\R}\left(x\cdot y-f(x)\right).
\]
It follows readily from the definitions that 
\begin{equation}
y\in\partial f(x)\iff x\in\partial f^{*}(y).\label{eq:y in gradient iff x in gradient}
\end{equation}
Moreover, it is well-known that 
\begin{equation}
\overline{\partial f(\{f>-\infty\})}=\overline{\partial f^{*}(\{f^{*}>-\infty\}).}\label{eq:gradient im}
\end{equation}
Note that, in general, $f^{**}$ is the concave envelope of $f:$
\begin{equation}
(f^{**})(x)=\inf_{a\,\text{affine}}\left\{ a(x):\,\,\,a\geq f\right\} =\inf_{g\,\text{concave},\text{finite}}\left\{ g(x):\,\,\,g\geq f\right\} .\label{eq:f star start as concave envelope}
\end{equation}
 Indeed, the first equality follows directly from the definition and
the second one is shown by, for a fixed $x,$ taking $a(x)$ to be
any affine function coinciding with $g$ at $x$ and with gradient
in $\partial g(x).$ 

We will also make use of the following lemmas (which are without doubt
essentially well-known, but for completeness proofs are provided in
the appendix). 
\begin{lem}
\label{lem:diff implies dual strc conc}Let $f$ be a concave function
on $\R$ and assume that $f$ is differentiable in a neighborhood
of $[x_{0},x_{1}].$ Then $f^{*}$ is strictly concave in the interior
of $[y_{0},y_{1}]:=[f'(x_{1}),f'(x_{0})].$ 
\end{lem}

Note that, in general, $f^{**}\geq f.$ Concerning the strict inequality
we have the following
\begin{lem}
\label{lem:convex envol affine}Let $f$ be a function on $\R$ such
that $\sup_{\R}f<\infty$ and $U\Subset\R$ an open set where $f$
is finite and usc. Then $\{f^{**}>f\}\cap U$ is open in $U$ and
$f^{**}$ is affine on $\{f^{**}>f\}\cap U.$ 
\end{lem}

\subsection{\label{subsec:Background-on-plurisubharmonicit}Background on plurisubharmonicity
and the class $PSH_{\boldsymbol{a}}$}

The Main Assumptions introduced in Section \ref{subsec:Summary-of-the}
involve the notion of plurisubharmonicity. While this notion is central
in the fields of several complex variables and complex geometry, it
may not be familiar to readers lacking background in these fields.
We thus recall the main definitions and properties that we shall use
and refer to \cite[Section 5.A.]{dem} for further background. We
will identify $\R^{2n}$ with $\C^{n}$ in the standard way. A function
$\psi$ on $\C^{n}$ is said to \emph{plurisubharmonic} (\emph{psh},
for short) if $\psi$ is upper semi-continuous (usc) taking values
in $[-\infty,\infty[$ and subharmonic along complex lines, i.e. if
$\zeta\mapsto\psi(z_{0}+\zeta a_{0})$ is a local subharmonic function
on $\C$ for any given $z_{0},a_{0}\in\C^{n},$ or equivalently that
\[
\psi(z_{0})\leq\frac{1}{2\pi}\int\psi(z_{0}+e^{i\theta}a_{0})d\theta.
\]
In particular, $\psi$ is then subharmonic on $\R^{2n}.$ If $\psi$
is smooth then it is psh iff the complex Hessian $\partial\bar{\partial}\psi$
of $\psi$ is a semi-positive Hermitian matrix at any $z:$
\[
\partial\bar{\partial}\psi(z):=(\frac{\partial^{2}\psi(z)}{\partial z_{i}\partial\bar{z}_{j}})\geq0,\,\,\,\frac{\partial}{\partial z_{i}}:=\frac{1}{2}\frac{\partial}{\partial x_{i}}-\frac{i}{2}\frac{\partial}{\partial y_{i}}
\]
Equivalently, a function $\psi$ is psh if, locally, it can be expressed
as a decreasing limit of smooth psh functions $\psi_{j}.$ In fact,
$\psi_{j}$ may be taken as a convolution of $\psi$ with any suitably
scaled smooth probability density with compact support. If $-u$ is
plurisubharmonic, then $u$ is called \emph{plurisuperharmonic}. An
open set $\Omega$ in $\C^{n}$ is said to be \emph{pseudoconvex}
if $\Omega$ admits a continuous psh exhaustion function $\rho$ i.e.
$\rho$ is psh on $\Omega$ and such that $\{\rho\leq C\}$ is a compact
subset of $\Omega.$ We recall the following essentially standard
lemma (see the appendix for a proof):
\begin{lem}
\label{lem:pseudo}Let $\phi$ be a psh function on a pseudoconvex
open set $\Omega.$ Then $\{\phi<0\}\cap\Omega$ is also pseudoconvex. 
\end{lem}

We also recall that the following standard facts \cite[Theorem 5.5.]{dem},
which allows one to construct a range of different types of psh functions:
\begin{lem}
If $\psi_{1},...,\psi_{r}$ are psh functions and $\chi(t_{1},...,t_{r})$
is a convex function on $\R^{r}$ which is increasing in each $t_{i},$
then $\chi(\psi_{1},...,\psi_{r})$ is psh. In particular, if $\alpha_{1},...,\alpha_{r}$
are non-negative functions, then
\[
\sum_{i=1}^{r}\alpha_{i}\psi_{i},\,\,\,\log\sum_{i=1}^{r}e^{\alpha_{i}\psi_{i}}\text{ and \ensuremath{\max\{\psi_{1},...,\psi_{r}\}} }
\]
are psh functions.
\end{lem}

In particular, if $\psi$ is psh and $\chi$ is a convex increasing
function on $\R,$ then the composed function $\chi(\phi)$ is psh.
Since $|f(z)|^{2}$ is psh for any holomorphic function $f(z)$ on
$\C^{n}$ (as follows, for example directly from the characterization)
it follows form the previous lemma that 
\[
\psi(z):=\log(\sum_{i=1}^{r}|f_{i}(z)|^{2})
\]
 is psh for any given holomorphic functions $f_{1},...,f_{r}.$ In
particular, $\log|z|^{2}$ is psh. Moreover, if a function $\psi$
only depends on  the absolute values of $z_{i},$ then $\psi(z)$
is psh iff it is convex with respect $(\log|z_{1}|,...,\log|z_{n}|)\in\R^{n}.$ 

\subsubsection{\label{subsec:The-class-}The class $PSH_{\boldsymbol{a}}$ }

Given $\boldsymbol{a}=(a_{1},..,a_{m})\in]0,\infty[^{n}$ we denote
by $\mathcal{V}_{\boldsymbol{a}}$ the vector field on $\C^{n}$ defined
by 
\begin{equation}
\mathcal{V}_{\boldsymbol{a}}:=\sum_{i=1}^{m}a_{i}\frac{\partial}{\partial\theta_{i}},\label{eq:def of vector field}
\end{equation}
 where $\frac{\partial}{\partial\theta_{i}}$ denotes the generator
of the $S^{1}-$action on $\C^{n}$ which rotates the $z_{i}-$coordinate
and leaves the other coordinates invariant (i.e. $e^{i\theta}\cdot z:=(z_{1},...,e^{i\theta}z_{i},...,z_{n})$).
In other words, $\mathcal{V}_{a}$ is the Hamiltonian vector field
corresponding to the Hamiltonian 
\begin{equation}
h_{a}(z):=\sum_{i=1}^{m}\frac{1}{2a_{i}}|z_{i}|^{2}\label{eq:def of Hamiltonian for oscillator}
\end{equation}
 on $\R^{2n},$ endowed with its standard symplectic form. Note that
the Hamiltonian $h_{a}$ is plurisubharmonic on $\C^{n}$ (since $|z_{i}|^{2}$
is). Now if $U$ is an open connected subset of $\C^{n}$ then the
class $PSH_{\boldsymbol{a}}(U)$ is defined as the class of all psh
functions $\psi$ on $U,$ not identically $-\infty,$ such that $\mathcal{V}_{\boldsymbol{a}}(\psi)=0.$
More generally, if $X$ is closed connected subset of $\C^{n}$ we
denote by $PSH_{\boldsymbol{a}}(X)$ the class of all functions $\psi$
such that $\psi$ is in $PSH_{\boldsymbol{a}}(U)$ for some open subset
$U$ containing $X$ (depending on $\psi$). 
\begin{example}
\label{exa:quasi-homo}\emph{(The ``algebraic and quasi-homogeneous''
case).} If $P(z_{1},...,z_{n})$ is a\emph{ }quasi-homogeneous polynomial,
i.e. there exists exist positive integer weights $a_{1},..a_{n}$
such that $P$ is homogeneous of degree $d$ wrt the corresponding
$\mathbb{R}_{+}$-action: 
\begin{equation}
P(c^{a_{1}}z_{1},...,\lambda^{a_{n}}z_{n})=c^{d}F(z_{1},...,z_{n})\label{eq:hypersurf}
\end{equation}
for any $c\in\mathbb{R}_{+},$ then $\log|P(z)|$ is in $PSH_{\boldsymbol{a}}(\C^{n}).$
More generally, if $P_{j}$ are polynomials on $\C^{n}$ which are
quasi-homogeneous of degree $d_{j}$ for the same weighs $a_{1},...,a_{n}$
and $\alpha_{i}>0,$ then 
\begin{equation}
\psi(z):=\log(\sum_{j=1}^{r}|P_{j}(z)|^{\alpha_{j}})\in PSH_{\boldsymbol{a}}(\C^{n})\label{eq:alg quasi-hom case}
\end{equation}
In the particular case when all $\alpha_{i}=1$ and $d_{i}=d$ we
call $d$ the degree of $\psi.$ 
\end{example}

By composing the previous examples with convex increasing functions
$\chi$ on $\R$ one may fabricate an abundance of examples of functions
in the class $PSH_{\boldsymbol{a}}(\C^{n}).$ For example, $\sum_{j=1}^{M}|P_{j}(z)|^{\alpha_{j}}$
is in $PSH_{\boldsymbol{a}}$ if $P_{j}(z)$ is a homogeneous polynomial
(wrt $\boldsymbol{a}$) and $\alpha_{j}>0.$ 

\section{\label{sec:Thermodynamical-equivalence}General properties of $S(e)$
and thermodynamic equivalence of ensembles}

In this section general properties of the entropy $S(e)$ are studied
and the notion of thermodynamic equivalence of ensembles introduced
in \cite{e-h-t} is recalled. The main new feature in this section,
as compared to the setup in \cite{e-h-t}, is that $E(\mu)$ is not
assumed to be continuous. This leads to some subtle aspects that do
not seem to have been addressed before. Throughout the section we
will consider the Very General Setup introduced in Section \ref{subsec:Very-general-setup}.

\subsection{Monotonicity of $S(e)$ }

The following lemma generalizes \cite[Prop 2.2]{clmp2} (with a similar
proof) and involves the following ad hoc property: 
\begin{defn}
Assume that $X$ is compact. Then a functional $E(\mu)$ on $\mathcal{P}(X)$
has the \emph{affine continuity property }if for any $\mu_{1}\in\mathcal{P}(X)$
such that $E(\mu_{1})<\infty$ and $S(\mu_{1})>-\infty$ the function
$t\mapsto E(\mu_{0}(1-t)+t\mu_{1})$ is continuous on $[0,1].$ For
a general $X$ the affine continuity property is said to hold if it
holds for all compact subsets of $X.$
\end{defn}

\begin{lem}
\label{lem:decreasing}(monotonicity of $S(e))$. Assume that $X$
is compact and $e_{0}:=E(\mu_{0})<\infty.$
\begin{itemize}
\item If $E(\mu)$ is convex on $\mathcal{P}(X),$ then $S(e)$ is increasing
for $e\leq e_{0}$ and strictly increasing in the subinterval where
$S(e)>-\infty.$ In particular, 
\[
S_{-}(e):=\sup_{E(\mu)\leq e}S(\mu)
\]
\item If $E(\mu)$ has the affine continuity property, then $S(e)$ is decreasing
for $e\geq e_{0}$ and strictly decreasing in the subinterval where
$S(e)>-\infty.$ In particular, 
\[
S_{+}(e):=\sup_{E(\mu)\geq e}S(\mu)
\]
More precisely, in the second point there is no need to assume that
$E$ is lsc on $\mathcal{P}(X)$ and thus it also follows that $S(e)$
is increasing for $e\leq E(\mu_{0}).$
\end{itemize}
\end{lem}

\begin{proof}
To prove the first point first observe that, since $E$ is lsc and
$X$ is compact $\{E(\mu)\leq e\}$ is compact (or empty). We may
assume that $S(\mu)$ is not identically equal to $-\infty$ on $\{E(\mu)\leq e\}$(otherwise
we are done). Since $S(\mu)$ is usc the sup of $S(\mu)$ on the set
$\{E(\mu)\leq e\}$ is thus attained at some $\mu_{1}$ in the set.
Assume in order to get a contradiction that $E(\mu_{1})<e.$ Consider
the affine segment $\mu_{t}$ in $\mathcal{P}(X)$ connecting $\mu_{0}$
and $\mu_{1};$ $\mu_{t}:=\mu_{0}(1-t)+t\mu_{1}$ for $t\in[0,1].$
By the assumed convexity of $E(\mu)$ 
\[
E(\mu_{t})\leq(1-t)E(\mu_{0})+tE(\mu_{1})<e
\]
 for $t$ sufficiently small, using that $E(\mu_{0})<\infty.$ But,
as is well-known, $S(\mu)$ is strictly concave on $\{S(\mu)>-\infty\}\subset\mathcal{P}(X)$
and attains its maximum at $\mu_{0}$ and hence $S(\mu_{t})<S(\mu_{1})$
for any $t\in[0,1[$ (as follows from Jensen's inequality). This contradicts
the assumption that $\mu_{1}$ is a maximizer and hence it must be
that $E(\mu_{1})=e,$ as desired.

To prove the second point it will be enough to show that for any $\mu_{1}\in\mathcal{P}(X)$
such that $E(\mu_{1})\geq e$ and $S(\mu_{1})>-\infty$ there exists
$\mu\in\mathcal{P}(X)$ such that $E(\mu)=e$ and $S(\mu)\geq S(\mu_{1}).$
To this it will, in the light of the previous argument, be enough
to show that there exists some $t\in[0,1]$ such that $E(\mu_{t})=e.$
But, by assumption $E(\mu_{0})\leq e$ and $E(\mu_{1})\geq e.$ We
can thus conclude by invoking the assumption that $E(\mu_{t})$ is
continuous. Since we have not used that $E$ is lsc on $\mathcal{P}(X)$
the same argument applies to $-E,$ which proves the last statement
of the lemma. 
\end{proof}

\subsection{\label{subsec:Thermodynamical-equivalence}Thermodynamic equivalence
of ensembles}

In this section we consider the Very General Setup. It follows readily
from the definitions that the Legendre-Fenchel transform $S^{*}$
of $S$ coincides with the free energy $F(\beta):$
\[
S^{*}=F.
\]
 Following \cite{e-h-t,t-e-t} we make the following
\begin{defn}
\emph{thermodynamic equivalence of ensembles} is said to hold \emph{globally}
if 
\[
S=F^{*}
\]
 and thermodynamic equivalence of ensembles is said to hold at\emph{
an energy level $e$} if $S(e)>-\infty$ and
\[
S(e)=F^{*}(e).
\]
\end{defn}

Recall that, in general, a function $S(e)$ is usc and concave iff
$S^{**}=S.$ It was shown in \cite[Prop 3.1a]{e-h-t} that $S$ is
always usc under the assumption that $X$ is compact and $E(\mu)$
is continuous wrt the weak topology on $\mathcal{P}(X)$ (this is
the case if $W$ and $V$ are continuous). In this case global thermodynamic
equivalence thus holds iff $S$ is concave. But here we need consider
the case when the continuity assumptions are not satisfied (and moreover
$X$ may be non-compact). We will impose the following compatibility
property between $\mu_{0}$ and $E(\mu)$.
\begin{defn}
A measure $\mu_{0}$ in $X$ is said to has the \emph{Energy Approximation
Property} if for any compactly supported probability measure $\mu$
there exists a sequence $\mu_{j}\in\mathcal{P}(X),$ supported in
the same compact set, converging weakly towards $\mu$ with the following
properties:
\end{defn}

\begin{itemize}
\item $\mu_{j}$ is absolutely continuous with respect to $\mu_{0}$
\item $\lim_{j\rightarrow\infty}E(\mu_{j})=E(\mu)$ 
\end{itemize}
\begin{rem}
This property was introduced in the context of large deviation theory
in\cite{c-g-z} and studied from a potential-theoretic point of view
in \cite{be1} (see the discussion in the end of Section \ref{subsec:The-necessity-of}).
\end{rem}

The energy approximation property ensures that $S(e)$ is finite on
$]e_{min},e_{max}[:$
\begin{lem}
\label{lem:energy appr implies S finite etc}Assume that $\mu_{0}$
has the energy approximation property and the affine continuity property
on compact subspaces of $X$. Then $S(e)$ is finite on $]e_{min},e_{max}[.$
\end{lem}

\begin{proof}
By Lemma \ref{lem:decreasing} we just have to verify the claim that
there exists some $\mu\in\mathcal{P}(X)_{0}$ such that $E(\mu)\leq e$
and $S(\mu)>-\infty.$ To this end take $\delta>0$ such that $e-\delta>e_{min}.$
By the verify definition of $e_{min}$ there exists $\mu$ such that
$E(\mu)\leq e-\delta.$ Moreover, by the monotone convergence theorem
$\mu$ may be chosen to have compact support. Now take a sequence
$\mu_{j}(=\rho_{j}\mu_{0})$ converging weakly towards $\mu$ with
the energy approximation property. Replacing $\rho_{j}$ with $\max\{\rho_{j},R\}/\int\{\rho_{j},R\}\mu_{0}$
for a given $R>0$ and using a diagonal argument we may as well assume
that $\rho_{j}\in L^{\infty}.$ In particular, 
\[
E(\mu_{j})\leq e,\,\,\,S(\mu_{j})>-\infty
\]
for $j$ sufficiently large, proving the claim when $e\in]e_{min},e_{0}[.$
A similar approximation argument applies if instead $e\in]e_{0},e_{max})[$
(again using Lemma \ref{lem:decreasing}). Finally, if $e=E(\mu_{0})$
then $S(\mu)\geq S(\mu_{0})=0,$ which concludes the proof of the
claim above.
\end{proof}
\begin{prop}
\label{prop:thermo equiv under energy appr}In the Very General Setup
the following holds: 
\begin{itemize}
\item If the  entropy $S(e)$ is concave on $]e_{min},e_{max}[$ and $\mu_{0}$
has the energy approximation property and the affine continuity property,
then $S(e)$ is continuous on $]e_{min},e_{max}[$ and thermodynamic
equivalence of ensembles holds for any $e\in]e_{min},e_{max}[.$ 
\item If the  entropy $S(e)$ is concave and continuous on $[e_{0},e_{max}[$
then thermodynamic equivalence of ensembles holds for any $e\in[e_{0},e_{max}[$
and moreover for any\emph{ }$e\in[e_{0},e_{max}[$ 
\begin{equation}
S(e)=\inf_{\beta\leq0}\left(\beta e-F(\beta)\right)\label{eq:S of e as inf over beta neg}
\end{equation}
\item If the entropy $S(e)$ is concave and continuous on $]e_{\min},e_{0}]$
then thermodynamic equivalence of ensembles holds for any $e\in]e_{\min},e_{0}]$
\end{itemize}
\end{prop}

\begin{proof}
In order to show that $S(e_{1})=S^{**}(e_{1})$ at a given point $e_{1}$
in $]e_{max},e_{min}[$ it is enough to find an affine function $s$
on $\R$ such that $s\geq S$ and $s(e_{1})=S(e_{1})$ (by formula
\ref{eq:f star start as concave envelope}). But since $s$ is concave
and finite on $]e_{min},e_{max}[$ its superdifferential $\partial S$
is non-empty, i.e. contains some $\beta\in\R.$ This means that the
affine function 
\begin{equation}
s(e):=\beta(e-e_{1})+S(e_{1})\label{eq:def of affine func}
\end{equation}
 coincides with $S$ at $e$ and has the property that $s\geq S$
on $]e_{min},e_{max}[.$ Hence, by Lemma \ref{lem:decreasing}, $s\geq S$
on all of $\R,$ which proves the first point. 

To prove the second point in the proposition fix $e_{1}\in]E(\mu_{0}),e_{max}[.$
By formula \ref{eq:f star start as concave envelope} it will be enough
to find an affine function $s$ on $\R$ such that $s\geq S$ and
$s(e_{1})=S(e_{1}).$ To this end first define the function $f(e)$
to be equal to $S(e)$ on $[e_{0},e_{max}[$ and $e_{0}$ when $e<e_{0}.$
Thus $f(e)=\max\{e_{0},S(e)\}$ is continuous and convex on $]-\infty,e_{max}[.$We
then obtain the desired affine function $s$ by picking an element
$\beta$ in the superdifferential $\partial f$ of $f$ at $e_{1}$
and again defining $s(e)$ by formula \ref{eq:def of affine func}.
Finally, to prove the last formula we have to show that the infimum
in formula \ref{eq:S of e as inf over beta neg} is attained for some
$\beta\leq0.$ But this follows from the fact that, in the previous
step, $\beta$ in formula \ref{eq:def of affine func} is non-positive,
since $f$ is decreasing (by Lemma \ref{lem:decreasing}). The third
point is shown in essentially the same way as the second one.
\end{proof}
\begin{rem}
If $e_{max}<\infty,$ then it could happen that $S(e_{max})\neq S^{**}(e_{max})$
in the first point of the previous proposition.  Also note that in
the case when $E(\mu)$ is of the form $E=E_{W,V}$ (as in formula\ref{eq:def of E intro})
then $e_{max}=\infty$ holds if either there exists $x_{0}$ such
that $V(x_{0})=\infty$ or $(x_{0},y_{0})$ such that $W(x_{0},y_{0})=\infty.$
Indeed, then $E(\mu)=\infty$ for $\mu=\delta_{x_{0}}/2+\delta_{x_{1}}/2.$
\end{rem}

As shown in Theorem \ref{thm:convex cont iff energy appr} below the
energy approximation property is not merely a technical assumption,
but essential.

\subsection{\label{subsec:Priors-versus-linear}Priors versus linear constraints}

Now consider the Very General Setup in the case when $X$ is a domain
in $\R^{d}$ and $\mu_{0}=dx.$ Given a continuous function $\psi_{0}$
and $\lambda\in\R$ we may then replace $\mu_{0}$ with the prior
defined by the probability measure 
\[
\mu_{\lambda}:=e^{-\lambda\psi_{0}}dx/Z_{\lambda},\,\,\,Z_{\lambda}:=\int_{X}e^{-\lambda\psi_{0}}dx,
\]
 assuming that $Z_{\lambda}<\infty.$ The corresponding corresponding
entropy function $S_{\mu_{\lambda}}(e)$ is closely related to the
multi-variable entropy function $S(e,l)$ on $\R^{2}$ defined by
\[
S(e,l):=\sup_{\mu\in\mathcal{P}(X)_{0}}\left\{ S(\mu):\,\,E(\mu)=e,\,\,\,L(\mu)=l\right\} ,\,\,\,L(\mu):=\int_{X}\psi_{0}\mu,
\]
obtained by imposing the linear constraint $L(\mu)=l$ (where $S(\mu)$
denotes the entropy of $\mu$ relative to $dx).$ Indeed, it follows
readily from the definition that, for a fixed $e,$ the Legendre-Fenchel
transform of the function $\lambda\mapsto S_{\mu_{\lambda}}(e)$ is
given by $-S(e,l)-\log Z_{\lambda}.$ Hence, under the hypothesis
that\emph{ $S(e,l)$ is concave and lower-semicontinuous wrt $l,$}
inverting the Legendre-Fenchel transform gives
\[
S(e,l)=\inf_{\lambda}\left(S_{\mu_{\lambda}}(e)+\lambda l+\log Z_{\lambda}\right).
\]
 As a consequence, if $S_{\mu_{\lambda}}(e)$ is globally concave
with respect to $e,$ for any fixed $\lambda$ such that $Z_{\lambda}$
is finite, then $S(e,l)$ is globally concave on $\R^{2}.$ Multi-variable
entropy functions are studied in \cite{e-h-t}, from the point of
view of equivalence of ensembles, but here we will focus on one-variable
entropy functions defined with respect to appropriate priors. Note
that in the non-compact case when $X=\R^{d}$ the inclusion of a function
$\psi_{0}$ with sufficient growth at infinity is crucial in order
to get a prior measure with finite total mass. In the presence of
rotational symmetry the standard choice of a prior is a centered Gaussian
measure. 
\begin{rem}
More generally, given $r$ functions $\psi_{1},...,\psi_{r}$ on $\R^{d}$
and $\lambda_{1},...,\lambda_{r}\in\R^{d}$ one can consider the prior
$\mu_{\boldsymbol{\lambda}}=e^{-\sum\lambda_{i}\psi_{i}}/Z_{\boldsymbol{\lambda}}$
and the corresponding entropy function $S(e,\boldsymbol{l})$ on $\R^{1+d}.$
Then the previous considerations still apply if $\lambda l$ is replaced
by the scalar product between $\boldsymbol{\lambda}$ and $\boldsymbol{l}.$ 
\end{rem}

\section{\label{sec:Macrostate-equivalence-of}Macrostate equivalence of ensembles
and existence of maximum entropy measures}

An important motivation for the notion of thermodynamic equivalence
of ensembles is that it implies that any maximum entropy measure $\mu^{e}$
(representing an equilibrium macrostate in the microcanonical ensemble)
minimizes the free energy $F_{\beta}(\mu)$ at an inverse temperature
$\beta$ corresponding to the energy level $e.$ This is made precise
by the following result (essentially contained in \cite{e-h-t}).
\begin{lem}
\label{lemma:(macrostate-equivalence).-Assume}(macrostate equivalence
of ensembles). Consider the Very General Setup. Assume that $S^{**}(e)=S(e)>-\infty$
and assume that $\partial S(e)$ is non-empty (this is the case if,
for example, $S^{**}=S>-\infty$ in a neighborhood of $e$ ). If $\mu^{e}$
is a maximal entropy measure with energy $e,$ i.e. $S(\mu^{e})=S(e),$
then $\mu^{e}$ minimizes the free energy functional $F_{\beta}(\mu)$
for any $\beta\in\partial S(e).$ 
\end{lem}

\begin{proof}
By assumption $S(e)>-\infty.$ Hence, the assumption that $\beta\in\partial S(e)$
means that $\beta\in(\partial F^{*})(e).$ Since $F=(F^{*})^{*}$
it follows from the definition of $\partial F^{*}$ that 
\[
F(\beta)=-F^{*}(e)+\beta e
\]
(since $0\in\partial(-F^{*}(e)+\beta e)$). In other words, 
\[
\inf_{\mu\in\mathcal{P}(X)}F_{\beta}(\mu)=-S(\mu^{e})+\beta E(\mu^{e}),
\]
 which means that $\mu^{e}$ minimizes $F_{\beta}(\mu),$ as desired.
\end{proof}
\begin{rem}
Without the property that $S(e)=S^{**}(e)$ a maximal entropy measure
$\mu^{e}$ will, in general, not minimize $F_{\beta}(\mu).$ This
is discussed in the context of BEG-model in the final section of \cite{e-t-t}
(where it is pointed out that $\mu^{e}$ may be merely a local minimizer
of $F_{\beta}(\mu)$ or even a saddle point). Moreover, even if $S(e)=S^{**}(e)$
there may, in general, exists minimizers of $F_{\beta}(\mu),$for
$\beta\in\partial S(e),$ which are not maximum entropy measures (at
energy $e),$ unless $S(e)$ is strictly concave at $e$ (see \cite{e-h-t}).

As shown in \cite{e-h-t}, the existence of $\mu^{e}$ is automatic
for any $e\in]e_{0},e_{max}[,$ when $X$ is compact and $E(\mu)$
is a continuous functional on $\mathcal{P}(X).$ However, since we
do not impose these assumptions in the Main Assumptions we next provide
some general existence result for $\mu^{e},$ that will be applied
to the Main Assumptions in Section \ref{subsec:Existence-of-maximum Main As}. 
\end{rem}

\subsection{\label{subsec:Existence-when compact}Existence of $\mu^{e}$ when
$X$ is compact}

We start with the low-energy region:
\begin{prop}
Consider the Very General Setup. Assume that $X$ is compact and that
the energy approximation property and the affine continuity property
holds. Then, for any $e\in]e_{min},e_{0}]$ there exists a maximum
entropy measure $\mu^{e}.$
\end{prop}

\begin{proof}
Fix $e\in]e_{min},e_{0}].$ First recall that by Lemma \ref{lem:energy appr implies S finite etc}
$S(e)$ is finite. Next, by Lemma \ref{lem:decreasing} (and its proof)
it is enough to prove that the functional $S(\mu)$ admits a maximizer
on $\{E(\mu)\leq e\}.$ But since $E$ is lsc, $\{E(\mu)\leq e\}$
is closed in the compact space $\mathcal{P}(X),$ hence compact. The
existence of $\mu^{e}$ thus follows from the upper-semicontinuity
of $S(\mu)$ on $\mathcal{P}(X).$ 
\end{proof}
In order to ensure the existence of maximum entropy measures in the
high-energy region we introduce the following stability property:
\begin{defn}
In the Very General Setup the \emph{thermal stability property} is
said to hold if there exists $\beta_{0}<0$ such that 
\[
\inf_{\mathcal{P}(X)}\left(\beta_{0}E-S\right)>-\infty.
\]
In other words, this property says that the critical inverse temperature
$\beta_{c}$ (discussed in Section \ref{sec:Critical-inverse-temperatures})
is strictly negative. Turning to the General Setup we will use the
following result, shown in the course of the proof of \cite[Lemma 2.13, formula 2.12]{be0}):
\end{defn}

\begin{lem}
\label{lem:(Energy/Entropy-compactness).-As}Consider the General
Setup and assume that $X$ is compact. If the thermal stability property
holds, then the functional $E_{V,W}$ is continuous on $\{\mu:\,S(\mu)\geq-C\}\Subset\mathcal{P}(X)$
for any given constant $C>0.$ 
\end{lem}

The following result generalizes the existence result in \cite{clmp2},
concerning the case when $W(x,y)$ has a logarithmic singularity along
the diagonal:
\begin{prop}
\label{prop:macro eq}Consider the General Setup. Assume that $X$
is compact and that the energy approximation property and the thermal
stability property hold. Then $S(e)$ is usc on $]e_{min},e_{max}[$
and for any $e$ in $]e_{min},e_{max}[$ there exists a maximum entropy
measure $\mu^{e}.$
\end{prop}

\begin{proof}
Take $e_{j}\rightarrow e\in]e_{min},e_{max}[$ and let $\mu_{j}$
be a sequence in $\mathcal{P}(X)$ such that $E(\mu_{j})=e_{j}$ and
$S(\mu_{j})\geq s(e_{j})-1/j.$ In particular, there exists a constant
$C$ such that $S(\mu_{j})\geq-C.$ By the previous lemma, we may,
after perhaps passing to a subsequence, assume that $\mu_{j}\rightarrow\mu_{\infty}$
in $\mathcal{P}(X)$ and $E(\mu_{j})\rightarrow E(\mu_{\infty}).$
Hence, $E(\mu_{\infty})=e$ and since $S$ is usc on $\mathcal{P}(X)$
$S(\mu_{\infty})\geq\limsup_{j\rightarrow\infty}S(\mu_{j}).$ This
shows that $S(e)\geq S(\mu_{\infty})\geq\limsup_{j\rightarrow\infty}S(e_{j}),$
i.e. that $S$ is usc. Similarly, the \emph{existence} of $\mu^{e}$
also follows from  the previous lemma, since it shows that $\{E(\mu)=e\}\cap S(\mu)\geq-C$
is closed (and thus $S$ attains its maximum value there for $C$
sufficiently large). 
\end{proof}
If the thermal stability property does not hold, then there may not
be no maximum entropy measures,where $S(e)$ is globally concave.
In fact, we have the following converse to the previous proposition
when $S(e)$ is concave and continuous on $[e_{0},e_{max}[.$
\begin{prop}
Consider the Very General Setup and assume that $X$ is compact and
that there exists a maximum entropy measure $\mu^{e}$ for some $e\in]e_{0},e_{max}[.$
Then the thermal stability property holds.
\end{prop}

\begin{proof}
The assumed concavity of $S(e)$ implies that the right derivative
of $S(e)$ tends to $\beta_{c}$ as $e\rightarrow e_{max}$ (see Cor
\ref{cor:slope formula for beta c} below and its proof). Hence, if
we assume that the thermal stability property does not hold, i.e.
that $\beta_{c}=0$ it follows, since $S(e)$ attains its maximum
at $e$ and is assumed continuous and concave on $[e_{0},e_{max}]$
that $S(e)\equiv S(e_{0}).$ But $S(\mu)=S(\mu_{0})$ iff $\mu=\mu_{0}$
(which implies $E(\mu)=e_{0})$ and hence there exists no maximum
entropy measure $\mu^{e}$ when $e>e_{0}.$
\end{proof}
The previous proposition is illustrated by the case of singular power-laws
in Section \ref{subsec:Non-concavity-in-the cata}. Before turning
to the non-compact case we point out that the following concrete bound
implies the thermal stability property (see Lemma \ref{lem:general bounds on the partition function}
below): 
\begin{equation}
\sup_{x\in X}\int e^{-\beta_{0}\left(\frac{1}{2}W(x,y)+V(y)\right)}\mu_{0}(y)<\infty,\,\,\,\int_{X}e^{-\beta_{0}V}\mu_{0}<\infty,\label{eq:sup exp integral finite}
\end{equation}
 for some $\beta_{0}<0,$which will turn out to be satisfied if the
Main Assumptions are complemented with the assumption that $W$ is
translationally invariant, up to a bounded term.

\subsection{\label{subsec:Existence-when-X is noncompact}Existence of $\mu^{e}$
when $X$ is non-compact}

In order to discuss maximum entropy measures in the case when $X$
is non-compact we first need to replace  the space $\mathcal{P}(X)_{0}$
of all probability measures with compact support, appearing in the
definition\ref{eq:def of S e} of $S(e)$, with probability measures
satisfying an appropriate growth assumption ``at infinity''. Indeed,
if for example, $E=E_{V}$ for a lsc function $V$ which is unbounded
both from above \emph{and} from below (say, $V(x)=-\log|x|$ in $\R^{d}),$
then it is not a priori clear how to define $E_{V}(\mu)$ if $\mu$
have unbounded support. To handle this issue we will make the following
\emph{growth assumption:} exists a continuous non-negative function
$\phi_{0}$ of $X$ such that 

\begin{equation}
-W(x,y)-\frac{1}{2}V(x)-\frac{1}{2}V(y)\leq\frac{1}{2}\phi_{0}(x)+\frac{1}{2}\phi_{0}(y)+C_{0}.\label{eq:growth assumption on W V}
\end{equation}
 Then we can decompose 
\begin{equation}
E(\mu)=E_{\phi_{0}}(\mu)-\int\mu\phi_{0},\,\,\,E_{\phi_{0}}(\mu):=\int\left(W(x,y)+\frac{1}{2}V(x)+\frac{1}{2}V(y)+\frac{1}{2}\phi_{0}(x)+\frac{1}{2}\phi_{0}(y)\right)\mu\otimes\mu\label{eq:decomp of E}
\end{equation}
 where the first term has a well-defined value in $]-\infty,\infty]$,
since the corresponding integrand is bounded from below. This means
that if we replace $\mathcal{P}(X)$ with the subspace 
\[
\mathcal{P}_{\phi_{0}}(X):=\left\{ \mu\in\mathcal{P}(X):\,\int_{X}\phi_{0}\mu<\infty\right\} 
\]
then $S(e)$ may be expressed as 
\begin{equation}
S(e):=\sup_{\mu\in\mathcal{P}_{\phi_{0}}(X)}\left\{ S(\mu):\,\,E(\mu)=e\right\} ,\label{eq:def of S e with phi not}
\end{equation}
 where $E(\mu)$ is defined by formula\ref{eq:decomp of E}. According
to the following result the existence of a maximizer $\mu^{e}$ is
guaranteed if $\phi_{0}$ has slower growth then an an appropriate
exhaustion function $\psi_{0}$ of $X$ (i.e. the sub-level sets $\{\psi_{0}\leq R\}$
are compact and exhaust $X$ when $R\rightarrow\infty$):
\begin{prop}
\label{prop:existence non-compact}Consider the General Setup and
assume that there exists a continuous exhaustion function $\psi_{0}$
of $X$ such that the following growth-properties hold:
\begin{itemize}
\item $\int e^{\delta\psi_{0}}\mu_{0}<\infty$ for some $\delta>0$
\item The growth-assumption \ref{eq:growth assumption on W V} holds for
a $\phi_{0}$ such that $\phi_{0}/\psi_{0}\rightarrow0$ uniformly
as $\psi_{0}\rightarrow\infty$ (e.g. for $\phi_{0}=\psi_{0}^{(1-\epsilon)}$
for some $\epsilon\in]0,1[$).
\end{itemize}
If the thermal stability property holds (i.e. $\beta_{c}<0),$ then
there exists a measure $\mu^{e}$ realizing the sup in formula \ref{eq:def of S e with phi not}
for any given $e\in[e_{0},e_{max}[.$ 
\end{prop}

\begin{proof}
Setting $\tilde{W}(x,y):=W(x,y)+\frac{1}{2}V(x)+\frac{1}{2}V(y)-\frac{1}{2}\phi_{0}(x)+\frac{1}{2}\phi_{0}(y)$
we can express $E_{\phi_{0}}(\mu)=\int\tilde{W}(x,y)\mu\otimes\mu.$
Now fix $e\in[e_{0},e[$ and recall that $S(e)$ is finite. Since,
by assumption, $\tilde{W}(x,y)$ is lsc on $X\times X$ and bounded
from below it extends to a lsc function on $\tilde{X}\times\tilde{X},$
where $\tilde{X}$ denotes the one-point compactification of $X.$
Moreover, we identify $\psi_{0}$ with a lsc function on $\tilde{X},$
taking the value $\infty$ at the point at infinity and $\mu_{0}$
with a probability measure on $\tilde{X},$ not charging the point
at infinity. Accordingly, we can identify $E_{\phi_{0}}(\mu)$ and
$S(\mu)$ with functionals on $\mathcal{P}(\tilde{X}).$ Denote by
$\tilde{S}(e)$ the corresponding entropy function. Since $\int_{\tilde{X}}\mu\psi_{0}<\infty$
implies that $\mu$ does not charge the point at infinity it will,
in order to prove the proposition, be enough to show that the sup
defining $\tilde{S}(e)$ is attained. To this end take a sequence
$\mu_{j}\in\mathcal{P}(X)$ such that $E(\mu_{j})=e$ and $S(\mu_{j})$
increases to $\tilde{S}(e).$ Decompose $\mu=e^{-\delta\Psi_{0}}\mu_{\delta}$
for $\delta>0$ such that $\mu_{\delta}:=e^{\delta\Psi_{0}}\mu_{0}$
has finite total mass. Then there exists a constant $C$ such that
\begin{equation}
S(\mu_{j})=S_{\mu_{\delta}}(\mu)-\delta\int\Psi_{0}\mu_{j}\geq-C.\label{eq:lower bound S mu j}
\end{equation}
Since $S_{\mu_{\delta}}(\mu)$ is uniformly bounded from above on
$\mathcal{P}(X)$ (using that $\mu_{\delta}$ has total finite mass)
this means that there exists a finite constant $C_{\delta}$ such
that
\begin{equation}
\int\psi_{0}\mu_{j}\leq C_{\delta}<\infty.\label{eq:upper bound integral Psi not}
\end{equation}
Now, since $\tilde{X}$ compact we may, after perhaps passing to a
subsequence, assume that $\mu_{j}\rightarrow\mu_{\infty}$ weakly
in $\mathcal{P}(\tilde{X})$ for some $\mu_{\infty}$ (which, by the
bound \ref{eq:upper bound integral Psi not}, does not charge the
point at infinity). Moreover, combining the bound \ref{eq:upper bound integral Psi not}
with the growth-assumption on the continuous function $\phi_{0}$
gives (using Markov's inequality) that 
\[
\lim_{j\rightarrow\infty}\int\phi_{0}\mu_{j}=\int\phi_{0}\mu_{\infty}.
\]
Since $S(\mu)$ is usc on $\mathcal{P}(\tilde{X})$ all that remains
is to verify that 
\begin{equation}
\lim_{j\rightarrow\infty}E_{\phi_{0}}(\mu_{j})=E_{\phi_{0}}(\mu_{\infty})\label{eq:conv of weighted energy}
\end{equation}
To this end we rewrite the assumed thermal stability property as 
\begin{equation}
\beta_{0}E_{\phi_{0}}(\mu)-\beta_{0}\int\phi_{0}\mu-S(\mu)\geq-C_{0},\,\,\,\beta_{0}<0\label{eq:thermal stab in pf}
\end{equation}
 Note that 
\begin{equation}
-\beta_{0}\int\phi_{0}\mu-S(\mu)=-S_{\mu_{\beta_{0}}}(\mu),\,\,\,\,\mu_{\beta_{0}}:=e^{\beta_{0}\phi_{0}}\mu_{0},\label{eq:decomp of entropy}
\end{equation}
 where the measure $\mu_{\beta_{0}}$ has finite mass (since $\beta_{0}\leq0$
and $\phi\geq0)$ and thus identifies with a measure on $\tilde{X}.$
Accordingly, can view \ref{eq:thermal stab in pf} as an inequality
on $\mathcal{P}(\tilde{X}),$ saying that lsc functional $E_{\phi_{0}}(\mu)$
has the thermal stability property wrt the measure $\mu_{\beta_{0}}$
on the compact space $\tilde{X}.$ Thus, it follows from Lemma \ref{lem:(Energy/Entropy-compactness).-As}
that $E_{\phi_{0}}$ is continuous on $\{S_{\mu_{\beta_{0}}}(\mu)\geq-C\}.$
Finally, combining \ref{eq:decomp of entropy}, \ref{eq:upper bound integral Psi not}
and \ref{eq:lower bound S mu j} reveals that $S_{\mu_{\beta_{0}}}(\mu_{j})\geq-C$
for some constant $C$ and hence the desired convergence \ref{eq:conv of weighted energy}
follows. 
\end{proof}
\begin{rem}
To see that the growth-properties in the previous proposition are
essential consider the case when $X=\R^{d},$ $\mu_{0}=e^{-|x|}dx$
and $V(x)=-|x|^{p}$ for $p>0.$ Then the thermal stability property
does hold (in fact, $\beta_{c}=-\infty,$ since $Z_{\beta}:=\int e^{-\beta V}\mu_{0}<\infty$
for any $\beta<0).$ Moreover, $\int e^{\delta\psi_{0}}\mu_{0}<\infty$
for $\psi_{0}:=|x|^{2}.$ However, for $e\leq e_{0}$ a maximum entropy
measure $\mu^{e}$ only exists under the assumption that $p<2,$ i.e.
precisely when $-V/\psi_{0}\rightarrow\infty$ (indeed, if $\mu^{e}$
exists, then $\mu^{e}=e^{-\beta V}/\int e^{-\beta V}\mu_{0}$ for
some $\beta>0$ (see Section \ref{subsec:The-case-of E affine}). 
\end{rem}

\section{\label{sec:Concavity-of-S low-energy}Concavity of $S(e)$ in the
low-energy region for convex $E(\mu)$}

\subsection{\label{subsec:Concavity-and monot}Concavity and monotonicity of
$S(e)$ in the low-energy region $e\protect\leq e_{0}$ when $E(\mu)$
is convex.}

We now consider the entropy $S(e)$ in the low-energy region $e\leq e_{0}$
under the assumption that $E(\mu)$ is convex. By way of motivation
we start with the case when $E(\mu)$ is affine. 

\subsubsection{\label{subsec:The-case-of E affine}The case of $E(\mu)$ affine}

In the case when $E(\mu)$ is affine on $\mathcal{P}(X)$ it follows
directly from the definition of $S(e)$ that $S(e)$ is globally concave,
using the concavity of $S(\mu)$ on $\mathcal{P}(X).$ Moreover, if
$X$ is compact and $E(\mu)=\left\langle V,\mu\right\rangle $ for
$V\in C^{0}(X),$ then a duality argument reveals that $S(e)$ is
finite and strictly concave on $]e_{min},e_{max}[.$ In fact,
\[
S(e)=F_{V}^{*}(e),\,\,F_{V}(\beta)=-\log\int_{X}e^{-\beta V}\mu_{0},
\]
 where $F_{V}(\beta)<\infty$ for all $\beta,$ since $X$ is compact
and $V$ is bounded. Indeed, in this case it follows from Jensen's
inequality that the free energy $F(\beta)$ is of the form $F_{V}(\beta)$
above \footnote{This is an instance of the classical Gibb's variational principle}.
Since $F_{V}(\beta)$ is differentiable on all of $\R$ and its derivative
tends to $\inf_{X}V(=e_{min})$ and $\sup_{X}V(=e_{max})$ as $\beta\rightarrow\infty$
and $\beta\rightarrow-\infty,$ respectively, it thus follows from
Lemma \ref{lem:f diff implies S strict concav in gener} below that
$S_{V}(e)$ is strictly concave on $]e_{min},e_{max}[.$ However,
if $X$ is non-compact, then the strict concavity of $S_{V}(e)$ may
fail as illustrated by the following simple example: 
\[
X=\R,\,\,\,\mu_{0}=e^{-|x|}dx,\,\,\,V(x)=|x|^{2}.
\]
 In this case $E(\mu_{0})<\infty,$ but $\int e^{-\beta V}\mu_{0}<\infty$
iff $\beta\geq0.$ It follows that $S(e)=S(e_{0})=0$ for $e>e_{0}$
and thus $S(e)$ is not strictly concave. Indeed, applying the second
point in Prop \ref{prop:thermo equiv under energy appr}, we get,
for $e\geq e_{0}$
\[
S(e)=\inf_{\beta\leq0}\left(\beta e-F_{V}(\beta)\right).
\]
 But, since $F_{V}(\beta)=\infty$ for $\beta<0$ the rhs above is
attained at $\beta=0,$ showing that $S(e)=0.$ Also note that replacing
$V$ with $-V$ yields an example where $S(e)$ fails to be strictly
concave in the low-energy region. Note also that in this example,
the sup defining $S(e)$ is not attained in the region where $S(e)=S(e_{0}),$
if $e\neq e_{0}.$ Indeed, if the sup is attained at $\mu^{e}$ satisfying
$E(\mu)=e,$ then $S(\mu^{e})=S(\mu_{0})$ and hence $\mu^{e}=\mu_{0},$
which forces $e=E(\mu_{0}):=e_{0}.$ 

\subsubsection{The case of $E(\mu)$ convex}

Using Lemma \ref{lem:decreasing} we observe that similar arguments
apply in the low-energy region when $E(\mu)$ is convex, under some
further regularity assumptions: 
\begin{prop}
\label{prop:E convex plus energy approx implies S concave}Let $X$
be a topological space and $E(\mu)$ a lsc convex functional on $\mathcal{P}(X)$
and $\mu_{0}\in\mathcal{P}(X).$ 
\begin{itemize}
\item If $X$ is compact and $e_{0}:=E(\mu_{0})<\infty,$ then $S(e)$ is
concave on $]-\infty,e_{0}].$
\item If $X$ is $\sigma-$compact (i.e. a countable union of compact space)
and $E(1_{K}\mu_{0})<\infty$ for any compact subspace $K$ of $X,$
then, if the energy approximation property holds, $S(e)$ is concave,
increasing and finite (hence continuous) on $]e_{min},e_{0}[.$
\end{itemize}
\end{prop}

\begin{proof}
Given $e_{1}$ and $e_{2}$ in $]-\infty,e_{0}]$ and $t\in[0,1]$
set $e_{t}:=(1-t)e_{0}+te_{1}.$ Let $\mu_{1}$ and $\mu_{2}$ be
contenders for the sup defining $S(e_{1})$ and $S(e_{2}),$ respectively.
Set $\mu_{t}:=(1-t)\mu_{1}+t\mu_{2}.$ Since $E(\mu)$ is assumed
convex, $E(\mu_{t})\leq e_{t}.$ Hence, if $X$ is compact and $E(\mu_{0})<\infty,$
then Lemma \ref{lem:decreasing} gives, $S(e_{t})\geq S(\mu_{t})\geq(1-t)S(\mu_{1})+tS(\mu_{2}),$
using that $S$ is concave on $\mathcal{P}(X).$ This proves the first
point. To prove the second one we write $X$ is an increasing union
of compact subspaces $X_{R}.$ Denoting by $S_{R}$ the  entropy corresponding
to $X_{R}$ it follows directly from the definition that $S_{R}(e)\leq S(e).$
Now, by the energy approximation property in Lemma \ref{lem:energy appr implies S finite etc},
$-\infty<S_{R}(e)\leq S(e).$ A slight variant of the argument in
the end of the proof of Theorem \ref{thm:concave S in text} then
shows that $S_{R}(e)$ increases towards $S(e)$ as $R\rightarrow\infty.$
Hence, we can conclude by invoking the first point.
\end{proof}
Next, a different duality argument yields \emph{strict} concavity
and continuity up to $e=e_{0}$ when $X$ is compact. The proof uses
the following duality criterion:
\begin{lem}
\label{lem:f diff implies S strict concav in gener}Consider the Very
general setup and assume that $X$ is compact and that the energy
approximation property holds. If $F(\beta)$ is differentiable in
a neighborhood of $[\beta_{0},\beta_{1}]$ and $[F'(\beta_{1}),F'(\beta_{0})]\subset]e_{min},e_{max}[,$
then $S(e)$ is strictly concave and equal to $F^{*}$ on $[F'(\beta_{1}),F'(\beta_{0})].$
Moreover, in general, if $F$ is differentiable at $\beta,$ then
$F'(\beta)=E(\mu_{\beta})$ for any minimizer of $F_{\beta}.$ 
\end{lem}

\begin{proof}
Since $F(\beta)$ is concave and $F=S^{*}$ Lemma \ref{lem:diff implies dual strc conc}
implies that $S^{**}$ is strictly concave on $[F'(\beta_{1}),F'(\beta_{0})].$
Next, by Prop \ref{lemma:(macrostate-equivalence).-Assume} $S$ is
usc on $U:=]e_{min},e_{max}[$ and hence Lemma \ref{lem:convex envol affine}
forces $S^{**}=S$ on $[F'(\beta_{1}),F'(\beta_{0})],$ which concludes
the proof of the first statement. The last statement follows directly
from letting $\delta$ tend to zero (from left and from right) in
the inequality
\begin{equation}
F(\beta+\delta)-F(\beta)\leq F_{\beta+\delta}(\mu_{\beta\text{ }})-F_{\beta}(\mu_{\beta})=\delta E(\mu_{\beta}).\label{eq:minim ineq}
\end{equation}
\end{proof}
\begin{prop}
\label{prop:E convex plus energy approx implies S str conc}Assume
that $X$ is compact, $E(\mu)$ is lsc and convex on $\mathcal{P}(X).$
Then $S(e)$ is strictly concave and $S(e)=F^{*}(e)$ on $]e_{min},e_{0}[.$
Moreover, $S(e)$ is continuous on $]e_{min},e_{0}].$ 
\end{prop}

\begin{proof}
The concavity was shown in \cite{e-s} under the extra assumption
that $E(\mu)$ be continuous on $\mathcal{P}(X).$ Here we note that
an alternative argument yields\emph{ strict} concavity under the more
general assumptions in the proposition. The starting point is the
observation that $F_{\beta}(\mu)$ is convex on $\mathcal{P}(X)$
for $\beta\geq0$ and strictly convex on $\{F_{\beta}<\infty\}.$
Indeed, since $E(\mu)$ is assumed convex this follows directly from
the corresponding property of $-S(\mu)$ (i.e. from the case $\beta=0),$
which is well-known \cite{d-z}. It then follows from general principles
that $F(\beta)$ is differentiable with derivative at $\beta$ given
by $e(\beta):=E(\mu_{\beta}),$ where $\mu_{\beta}$ is the unique
minimizer of $F_{\beta}.$ Indeed, this follows from the general statement
in the appendix of \cite{b-b0}, using that $E(\mu_{\beta})$ is continuous
in $\beta$ by the argument below. Hence, by Lemma \ref{lem:f diff implies S strict concav in gener}
$S(e)$ is strictly concave and equal to $F^{*}$ on the interval
$]\lim_{\beta\rightarrow\infty}e(\beta),\lim_{\beta\rightarrow0}e(\beta)[.$
By the concavity of $F(\beta)$ the function $e(\beta)$ is decreasing.
Moreover, the energy approximation property implies, in a rather straight-forward
manner, that 
\[
\lim_{\beta\rightarrow0}e(\beta)=e_{min}
\]
 (see \cite{be1}). All that remains is thus to verify that 
\[
\lim_{\beta\rightarrow0}e(\beta)=e_{0}.
\]
 But since $e(\beta)$ is decreasing this follows readily from the
lower-semi continuity of $E(\mu)$ (see \cite{be1}). To prove that
that $S(e)$ is continuous on $]e_{min},e_{0}]$ it will be enough,
by the previous step, to show that $F^{*}(e)$ is continuous on $]e_{min},e_{0}]$
and $F^{*}(e_{0})=0.$ Since $F^{*}$ is concave it is enough to show
that $F^{*}(e)$ is finite on $]e_{min},e_{max}[.$ But
\[
S\leq S^{**}=F^{*}\leq0,
\]
 where the last inequality follows from restricting the inf defining
$F^{*}$ to $\beta=0.$ Since $S$ is finite (by the previous proposition)
it follows that is $F^{*}$ is also finite and thus continuous on
$]e_{min},e_{max}[.$ Hence, by the continuity of $F^{*}$ at $e_{0}$
we get $S(e)\rightarrow F^{*}(e_{0})$ as $e\rightarrow e_{0}.$ But
\[
F^{*}(e_{0})=\inf_{\beta\in\R}\left(\beta e_{0}-F(\beta)\right),\,\,\,F(\beta)=\inf_{\mu\in\mathcal{P}(X)}\beta E(\mu)-S(\mu)\leq\beta E(\mu_{0})-S(\mu_{0})=\beta e_{0}
\]
Hence, 
\[
F^{*}(e_{0})=\inf_{\beta\in\R}\left(\beta e_{0}-F(\beta)\right)\geq\inf_{\beta\in\R}\left(\beta e_{0}-\beta e_{0}\right)=0.
\]
which gives $\liminf_{e\rightarrow e_{0}}S(e)\geq0.$ Since, trivially,
$S(e)\leq S(e_{0})=0$ it follows that $S(e)\rightarrow0=S(e_{0}),$
as desired.
\end{proof}

\subsection{\label{subsec:The-necessity-of}The necessity of the energy approximation
property for thermodynamic equivalence of ensembles}

We next show that the assumption that $\mu_{0}$ has the energy approximation
property, used in the previous section is necessary for having thermodynamic
equivalence of ensembles:
\begin{thm}
\label{thm:convex cont iff energy appr}Let $X$ be a compact topological
space endowed with a measure $\mu_{0}$ such that $E(\mu_{0})<\infty$
and assume that $E(\mu)$ is a lsc convex functional on $\mathcal{P}(X)$
and $V\in C^{0}(X).$ Denote by $S_{V}(e)$  entropy $S_{V}(e)$ associated
to $E_{V}(\mu):=E(\mu)+\left\langle V,\mu\right\rangle $ and the
measure $\mu_{0}.$ Then $S_{V}(e)$ is concave and finite on $]e_{min},e_{0}]$
for any $V\in C^{0}(X)$ iff $\mu_{0}$ has the energy approximation
property. In other words, thermodynamic equivalence of ensembles holds
in the low-energy regions $]e_{min},e_{0}]$ for all $V\in C^{0}(X)$
iff $\mu_{0}$ has the energy approximation property.
\end{thm}

\begin{proof}
First assume that $\mu_{0}$ has the energy approximation property.
Since $E_{V}(X)$ is lsc and convex it then follows from the previous
proposition that $S_{V}(e)$ is concave on $]e_{min},e_{0}[$ for
any $V\in C^{0}(X).$ To prove the converse first note that, by the
third point in Prop \ref{prop:thermo equiv under energy appr}, the
restriction of $S_{V}$ to $]e_{min},e_{0}]$ is equal to the Legendre-Fenchel
transform of $F_{V}(\beta).$ Hence, since $S_{V}(e)$ is assumed
finite on $]e_{min},e_{0}[$ it follows from the property of gradient
images in formula \ref{eq:gradient im} that $dF_{V}(\beta)/d\beta\rightarrow e_{min}$
as $\beta\rightarrow\infty$ (using either left or right derivatives).
Since $F_{V}(\beta)$ is concave this means that 
\[
\lim_{\beta\rightarrow\infty}F_{V}(\beta)/\beta=\inf_{\mathcal{P}(X)}E(\mu).
\]
 Now, by definition, $F_{V}(\mu)/\beta=E(\mu)-S(\mu)/\beta$ and hence
\[
\lim_{\beta\rightarrow\infty}\inf_{\mathcal{P}(X)}\left(E(\mu)-S(\mu)/\beta\right)=\inf_{\mathcal{P}(X)}E(\mu).
\]
But as as shown in \cite{be1} the latter convergence holds for all
$V\in C^{0}(X)$ iff $\mu_{0}$ has the approximation property (briefly,
the point is that the convergence in question is, since $E(\mu)$
is convex equivalent to the $\Gamma-$convergence of $E(\mu)-S(\mu)/\beta$
towards $E(\mu),$ which, in turn, is equivalent to the energy approximation
property of $\mu_{0}).$
\end{proof}
In the case when $W(x,y)$ is the repulsive logarithmic interaction
in $\R^{2}$ or $W(x,y)=|x-y|^{-s}$ in $\R^{n}$ for $s\in]d-2,d[$
(specializing to the Coulomb interaction when $s=d-2)$ a potential-theoretic
characterization of measures $\mu_{0}$ satisfying the energy approximation
property was given in \cite{be1}. In particular, it was shown that
any compact domain $X$ with smooth boundary admits probability measures
$\mu_{0}$ with support $X$ and a density in $L^{1}(X,dx),$ for
which the energy-approximation property fails. Hence, by the previous
theorem thermodynamic equivalence of ensembles also fails. On the
other hand, Lebesgue measure on a compact domain $X$ has the energy
approximation property, if $X$ is non-thin at all boundary points,
in the sense of classical potential theory. For example, this is the
case if any point $x\in\partial X$ is the vertex of a cone contained
in $X$(e.g. if $X$ is a Lipschitz domain). 

\subsection{\label{subsec:Non-concavity-in-the cata}The catastrophic case of
singular power-laws}

Consider now the case when $X$ is compact and $W(x,y)$ is a repulsive
power-law singularity:
\begin{equation}
W(x,y):=|x-y|^{-\alpha}+H(x,y),\,\,\,\alpha>0\label{eq:power-law sing}
\end{equation}
 for $H$ continuous on $X\times X.$ In particular, $e_{\infty}=\infty.$
We will say that a compact set $X$ is\emph{ strictly star-shaped}
if for any point $x\in X$ and $c\in[0,1[$ the scaled point $cx$
is contained in the interior of $X.$
\begin{prop}
\label{prop:cata}Consider a repulsive power-law singularity $W$
on a compact strictly star-shaped subset $X$ of $\R^{d}$ and let
$\mu_{0}$ be proportional (or comparable) to Lebesgue measure on
$X.$ Then $S(e)$ is concave on $\R$ and finite (hence continuous)
on $]e_{min},\infty[.$ Moreover, $S(e)=S(e_{0})$ for for any $e\geq e_{0}$
and, as a consequence, there exists no maximum entropy measure $\mu^{e}$
when $e>e_{0}.$ 
\end{prop}

\begin{proof}
For simplicity we will assume that $H=0,$ but the general case is
shown in essentially the same way. First observe that $E$ has has
the energy approximation property. Indeed, using that $X$ is assumed
strictly star-shaped and $W(e^{t})$ is monotone in $t$ it is, by
the argument in the proof of Lemma \ref{lem:superharmonic gives energy approx},
enough to show this when the support of $\mu$ is contained in the
interior of $X.$ Let $\mu_{\epsilon}$ be defined as in formula \ref{eq:def of mu eps}.
First using that $E$ is convex and then that $E$ is translationally
invariant gives 
\[
E(\mu_{\epsilon})\geq\int_{a\in B_{\epsilon}}\sigma_{\epsilon}E\left((T_{a})_{*}\mu\right)=\int_{a\in B_{\epsilon}}\sigma_{\epsilon}E\left(\mu\right)=E(\mu).
\]
 The reversed asymptotic inequality follows directly from the lower
semi-continuity of $E,$ resulting from the assumed lower semi-continuity
of $w.$ Next note that the affine continuity property appearing in
Lemma \ref{lem:decreasing} holds, as is seen by modifying the proof
of Lemma \ref{lem:Main Ass implies affine}. Indeed, using the Cauchy-Schwartz
inequality the finiteness in formula \ref{eq:finiteness in proof lemma affine cont}
follows from the positive definiteness of $W(x,y)$ and that $W\in L^{1}(X^{2}).$
Hence, by Prop \ref{prop:E convex plus energy approx implies S concave}
$S(e)$ is concave and continuous on $]e_{min},e_{0}]$ and $S(e)>-\infty$
for all $e\in]e_{min},\infty[.$ 

Next, we will show that $S(e)=S(e_{0})$ for for any $e>e_{0}.$ To
this end it will, thanks to Lemma \ref{lem:decreasing} , be enough
to show that there exists a family $\nu_{\epsilon}\in\mathcal{P}(X)$
parametrized by $\epsilon>0$ such that, as $\epsilon\rightarrow0,$
\begin{equation}
(i)\,E(\nu_{\epsilon})\rightarrow\infty,\,\,\,(ii)\,S(\nu_{\epsilon})\rightarrow S(\mu_{0})=0.\label{eq:E and S along nu eps}
\end{equation}
 We will take $\nu_{\epsilon}=\epsilon^{\alpha/4}(T_{\epsilon})_{*}\mu_{0}+(1-\epsilon^{\alpha/4})\mu_{0},$
where, as before, $T_{\epsilon}$ denotes the scaling map $x\mapsto\epsilon x.$
First observe that since $S((T_{\epsilon})_{*}\mu_{0})=d\log\epsilon$
for some (as seen by making the change of variables $x\mapsto T_{\epsilon}(x)$
in the integrals) we get, using the concavity of $S(\mu)$ on $\mathcal{P}(X),$
\[
S(\nu_{\epsilon})\geq\epsilon^{\alpha/4}S((T_{\epsilon})_{*}\mu_{0})+(1-\epsilon^{\alpha/4})S(\mu_{0})\geq\epsilon^{\alpha/4}d\log\epsilon+0,
\]
 which verifies the second item in formula \ref{eq:E and S along nu eps}.
To prove the first one observe that, making the change of variables
$x\mapsto T_{\epsilon}(x)$ in the integrals, reveals that $E((T_{\epsilon})_{*}\mu_{0})=\epsilon^{-\alpha}E(\mu_{0})$
and hence $E(\nu_{\epsilon})=\epsilon^{\alpha/2}E((T_{\epsilon})_{*}\mu_{0})=\epsilon^{\alpha/2}\epsilon^{-\alpha}E(\mu_{0}),$
which proves the first item in formula \ref{eq:E and S along nu eps}.
Hence, $S(e)=S(e_{0})$ for any $e>e_{0}.$ Since we have shown that
$S(e)$ is concave, increasing and continuous on $]e_{min},e_{0}]$
and $S(e)=S(e_{0})$ it follows that $S(e)$ if concave and continuous
on $]e_{min},e_{0}].$ Moreover, since $S(\mu)=S(\mu_{0})$ iff $\mu=\mu_{0}$
(which implies $E(\mu)=e_{0})$ it follows that there exists no maximum
entropy measure for $e>e_{0}.$
\end{proof}
More precisely, the proof of the previous proposition reveals that,
for any given $e>e_{0}$ there exists $\mu\in\mathcal{P}(X)$ with
energy $e,$ i.e. $E(\mu)=e,$ whose entropy $S(\mu)$ can be taken
to be arbitrarily close to the maximal entropy and such that $\mu$
has a ``core-halo'' structure, i.e. $\mu$ is a convex combination
\[
(1-\lambda)\mu_{0}+\lambda\mu_{1}
\]
for some $\lambda\in]0,1[,$ where $\mu_{1}$ (the ``core'') can
be taken to be a uniform measure of arbitrarily large density on a
ball with arbitrarily small radius, centered at a given point in the
interior of $X.$ 
\begin{rem}
The assumption the $X$ be star-shaped was imposed to ensure the energy
approx property and can certainly be relaxed. For example, if $W(x,y)$
is the Coulomb interaction in $\R^{n},$ then, as pointed out in Section
\ref{subsec:The-necessity-of}, the energy approximation property
in question holds if the interior of $X$ is non-thin at all boundary
points.
\end{rem}

The previous proposition also applies to the corresponding singular
power-laws obtained by switching the sign of $W,$ if at the same
time $e$ is replaced by $-e$ (using that $S(e)$ is concave iff
$S(-e)$ is). In the case of the Newtonian pair-interaction in $\R^{3}$
the non-existence of the corresponding maximum entropy measure is
closely related to the \emph{gravitational catastrophe} (Antonov instability)
which plays a central role in astrophysics \cite[Section 4.10.1]{b-t}
(as explained in \cite{be2}).

\section{\label{sec:Concavity-of-the}Concavity in the high energy region
under the Main Assumptions}

We start by recalling the Main/Homogeneous Assumptions stated in the
introduction of the paper. 

\subsection{\label{subsec:The-Main-and}The Main and Homogeneous Assumptions}

Let $X$ is a (possible non-compact) subset of $\R^{2n}$ end let
$\phi$ be a defining function for $X,$ i.e. a continuous function
such that 
\[
X=\{\phi\leq0\}.
\]
 Endow $X$ with a measure $\mu_{0}$ which is absolutely continuous
wrt Lebesgue measure $d\lambda:$ 
\[
\mu_{0}=e^{-\Psi_{0}}d\lambda.
\]
 on $\R^{2n}.$ As pointed out above we will identify $\R^{2n}$ with
$\C^{n}$ and denote by $(z_{1},...,z_{n})$ the standard holomorphic
coordinates on $\C^{n}.$ 

$\medskip$

\paragraph*{\emph{\noun{Main Assumptions: }}$\phi\in PSH_{\boldsymbol{a}}(\C^{n}),$
$\Psi_{0},-V\in PSH_{\boldsymbol{a}}(X)$ and $-W\in PSH_{\boldsymbol{a,a}}(X\times X)$
for some $\boldsymbol{a}\in]0,\infty[^{n}$}

$\medskip$

The class $PSH_{\boldsymbol{a}}(X)$ was defined in Section\ref{subsec:Background-on-plurisubharmonicit}
and the class $PSH_{\boldsymbol{a,a}}(X\times X)$ is defined similarly,
by identifying $\C^{n}\times\C^{n}$ with $\C^{2n}$ and using the
weight vector $(\boldsymbol{a},\boldsymbol{a}).$ Recall that we also
introduced the Homogeneous Assumptions in Section \ref{subsec:Summary-of-the},
which according to the following lemma is a special case of the Main
Assumptions:
\begin{lem}
\label{lem:special}If the Homogeneous Assumptions are satisfied,
then so are the Main Assumptions. 
\end{lem}

\begin{proof}
First note that $-v(r)$ is increasing in $r.$ Indeed, since $\phi(t):=v(e^{t})$
is convex in $t$ the limit, denoted by $\dot{\phi}(-\infty),$ of
the one sided derivative $\phi'(t+)$ exists as $t\rightarrow\infty.$
Since $\phi(t)$ is assumed bounded from above as $t\rightarrow-\infty$
it follows that $\dot{\phi}(-\infty)\geq0.$ Hence, by convexity,
$\phi'(t+)\geq0$ for all $t,$ showing that $\phi(t)$ is increasing
in $t,$ as desired. Since $\log|z|$ is psh this means that $V(z)$
is an increasing convex function of the psh function $\log|z|$ when
$z\neq0$ and bounded from above in a punctured neighborhood of the
origin in $\C^{n}.$ But any psh function which is locally bounded
from a above on the complement of a pluripolar set $A$ (i.e. a set
which is locally the $-\infty-$set of a psh function) extends over
$A$ to a unique psh function \cite[Thm 5.24]{dem}. Thus $-V$ indeed
defines a psh function on $X$ and the same argument applies to $\Psi_{0}(z).$
Similarly since $(z,\zeta)\mapsto(z-\zeta)$ is holomorphic the function
$\log|z-\zeta|$ is psh on $\C^{2n}$ and thus $-W(z,\zeta)$ is an
increasing convex function of a psh function when $\log|z-\zeta|\neq-\infty$
and thus psh. All in all this means that the Main Assumptions are
satisfied with, for example, $a_{0}=...=a_{n}=1.$ 
\end{proof}
We next show that the ``affine continuity property'' and energy
approximation property introduced in Section \ref{sec:Thermodynamical-equivalence}
both hold under the Main Assumptions. 
\begin{lem}
\label{lem:Main Ass implies affine}Under the Main Assumptions the
affine continuity property holds.
\end{lem}

\begin{proof}
Since $X$ may be assumed compact and $W$ is lsc we may after perhaps
replacing $W$ with $W+C$ i.e. $E$ with $E+C,$ as well assume that
$W\geq0$ on $X\times X.$ Hence, by the dominated convergence theorem
it is enough to verify that if $E(\mu)<\infty,$ then 
\begin{equation}
\int_{X\times X}W\mu\otimes\mu_{0}<\infty.\label{eq:finiteness in proof lemma affine cont}
\end{equation}
 Set $u_{\mu}(x):=\int_{X}W(x,y)\mu(y).$ Since $-W$ is psh on a
neighborhood of $X\times X$ the function $-u_{\mu}(x)$ psh on $X.$
Now since $X$ is connected, as shown in the course of the proof of
Theorem \ref{thm:conv for N for two-point etc}, any psh function
(or more generally, subharmonic function) is either identically equal
to $-\infty$ or in $L_{loc}^{1}$ (as follows from the submean property
of subharmonic functions). But, by assumption, $\int_{X}u_{\mu}\mu=E(\mu)<\infty$
and hence $-u_{\mu}$ cannot be identically $-\infty.$ Since $\mu_{0}=e^{-\Psi_{0}}d\lambda$
\ref{eq:finiteness in proof lemma affine cont} thus follows directly
in the case when $\Psi_{0}$ is bounded. In the general case we can
use that by Cor \ref{cor:integrability threshold as slope} below,
there exists $q>1$ such that $\int_{X}e^{-q\psi}d\lambda<\infty$
for any psh function $\psi$ (not identically $-\infty)$ and apply
Hölder's inequality to conclude .
\end{proof}
\begin{lem}
\label{lem:superharmonic gives energy approx}Assume that the Main
Assumptions hold. Then the corresponding energy approximation property
is satisfied. 
\end{lem}

\begin{proof}
First consider the case when the support of $\mu$ is contained in
the interior of $X.$ Set
\begin{equation}
\mu_{\epsilon}:=\int_{a\in B_{\epsilon}}\sigma_{\epsilon}(T_{a})_{*}\mu,\label{eq:def of mu eps}
\end{equation}
 where, for a given $a\in\R^{d},$ $T_{a}$ is the map $x\mapsto x+a$
and $B_{\epsilon}$ denotes the ball of radius $\epsilon$ centered
at the origin. For $\epsilon$ sufficiently small $\mu_{\epsilon}$
is also supported in $X.$ It is a standard fact that $\mu_{\epsilon}$
is absolutely continuous wrt Lebesgue measure and $\mu_{\epsilon}\rightarrow\mu$
weakly as $\epsilon\rightarrow0.$ Moreover, 
\begin{equation}
\lim_{\epsilon\rightarrow0}E(\mu_{\epsilon})=E(\mu).\label{eq:conv of energy epsilon}
\end{equation}
 Indeed, setting $\Psi:=-W(x,y)+V(x)+V(y)$ and changing the order
of integration gives 
\[
-E(\mu_{\epsilon})=\int_{B_{\epsilon}\times B_{\epsilon}}\sigma_{\epsilon}\otimes\sigma_{\epsilon}\int_{X\times X}\Psi(T_{a})_{*}\mu\otimes(T_{b})_{*}\mu=\int_{X\times X}\mu(x)\otimes\mu(y)\int_{B_{\epsilon}\times B_{\epsilon}}\sigma_{\epsilon}\otimes\sigma_{\epsilon}\Psi(x+a,x+b).
\]
Recall that, in general, if $\psi(x)$ is a subharmonic function,
then $\int_{B_{\epsilon}}\sigma_{\epsilon}\psi(x+a)$ decreases to
$\psi(x),$ as $\epsilon$ decreases to $0.$ Hence, the convergence
\ref{eq:conv of energy epsilon} follows from the monotone convergence
theorem. Finally, for a general $\mu\in\mathcal{P}(X)$ we consider
for $\tau\in\C$ the holomorphic action $\tau\circledcirc z$ defined
by formula \ref{eq:t ring z}. If $z\in X$ and $\Re\tau<0,$ then
it follows readily from the definitions that $\tau\circledcirc z$
is contained in the interior of $X$ (compare the proof of Theorem
\ref{thm:conv for N for two-point etc}). Hence, fixing $t<0$ and
setting $F_{t}(z):=t\circledcirc z$ the probability measure $\mu^{t}:=(F_{t})_{*}\mu$
is supported in the interior of $X.$ Moreover, $\mu^{t}$ converges
weakly towards $\mu$ when $t\rightarrow0$ and 
\begin{equation}
\lim_{t\rightarrow0}E(\mu^{t})=E(\mu).\label{eq:convergence of energy t}
\end{equation}
 Indeed, proceeding as above
\[
-E(\mu^{t})=\int\Psi(e^{t}x,e^{t}y)\mu(x)\otimes\mu(y)
\]
 where $t\mapsto\Psi(e^{t}x,e^{t}y)$ is increasing (as shown in the
course of the proof of Theorem \ref{thm:conv for N for two-point etc}).
Hence, the convergence \ref{eq:convergence of energy t} follows from
the monotone convergence theorem. We can thus conclude the proof by
combining the convergence in \ref{eq:conv of energy epsilon} and
\ref{eq:conv of energy epsilon}, using a standard diagonal argument. 

Combining the previous two lemmas with Lemma \ref{lem:energy appr implies S finite etc}
we arrive at the following
\end{proof}
\begin{prop}
\label{prop:Se finite under Main }Under the Main Assumptions $S(e)$
is finite on $]e_{min},e_{max}[.$ 
\end{prop}

\subsection{Concavity of the microcanonical entropy $S_{+}^{(N)}(e)$}

The following result is a slight generalization of a result shown
in the course of the proof of \cite[Theorem 2.3]{b-b}, which, in
turn, is based on the main result in \cite{bern}.
\begin{prop}
\label{prop:concavity N finite general Psi}Let $Y$ be a pseudoconvex
domain, $\Psi$ a psh function on $Y$ and $\mu_{0}$ a measure on
$Y$ such that $\mu_{0}=e^{-\psi_{0}}d\lambda$ for a psh function
$\psi_{0}$ on $Y.$ Assume that $Y$ is endowed with a holomorphic
action of by compact group $G$ such that If $\Psi$ and $\mu_{0}$
are $G-$invariant and assume also that for any $t\in\R$ any $G-$invariant
holomorphic function on $\left\{ \Psi<t\right\} $ is constant. Then
either the function
\[
t\mapsto\log\mu_{0}\left\{ \Psi<t\right\} 
\]
is identically equal to $+\infty$ or concave. In particular, if $\mu_{0}$
is moreover assumed to have finite total mass, then the concavity
in question holds. 
\end{prop}

\begin{proof}
First recall the main result in \cite{bern}. Consider $\C^{n+1}$
with holomorphic coordinates $(z,t)\in\C^{n}\times\C.$ Let $\mathcal{D}$
be a pseudoconvex domain in $\C^{n+1}$ endowed with a psh function
$\psi(z,t).$ Denote by $D_{t}$ the subset of $\mathcal{D}$ obtained
by fixing the $t-$coordinate. According to the main result of \cite{bern}
the function $B_{t}(z)$ on $\mathcal{D}$ defined by 
\begin{equation}
B_{t}(z):=\sup\left\{ |f(z)|^{2}:\,f\,\text{holomorphic on \ensuremath{D_{t}} and}\int_{D_{t}}|f|^{2}e^{-\psi(\cdot,t)}d\lambda\leq1\right\} \label{eq:def of B t}
\end{equation}
has the property that either $\log B_{t}(z)$ is subharmonic in $t$
or identically equal to $-\infty.$ In the present case we take
\[
\mathcal{D}:=\{\Psi(z)-\Re(t)<0\}\subset Y\times\C
\]
 and $\psi(z,t)=\psi_{0}(z).$ Note that $\mathcal{D}$ is a pseudoconvex
domain in $\C^{n+1}.$ Indeed, this follows from Lemma \ref{lem:pseudo},
using that $\Psi(z)-\Re(t)$ is psh in $Y\times\C$ and $Y\times\C$
is pseudoconvex (since $Y$ is). Now, by assumption, the group $G$
acts holomorphically on $D_{t}.$ In particular, if $dG$ denotes
a $G-$invariant measure on $G$ (i.e. Haar measures) and $f$ is
holomorphic on $D_{t}$ then the function 
\[
f_{G}(z):=\int_{g\in G}f(g\cdot z)dG
\]
 is holomorphic and $G-$invariant. Hence, replacing $f$ with $f_{G}$
and using the ``triangle inequality'' the sup in formula \ref{eq:def of B t}
may as well be restricted to all $G-$invariant holomorphic functions
$f.$ But, by assumption, any such functions is constant and hence
we may as well take $f$ to be identically equal to $1.$ But this
means that $B_{t}(z)=1/\int_{Y\cap\{\Psi(z)<t\}}e^{-\psi_{0}}d\lambda.$
The theorem thus follows from the main result of \cite{bern}, recalled
above (also using that if $\phi(t)$ is subharmonic in $t$ and only
depends on $\Re(t)$ then $\phi(t)$ is convex wrt $t\in\R).$ 
\end{proof}
\begin{rem}
\label{rem:B-M}\emph{(Brunn-Minkowski inequality)} This proposition
can be viewed as a generalization of the classical fact that the logarithm
of the volume $\mu_{0}\left(\{\phi\leq t\}\right)$ is concave if
$\phi$ is a convex function on $\R^{n}$ and $\mu_{0}$ is a log
concave measures i.e. $\mu_{0}=1_{C}e^{-\phi_{0}(x)}d\lambda$ for
$C\subset\R^{n}$ a convex body and $\phi_{0}$ a convex function.
This is a consequence of the Brunn-Minkowski inequality, but it also
follows from the previous proposition by considering the map $L(z)=(\log|z_{1}|,...,\log|z_{n}|)$
from $(\C-\{0\})^{n}$ onto $\R^{n}$ which has the property that
$\phi(x)$ is convex iff $\psi:=L^{*}\phi$ is psh and hence $C$
is convex iff $Y:=L^{-1}(C)$ is pseudoconvex in $\C^{n}$ (using
that $L^{*}\phi$ is bounded from above and this extends to a psh
function on $\C^{n})$. The classical fact in question then follows
by taking $G$ as the $n-$dimensional compact torus, acting on $\C^{n}$
in the standard way (and thus preserving the fibers of the map $L)$. 
\end{rem}

It should be stressed that plurisubharmonicity alone is not enough
to ensure the concavity in the previous proposition, as illustrated
by the case when $Y$ is the unit-disc in $\C$ and $\Psi(z)$ is
the Green function for the Laplacian with a pole at $w\in Y,$ for
a given non-zero $w$ in the interior of $Y,$ i.e. $\Psi(z)=\log|(z-w)/(1-\bar{w}z)|$
Indeed, as shown in the proof of \cite[Thm 2.3]{b-b} the concavity
in question is then equivalent to the subharmonicity of the Schwartz
symmetrization of $\Psi(z),$ which only holds when $a$ is zero,
i.e. when $\Psi$ is $S^{1}-$invariant, as pointed out in the introduction
of \cite{b-b}.

We next apply the previous proposition to the case when the $N-$particle
Hamiltonian $H^{(N)}$ on $X^{N}$ comes from a pair-interaction $W$
and exterior potential $V$ such that $-W$ and $-V$ and also $\phi$
and $\Psi_{0}$ satisfy the Main Assumptions (but there is no need
to assume that $W(x,y)$ is symmetric or that the indices $i,j$ range
over all of $\{1,2,...,N\})$: 
\begin{thm}
\label{thm:conv for N for two-point etc}If the Main Assumptions hold
and $H^{(N)}$ is the function on $X^{N}$ defined by 
\[
H^{(N)}(x_{1},...,x_{N}):=a_{ij}\sum_{(i,j)\in\mathcal{I}}W(x_{i},x_{j})+b_{i}\sum_{j\in\mathcal{J}}V(x_{j})
\]
for some subsets $\mathcal{I}$ of $\{1,..,N\}^{2}$ and $\mathcal{J}$
of $\{1,...,N]$ and non-negative constants $a_{ij}$ and $b_{i}.$
Then

\[
S_{+}^{(N)}(e):=\log\mu_{0}^{\otimes N}\left\{ H^{(N)}>e\right\} 
\]
 is concave in $e$ and finite when $e>\sup_{X^{N}}H^{(N)}.$ In particular,
$S_{+}^{(N)}(e)$ is concave when $H^{(N)}$ is the mean field Hamiltonian
\ref{eq:Hamilt for W and V intro}.
\end{thm}

\begin{proof}
It is a standard fact that the closure of the orbits of the vector
field $\mathcal{V}_{a}$ \ref{eq:def of vector field} coincide with
the orbits of a compact torus $G$ acting holomorphically on $\C^{n}.$
The assumptions imply that $\phi,V,W$ and $\mu_{0}$ are invariant
under the action of $G$ (using the diagonal action on $X\times X).$
By assumption $X$ admits a continuous psh exhaustion function $\rho.$
By the construction in Lemma \ref{lem:pseudo} $\rho$ may as well
be assumed to be $G-$invariant. Since the maximum of a finite number
of psh functions is still psh the function $\rho_{N}$ on $X^{N}$
defined by 
\[
\rho(z_{1},...,z_{N}):=\max_{i=1,...,N}\rho(z_{i})
\]
 is psh and $G-$invariant wrt the diagonal action of $G$ on $X^{N}$
and thus defines a $G-$invariant continuous psh exhaustion function
of $X^{N}.$ The corollary will thus follow from the previous theorem
applied to $X^{N},$ $\Psi_{N}:=-H^{(N)}$ and the measure $\mu_{0}^{\otimes N}$
if $X^{N},$ once the assumptions on $G$ have been verified. To this
end consider the holomorphic action of the additive group $\C$ on
$\C^{n}$ defined as follows: given $\tau\in\C$ and $z\in\C^{n}$
\begin{equation}
\tau\circledcirc z:=(e^{a_{1}\tau}z_{1},...,a^{a_{n}\tau}z_{n}).\label{eq:t ring z}
\end{equation}
Note that, for $z$ fixed, $\tau\circledcirc z\rightarrow0$ as the
real part $\Re\tau\rightarrow-\infty.$ Moreover, if $\phi$ is a
psh function on $\C^{n}$ then 
\begin{equation}
\phi(\tau\circledcirc z)\leq\phi(z)\,\,\,\text{if \ensuremath{\Re\tau\leq0}}\label{eq:ineq for phi with action}
\end{equation}
To see this first first observe that since $\phi$ is psh and the
orbits of the $\C-$action define holomorphic curves the function
$\phi(\tau):=\phi(\tau\circledcirc z)$ is subharmonic for a fixed
$z.$ Moreover, since $\phi\in PSH_{\boldsymbol{a}}$ the function
$\phi(\tau)$ is independent of the imaginary part $\Im\tau$ and
hence $\phi(\tau)$ is convex wrt the real part $t$ of $\Re\tau.$
Since $\phi$ is bounded from above close to the origin it follows
that there exists a constant $C$ such that $\phi(\tau)\leq C$ as
$\Re\tau\rightarrow-\infty.$ But then the convexity of $\phi(t)$
implies that $d\phi(t)/dt\rightarrow0$ as $\tau\rightarrow-\infty$
in $\R.$ Hence, by convexity, $\phi(\tau)$ is increasing in the
real part of $\tau,$ proving the inequality. As a consequence, $X$
is $G-$invariant, connected and the origin $0$ is contained in the
interior of $X$ (using that the action by $\C$ is locally free).
Moreover, the same thing goes for the sublevel sets $\{\Psi_{N}(z_{1},...,z_{N})<t\}.$
Indeed, by the previous argument $\Psi_{N}(\tau\circledcirc z_{N},...,\tau\circledcirc z_{N})$
is increasing wrt the real part of $\tau.$ In particular, the minimum
of $\Psi$ is attained at the origin in $X^{N},$ which implies that
$0$ is an interior point of $\{\Psi_{N}(z_{1},...,z_{N})<t\},$ as
long as $t>\inf_{X^{N}}\Psi_{N}.$ Thus if $f$ is a holomorphic function
on $\{\Psi_{N}(z_{1},...,z_{N})<t\},$ then in order to verify that
$f$ is constant on $\{\Psi_{N}(z_{1},...,z_{N})<t\}$ it is enough
to verify that its Taylor expansion at the origin $0$ in $\C^{nN}$
is a constant. To simplify the notation we will prove this when $N=1$
(but the general case if the same up to a change of notation). Using
multinomial notation the action of $\nu_{\boldsymbol{a}}$ on $f(z)$
close to the origin in $\C^{n}$ gives, by Taylor expansion of $f,$
\[
\nu_{\boldsymbol{a}}(f)=\nu_{\boldsymbol{a}}(\sum_{\alpha_{i}\geq0}c_{\boldsymbol{\alpha}}z_{1}^{\alpha_{1}}\cdots z_{n}^{\alpha_{n}})=\sum_{\alpha_{i}\geq0}i\boldsymbol{a}\cdot\boldsymbol{\alpha}c_{\boldsymbol{\alpha}}z_{1}^{\alpha_{1}}\cdots z_{n}^{\alpha_{n}}.
\]
Since $a_{i}>0$ the scalar product $\boldsymbol{a}\cdot\boldsymbol{\alpha}$
is non-vanishing for $\boldsymbol{\alpha\neq0.}$Hence, $\nu_{\boldsymbol{a}}f=0$
can only hold if the Taylor coefficients $c_{\boldsymbol{\alpha}}$
vanish for $\boldsymbol{\alpha\neq0.}$ Since $X$ is connected it
follows that $f$ is identically constant (by the identity principle
for holomorphic functions).

As for the finiteness of $S_{N}^{+}(e),$ for $e>\sup_{X^{N}}H^{(N)},$
it follows directly from the fact that any psh function is usc, hence
the subset where $H^{(N)}>e$ is open.
\end{proof}
When $H^{(N)}$ is replaced by the ``attractive'' Hamiltonian $-H^{(N)}$
the previous theorem also shows that $S_{-}^{(N)}(e)$ (formula \ref{eq:def of S N minus})
is concave. The following simple example illustrates the relevance
of the plurisubharmonicity assumption in the previous theorem. Consider
the ``attractive'' Hamiltonian obtained by taking $W=0$ in formula
\ref{eq:Hamilt for W and V intro} and assume that $V$ is $S^{1}-$invariant.
Then $V$ is psh iff $V=\phi(\log|z|)$ for a convex increasing function
$\phi$ on $\R.$ Assuming that $\phi(x)$ is \emph{strictly} increasing
and $N=1$ we get

\begin{equation}
\mu_{0}^{\otimes N}\left\{ H^{(N)}\leq e\right\} =C_{n}e^{2nf(e)},\label{eq:vol of sublevelset when N one}
\end{equation}
 where $f(e)$ is the function defined, on the image of $\phi,$ as
the inverse of $\phi(x)$ and $C_{nN}$ is the volume of the unit-ball
in $\R^{2n}.$ Hence, the logarithm of $\mu_{0}^{\otimes N}\left\{ H^{(N)}\leq e\right\} $
is concave iff $f$ is concave iff its inverse $\phi$ is convex iff
$V$ is psh. In the case when $N\geq1$ an illustrative class of ``attractive''
Hamiltonians is given by the case when $V$ is a power-law, $V=|z|^{\alpha}$
for $\alpha>0$ (i.e. $\phi(x)$ is the convex function $=e^{\alpha x}).$
Then a simple scaling argument reveals that the volume $\mu_{0}^{\otimes N}\left\{ H^{(N)}\leq e\right\} $
is of the form \ref{eq:vol of sublevelset when N one} for $f(e)=\log e$
when $n$ is replaced by $nN/\alpha,$ if $e\geq0.$ Hence, the logarithm
of $\mu_{0}^{\otimes N}\left\{ H^{(N)}\leq e\right\} $ is indeed
concave. Note that in this example the logarithm of the surface area
of $\{H^{(N)}=e\},$ i.e. of the derivative of $\mu_{0}^{\otimes N}\left\{ H^{(N)}\leq e\right\} ,$
is \emph{not }concave unless $N$ is taken sufficently large; $N\geq\alpha/2n$
(otherwise it is convex). 
\begin{rem}
\label{rem:joint cond on W and V}In the case of the mean field Hamiltonian
$H^{(N)}$ the concavity of $S_{+}^{(N)}(e)$ holds more generally
if the assumptions on $W$ and $V$ are replaced by the weaker assumption
that the negative of $W(x,y)+\frac{N}{N-1}\left(V(x)+V(y)\right)$
is in $PSH_{\boldsymbol{a},\boldsymbol{a}}(X\times X).$ Indeed, rewriting
\[
H^{(N)}(x_{1},...,x_{N})=\frac{1}{2}\frac{1}{N}\sum_{i\neq j\leq N}\left(W(x_{i},x_{j})+\frac{N}{N-1}\left(V(x_{i})+V(x_{j})\right)\right)
\]
 we can then apply the previous theorem to the mean field Hamiltonian
corresponding to the pair-interaction $W(x,y)+\frac{N}{N-1}(V(x)+V(y)).$
\end{rem}

\subsubsection{Incorporating constraints}

Prop \ref{prop:concavity N finite general Psi} may be generalized
by replacing $\Psi$ with a finite number of functions $\psi_{1},...,\psi_{r}$
on $Y$ satisfying the same assumptions as $\Psi$ and replacing the
sublevel $\left\{ \Psi(<t\right\} $ with the intersection of the
sub-level sets $\{\psi_{1}<t_{1}\},....,\{\psi_{r}<t_{r}\}$ for a
given $\boldsymbol{t}=(t_{1},...,t_{r})\in\R^{r}.$ Then the logarithm
of the corresponding volume defines a concave function of $\boldsymbol{t}\in\R^{r}.$
Indeed, one simply replaced the domain $\mathcal{D}$ in the proof
with the intersection of the pseudo-convex domains $\{\psi_{1}(z)-\Re(t_{1})<0\}$
in $\C^{n}\times\C^{r}.$ Since the intersection of pseudo-convex
domains is pseudo-convex the main result in \cite{bern} then implies
that the corresponding function $\log B(\boldsymbol{t})$ is a psh
function of $\boldsymbol{t}\in\C^{r}$and thus, by translational invariance
in the imaginary arguments, it defines a convex function on $\R^{r}.$
As a consequence if one assumes given $\psi_{1},...,\psi_{r}$ as
above then Theorem \ref{thm:convex cont iff energy appr} may be generalized
to the statement that the ``constrained microscopic entropy''
\[
\log\mu_{0}^{\otimes N}\left\{ H^{(N)}(\boldsymbol{z}_{1},...,\boldsymbol{z}_{N})>e,\,\sum_{i=1}^{N}\psi_{1}(\boldsymbol{z}_{i})\leq l_{1},...,\sum_{i=1}^{N}\psi_{r}(\boldsymbol{z}_{i})\leq l_{r}\right\} 
\]
 is a concave function of $(e,l_{1},...,l_{r})\in\R^{1+r}.$ In particular,
when $X\subset\C^{n}$ this applies to $\psi_{i}(z_{1},..,z_{n})=\lambda_{i}|z_{i}|^{2}$
for given positive numbers $\lambda_{1},...,\lambda_{n},$ as in the
Gaussian case discussed in Section \ref{subsec:Priors-versus-linear}.
Anyhow, in this paper we will, for simplicity, stick to the non-constrained
setup. On the other hand, as pointed out in Section \ref{subsec:Priors-versus-linear},
the constraints may be incorporated in the prior measure.

\subsection{Concavity of the entropy $S(e)$ when $e_{0}\protect\leq e$}

Now assume that $X\Subset\R^{2n}$ that we identify with $\C^{n},$
as usual. 
\begin{thm}
\label{thm:concave S in text} If the Main Assumptions holds then
the  entropy $S(e)$ is a decreasing concave continuous function on
$[e_{0},e_{max}[$ 
\end{thm}

In order to prove this we first assume that $X$ is compact and invoke
the following 
\begin{prop}
\cite{e-s} Assume that $W$ and $V$ are continuous and $X$ is compact.
Then 
\[
\lim_{N\rightarrow\infty}S_{+}^{(N)}(e)=S_{+}(e):=\sup_{E(\mu)\geq e}S(\mu)
\]
\end{prop}

\begin{proof}
Consider the open interval $\Delta:=\{t>e\}$ in $\R.$ By \cite[Thm 2.1]{e-s}
the limsup and liminf of $S_{+}^{(N)}(e)$ is equal to the sup of
$S(\mu)$ over all $\mu\in\mathcal{P}(X)$ such that $E(\mu)\in\overline{\Delta}$
and $E(\mu)\in\Delta,$ respectively. But by Lemma \ref{lem:decreasing}
below both these quantities are equal to $S_{+}(e).$ 
\end{proof}
When $W$ and $V$ are continuous and $X$ is compact Theorem\ref{thm:conv for N for two-point etc}
thus shows that $S_{+}(e)$ is a limit of concave functions on $\R$
and thus concave on $\R.$ Next, we invoke the monotonicity properties
shown in Lemma \ref{lem:decreasing} below which show that 
\[
S_{+}(e)=\max\{e_{0},S(e)\}
\]
 and hence $\max\{e_{0},S(e)\}$is is concave. But by Lemmas \ref{lem:energy appr implies S finite etc},
\ref{lem:superharmonic gives energy approx} $S(e)$ is finite for
any $e\in]e_{min},e_{max}[.$ Hence, $\max\{e_{0},S(e)\}$ is concave
and finite on $]e_{min},e_{max}[$ and thus concave and continuous
on $]e_{min},e_{max}[$. This proves the theorem in the case when
$W$ and $V$ are continuous and $X$ is compact. 

\subsubsection{Conclusion of the proof of Theorem \ref{thm:concave S in text}}

Still assuming that $X$ is compact we will next show that $S(e)$
is decreasing, concave and continuous when $e\in[e_{0},e_{max}[,$
i.e. when 
\begin{equation}
E(\mu_{0})\leq e<\sup_{\mathcal{P}(X)}E(\mu)\label{eq:e in interval}
\end{equation}
We will proceed by an approximation argument and exploit that $S(e)>-\infty$
(Prop \ref{prop:Se finite under Main }). Take a sequence $W_{\delta}(x,y)$
of continuous pair-interactions increasing to $W$ satisfying the
Main Assumptions. For example, $W_{\delta}$ may be defined as a convolution
of $W$ with a compactly supported smooth density $\rho_{\delta}.$
First observe that since, by assumption, $E(\mu_{0})<e,$ we get $E_{\delta}(\mu_{0})<e$
for $\delta$ sufficiently small (by the monotone convergence theorem).
Thus, by the concavity and continuity of $S_{+,\delta}(e)$ on $\R$
established in the previous section, we just have to verify that 
\begin{equation}
\lim_{\delta\rightarrow0}S_{\delta}(e)=S(e)\label{eq:conv delta zero}
\end{equation}
for any fixed $e$ satisfying the inequalities in formula \ref{eq:e in interval}.
To this end first note that 
\begin{equation}
S_{\delta}(e)\leq S(e).\label{eq:upper bound on S delta}
\end{equation}
 Indeed, by Lemma \ref{lem:decreasing}, it is enough to prove the
corresponding inequality for the \emph{upper}  entropies, where it
follows directly from the assumption that $E_{\delta}(\mu)\leq E_{0}(\mu).$
 Now fix a candidate $\mu$ for the sup defining $S_{0}(e)$ and set
$e_{\delta}:=E_{\delta}(\mu).$ Then 
\[
S_{0}(\mu)\leq S_{\delta}(e_{\delta})=S(\mu_{\delta}),
\]
 where $\mu_{\delta}$ realizes the sup defining $S_{\delta}(e_{\delta}).$
Moreover, fixing a positive number $\epsilon$ we have 
\[
e_{\delta}\geq e-\epsilon
\]
 for $\delta$ sufficiently small ($\delta<\delta_{\epsilon}).$ Hence,
since $S_{\delta}$ is decreasing when $e>E_{\delta}(\mu_{0})$ (by
Lemma \ref{lem:decreasing}), we get 
\[
S_{0}(\mu)\leq S_{\delta}(e-\epsilon)
\]
 for $\delta<\delta_{\epsilon}.$ Using that $S_{\delta}(e)$ is concave
we thus deduce that 
\[
S_{\delta}(e_{\delta})\leq S_{\delta}(e)+\epsilon|\frac{d}{de}S_{\delta}(e)|.
\]
Combining the latter inequality with the inequality \ref{eq:upper bound on S delta}
reveals that all that remains, in order to prove the convergence \ref{eq:conv delta zero},
is to verify that 
\begin{equation}
|\frac{d}{de}S_{\delta}(e)|\leq C\label{eq:uniform bound}
\end{equation}
 as $\delta\rightarrow0.$ To this end first observe that, using again
that $S_{\delta}(e)$ is decreasing and concave yields for any fixed
$e'>e$ 
\[
|\frac{d}{de}S_{\delta}(e)|=-\frac{d}{de}S_{\delta}(e)\leq\frac{S_{\delta}(e)-S_{\delta}(e')}{e'-e}\leq\frac{-S_{\delta}(e')}{e'-e}.
\]
In particular, if $e'$ is a fixed number satisfying $e<e'<\sup_{\mathcal{P}(X)}E(\mu)$
we get (by Lemma \ref{lem:decreasing}) that 
\[
S_{\delta}(e')\geq S(\mu')
\]
 for any $\mu'\in\mathcal{P}(X)$ such that $E_{\delta}(\mu')\geq e'.$
Now, by Lemmas \ref{lem:energy appr implies S finite etc}, \ref{lem:superharmonic gives energy approx}
below $\mu'$ can be chosen, independently of $\delta,$ so that $E(\mu')\geq e'+\epsilon$
and $S(\mu')>-\infty.$ We then get, for any $\delta$ sufficiently
small, that $E_{\delta}(\mu')\geq e'$ (by the monotone convergence
theorem) and thus the uniform bound \ref{eq:uniform bound} follows.

This concludes the proof of Theorem \ref{thm:concave S in text} in
the case when $X$ is compact. In the general case we fix $R>0$ and
denote by $X_{R}$ the intersection of $X$ with a ball $B_{R}$ of
radius $R$ centered at the origin. Then $X_{R}$ is also pseudoconvex
(as follows from Lemma\ref{lem:pseudo} applied to $\phi(z)=|z|^{2}-R).$
Thus, as shown in the previous section, the  entropy $S_{R}(e)$ associated
to the restrictions to $B_{R}$ of $W,V$ and $\mu_{0}$ is concave
in $e$ for $e>E(\mu_{0}).$ Hence all that remains is to verify that
\begin{equation}
\lim_{R\rightarrow\infty}S_{R}(e)=S(e)\label{eq:conv R infty}
\end{equation}
for any fixed $e$ satisfying the inequalities \ref{eq:e in interval}.
To this end first note that, since $X_{R}\subset X,$ it follows immediately
that $S_{R}(e)\leq S(e).$ Now assume that $e>E(\mu_{0})$ and fix
a candidate $\mu$ for the sup defining $S(e).$ Set 
\[
\mu_{R}:=1_{B_{R}}\mu/\mu(B_{R}).
\]
Then 
\[
S(\mu)\leq S_{R}(e_{R}),\,\,\,e_{R}:=E(\mu_{R}).
\]
By the monotone convergence theorem $E(\mu_{R})\rightarrow E(\mu).$
Hence, using that $S_{R}(e)$ is decreasing and concave for $R$ sufficiently
large (by the previous step) we can proceed essentially as when approximating
$W$ with $W_{\delta}$ above, to get 
\[
S(\mu)\leq\limsup_{R\rightarrow\infty}S_{R}(e),
\]
 which concludes the proof of the convergence \ref{eq:conv R infty}
and thus the concavity in Theorem \ref{thm:concave S in text}. 

\section{\label{sec:global con}Global concavity of $S(e)$ and examples}

Recall that, in classical terminology, a symmetric function $W(x,y)$
is a \emph{weakly positive definite kernel,} i.e. that for any positive
integer $N$
\[
\sum_{i,j\leq N}W(x_{i},x_{j})a_{i}a_{j}\geq0\,\,\,\forall(a_{i})\in\R^{N}:\,\sum_{i=1}^{N}a_{i}=0
\]
If the first inequality holds for \emph{any} sequence $(a_{i})_{i=1}^{N},$
then $W(x,y)$ is called a \emph{positive definite kernel.} \footnote{$W(x,y)$ is a weakly positive kernel iff $-W(x,y)$ is a\emph{ negative
definite kernel} in the terminology of \cite{b-c-r}.}

Now assume that $W(x,y)$ is weakly positive definite and satisfies
the Main Assumptions in a neighborhood of $X\times X.$ Then $W$
can be expressed as increasing limit of \emph{continuous }(and even
smooth)\emph{ }such functions $W_{\delta}(x,y).$ Indeed, if $\rho$
is a smooth compactly supported probability density on $\R^{2n}\times\R^{2n}$
we can take
\begin{equation}
W_{\delta}:=(W*\rho_{\delta}):=\int W(\cdot+a,\cdot+b)\rho_{\delta}(a,b)d\lambda(a)d\lambda(b),\,\,\,\rho_{\delta}(x,y)=\rho(\delta^{-1}x,\delta^{-1}x)\delta^{4n}.\label{eq:def of W delta conv}
\end{equation}
 Since $-W$ is psh $W_{\delta}$ indeed increases to $W.$ Moreover,
since $W(\cdot+a,\cdot+b)$ is weakly positive definite and satisfies
the Main assumptions for any $(a,b)$ so does $W_{\delta}.$ 
\begin{thm}
\label{thm:Main-assumptions+weakly-positiv}If the Main Assumptions
hold and moreover $W(x,y)$ is assumed weakly positive definite, then
the  entropy $S(e)$ is globally concave and hence thermodynamic equivalence
of ensembles holds for any $e\in]e_{min},e_{max}[.$
\end{thm}

\begin{proof}
By Theorem \ref{thm:concave S in text} $S(e)$ is concave and continuous
on $[e_{0},e_{max}[.$ Next, recall the classical fact that a weakly
positive definite kernel defines a convex functional $E_{W}(\mu)$
on $\mathcal{P}(X)$ (and vice versa). Hence, by Prop \ref{prop:E convex plus energy approx implies S str conc}
$S(e)$ is concave and continuous on $]e_{min},e_{0}]$ when $X$
is compact and $W$ is continuous. The theorem thus follows, in the
compact and continuous case, from the second and third point in Prop
\ref{prop:thermo equiv under energy appr}. Next assume that $X$
is still compact and define $W_{\delta}$ to be a regularization as
in formula \ref{eq:def of W delta conv}. Then the corresponding  entropy
$S_{\delta}(e)$ is concave on $]e_{min,\delta},e_{max,\delta}].$
Moreover, by the approximation argument used in the proof of Theorem
\ref{thm:concave S in text} and the finiteness of $S(e)$ the function
$S_{\delta}(e)$ converge point-wise to $S(e)$ on $]e_{min},e_{0}[.$Thus
$S(e)$ is also concave on $]e_{min},e_{0}[.$ Finally, the general
non-compact case is deduced from the compact case using again the
approximation arguments in  the proof of Theorem \ref{thm:concave S in text}
and the finiteness of $S(e).$

\end{proof}
Recall that, by Bochner's classical theorem, a translationally invariant
kernel $W(x,y)=\mathcal{W}(x-y)$ is positive definite iff the function
$\mathcal{W}$ on $\R^{d}$ is the Fourier transform of a (positive)
measure on $\R^{d}.$ In the case of translationally and rotationally
invariant kernels the following classical result holds \cite{b-c-r}: 
\begin{lem}
(Bernstein+Schoenberg). Let $w(r)$ be a continuous function on $[0,\infty[$
which is smooth on $]0,\infty[.$ Then $W(x,y):=w(|x-y|)$ is a positive
definite kernel iff $f(r):=w(r^{1/2})$ is \emph{completely monotone},
i.e. $(-1)^{m}\partial^{m}f(r)/dr^{m}\geq0$ for all non-negative
integers $m.$ 
\end{lem}

The previous lemma implies that if $w$ is non-negative on $[0,\infty[$
(but possibly equal to $\infty$ at $r=0)$ and $w(r^{1/2})$ is completely
monotone for $r>0$, then $w(|x-y|)$ is still positive definite.
Indeed, one can apply the previous lemma to 
\[
w_{\epsilon}(r):=w\left((r^{2}+\epsilon)^{1/2}\right)
\]
 and then let $\epsilon\rightarrow0.$ 
\begin{cor}
Under the Homogeneous Assumptions together with the assumption that
$w(r^{1/2})$ is completely monotone for $r>0$ the  entropy $S(e)$
is globally concave and thermodynamic equivalence of ensembles holds
at all energies.
\end{cor}

It should be pointed out that assumptions in the previous corollary
are preserved if $w$ is replaced by $w_{\epsilon}(r)$ above (using
that $\log(|z|^{2}+\epsilon)$ is psh) and similarly for $v$ and
$\psi_{0}.$ This gives a convenient explicit regularization procedure
preserving the property that $S(e)$ is globally concave. 

\subsection{\label{subsec:Examples}Examples where $S(e)$ is globally concave}

We next provide some examples where Theorem \ref{thm:Main-assumptions+weakly-positiv}
applies and thus $S(e)$ is globally concave. More examples may, for
example, be obtained by taking convolutions (as in formula \ref{eq:def of W delta conv}).
Note also that if the entropy $S_{W,V}(e)$ corresponding to the interactions
$W$ and $V$ is globally concave, then so is $S_{-W,-V}(e),$ since
$S_{-W,-V}(e)=S_{W,V}(-e).$ In this way one may thus go from a situation
of repulsive interactions to attractive ones.

Theorem \ref{thm:Main-assumptions+weakly-positiv} applies to the
case when $W(x,y)=-\log|z-w|$ when, for example, $\mu_{0}$ is Lebesgue
measure on for example a ball in $\R^{2n}$ or a centered (possibly
non-standard) Gaussian measure in $\R^{2n}$ (as in formula \ref{eq:gaussian measures}).
Indeed, then $W$ satisfies the Homogeneous Assumptions and the positive
definiteness follows, for example, from the fact that $W$ is the
Green kernel on $\R^{2n}$ of the $n$ th power of the Laplacian,
which is positive definite as a formally self-adjoint operator. More
generally, the ball may be replaced with any domain $X$ satisfying
the Main Assumptions, for example 
\[
X=\{z\in\R^{2n}:\,\sum_{i=1}^{r}|P_{i}(z)|^{\alpha_{i}}\leq1\},
\]
 for a quasi-homogeneous polynomials $P_{1},...,P_{r}$ and $\alpha_{i}>0$
(see Example \ref{exa:quasi-homo}). 

Theorem \ref{thm:Main-assumptions+weakly-positiv} also applies to
the continuous repulsive \emph{power-laws} with exponent in $]0,2]$
\[
W(x,y)=-|x-y|^{a},\,\,a\in]0,2],
\]
 as well as to 
\[
W(x,y)=e^{-\alpha|x-y|^{a}},\,\,a\in]0,2].
\]
 when $X$ is taken to be a disc centered at the origin with radius
at most $(1/2\alpha)^{1/a}.$ Indeed, a direct computation reveals
that $w(r)$ satisfies the Homogeneous Assumptions for any $a,\alpha>0$
(by a scaling it is enough to verify the case when $a=\alpha=1$)
Moreover, by \cite[Cor 3.3]{b-c-r} (and its proof) the kernels in
question are weakly positive definite when $a\in]0,2].$ Note that
in the case of the repulsive logarithmic interaction, as well as for
repulsive power-laws with $a\in]0,2[,$ Prop \ref{prop:existence non-compact}
ensures the existence of maximum entropy measures $\mu^{e},$ when
$\mu_{0}$ is a centered Gaussian measure (by taking $\psi_{0}=|x|^{2}).$

\subsubsection{The point vortex model}

Consider the point vortex model (for vortices with identical circulations)
on a domain $X$ in $\R^{2}.$ In the case when $X=\R^{2}$ 
\[
W(x,y)=-\log|x-y|,\,\,\,V(x)=0
\]
(with our normalizations). As discussed in the previous section, $S_{+}^{(N)}(e)$
and $S(e)$ are both globally concave (and thermodynamic equivalence
of ensemble holds) if $\mu_{0}$ is taken to be a centered Gaussian
measure. As indicated in \cite[Section 5]{clmp2}, the concavity of
$S(e)$ also follows from the results in \cite{clmp2}, using completely
different techniques (see also\cite{k2} where the concavity of the
corresponding multi-variable entropy $S(e,l),$discussed in Section
\ref{subsec:Priors-versus-linear}, is shown). But, as discussed in
the introduction of the paper, the main point of the present technique
is that it also applies to regularizations of $W.$

In the case of when $X$ is a compact domain with smooth boundary
$W(x,y)$ is defined as the negative of Green function $G_{X}(x,y)$
for the Laplacian on $X$ with Dirichlet boundary conditions and $V(x)=\gamma(x)/N$
where $\gamma$ is the restriction to the diagonal of $G_{X}(x,y)+\log|x-y|$
\cite{clmp1,clmp2,m-p}. In particular, when $X$ is the unit-disc
\begin{equation}
W(z,w)=-\log\frac{|z-w|}{|1-z\bar{w}|},\,\,\,V(x)=\frac{1}{N}\log|1-|z|^{2}|\label{eq:W and V in disc for vortex}
\end{equation}
In this case Theorem \ref{thm:conv for N for two-point etc} implies
that $S_{+}^{(N)}(e)$ is globally concave when $N\leq3,$ as follows
from combining Remark \ref{rem:joint cond on W and V}with the following
lemma, proved in the appendix.
\begin{lem}
\label{lem:psh in disc}Denote by $D$ the interior of the unit-disc
in $\C$ and set 
\[
\psi(z,w):=\log\left(|z-w|^{2}/|1-z\bar{w}|^{2}\right),\,\,\,\phi(z)=-\log\left(\left|1-|z|^{2}\right|^{2}\right)
\]
The function $\psi(z,w)+\lambda\left(\phi(z)+\phi(w)\right)$ is psh
in $D\times D$ iff $\lambda\geq1/2.$ 
\end{lem}

We leave open the question whether $S_{+}^{(N)}(e)$ is concave also
when $N>3.$ As for $S(e)$ it was shown to be concave in \cite{clmp2},
using a completely different method. In the case when a rotationally
invariant exterior potential $V_{e}$ is added to $V(x)$ in formula
\ref{eq:W and V in disc for vortex}, the previous lemma shows that
$S_{+}^{(N)}(e)$ is concave for any $N$ (and hence also $S(e)$)
if $-\partial\bar{\partial}V_{e}\geq\partial\bar{\partial}\phi/2$
in $D$ i.e. if the Laplacian of $V_{e}$ is sufficiently negative:
\[
\frac{1}{4}\Delta V_{e}(z)\leq-\frac{1}{(1-|z|^{2})^{2}}.
\]
 This should be contrasted with the fact that the global concavity
of $S(e)$ may fail if the Laplacian is positive, e.g. in the case
when $V_{e}(z)=\omega|z|^{2},$ for $\omega>0,$ studied in \cite{s-o}
and \cite[Lemma 8.2]{clmp2}.

\subsubsection{Insulated plasmas and self-gravitating matter in 2D }

The point vortex model on a compact domain $X$ is physically equivalent
to a one-component Coulomb plasma if inertial effects are ignored
(i.e. the limit of infinite damping is considered) and the boundary
of $X$ is assumed to be conductive \cite{s-o}. On the other hand,
the case when the boundary of $X$ is non-conducting, i.e. $X$ is
insulated, corresponds to the mean field Hamiltonian on $X$ with
Coulomb pair-interaction $-\log|x-y|$ (and $V\equiv0$) \cite{ga}.
In this case the Main Assumptions apply when $\mu_{0}$ is the uniform
measure on the $X$ unit-disc $X,$ as discussed in the beginning
of Section \ref{subsec:Examples}. More generally, the Main assumptions
apply when the exterior potential $V$ is radial and $\Delta V\leq0,$
i.e. $V$ is the potential induced by a distribution of fixed particles
with the same charge as the plasma. Switching the sign of the Coulomb
interaction yields a system of self-gravitating matter, studied in
\cite{al} with inertial effects included. Generalizations to regularized
self-gravitating matters are given in \cite{be2} (briefly outlined
in Section \ref{subsec:Outlook}).

\section{\label{sec:Critical-inverse-temperatures}Critical inverse temperatures
and existence of maximum entropy measures}

In the Very General Setup the\emph{ macroscopic inverse temperatures}
is defined by 
\begin{equation}
\beta_{c}:=\inf\left\{ \beta\in\R:\,\inf_{\mu}F_{\beta}(\mu)>-\infty\right\} .\label{eq:beta c intro}
\end{equation}
The \emph{microscopic inverse temperature} $\beta_{c,N}$ is, in the
General Setup, defined by 

\[
\beta_{c,N}:=\left\{ \beta\in\R:\,Z_{N,\beta}:=\int_{X^{N}}e^{-\beta H^{(N)}}(e^{-\Psi_{0}}dx)^{\otimes N}<\infty\right\} 
\]
 and respectively, where $H^{(N)}$ denotes the mean field Hamiltonian
\ref{eq:Hamilt for W and V intro} corresponding to $W$ and $V.$

\subsection{Dual expressions for the critical inverse temperatures}

We start with the following dual ``slope formula'' for $\beta_{N,c},$
under the Main Assumptions, which also shows that $\beta_{N,c}<0.$ 
\begin{cor}
\label{cor:integrability threshold as slope}Under the same assumptions
as in Prop \ref{prop:concavity N finite general Psi} the following
holds if $\mu_{0}$ has finite mass on $Y$ and $\Psi$ is not identically
constant: 
\[
c_{(Y,\mu_{0})}(\Psi):=-\inf\left\{ \beta\in]-\infty,0]:\int_{Y}e^{\beta\Psi}\mu_{0}<\infty\right\} =\lim_{e\rightarrow-\inf_{Y}\Psi}\frac{d}{de}\log\left(\mu_{0}\left\{ \Psi<-e\right\} \right),
\]
 using either right or left derivatives in the rhs. As a consequence,
the set of all negative $\beta$ such that $\int_{Y}e^{\beta\Psi}\mu_{0}<\infty$
is open. In particular, under the Main Assumptions
\[
\beta_{N,c}=\lim_{e\rightarrow\sup_{X^{N}}E_{N}}\frac{dS^{(N)}(e)}{de}<0,\,\,\,Z_{N,\beta_{N_{c}}}=\infty
\]
\end{cor}

\begin{proof}
By Prop \ref{prop:concavity N finite general Psi} (and Theorem\ref{thm:conv for N for two-point etc})
the function 
\[
\phi(t):=-\log\mu(t),\,\,\,\mu(t):=\left(\mu_{0}\left\{ \Psi<-t\right\} \right)
\]
 is convex wrt $t\in\R.$ Consider first the case when $t_{0}:=\inf_{Y}\Psi>-\infty.$
Then, trivially, $\beta_{c}=-\infty.$ Moreover, $\phi(t)$ is convex
and finite for $t>t_{0}$ and $\phi(t)\rightarrow\infty$ as $t$
decreases to $t_{0}.$ But this forces $d\phi(t)/dt\rightarrow-\infty$
as $t$ decreases to $t_{0}.$ Indeed, by the convexity of $\phi$
the limit of $d\phi(t)/dt$ decreases to $M_{0}\in[-\infty,\infty[$
as $t$ decreases to $t_{0}.$ Assume, to get a contradiction, that
$M_{0}>-\infty.$ Then, fixing $t_{1}>t_{0}$ gives $\phi(t)\leq\phi(t_{1})+|M||t_{1}-t_{0}|<\infty$
as $t\rightarrow t_{0},$ which contradicts that $\phi(t)\rightarrow\infty$
as $t$ decreases to $t_{0}.$

Next, assume that $\inf_{Y}\Psi=-\infty.$ Since $\beta\leq0$ we
have 
\[
\int_{Y}e^{\beta\Psi}\mu_{0}\leq\int_{\{\Psi<0\}}e^{\beta\Psi}\mu_{0}+\mu_{0}(Y),
\]
where, by assumption, the second term is finite. Pushing forward the
measure $\mu_{0}$ on $Y$ to $\R$ under the map $z\mapsto\Psi(z)$
gives
\[
\int_{\{\Psi<0\}}e^{\beta\Psi}\mu_{0}=\int_{-\infty}^{0}e^{\beta t}\frac{dV(t)}{dt}dt=-\beta\mathcal{Z}(\beta)+V(0),\,\,\,\mathcal{Z}(\beta):=\int_{-\infty}^{0}e^{\beta t}V(t)dt,
\]
where the second equality follows from integrating by parts. We may
then conclude the proof of the first formula in the corollary by expressing
\[
\mathcal{Z}(\beta):=\int_{-\infty}^{0}e^{\beta t-\phi(t)}dt
\]
and applying Lemma \ref{lem:exponential integral over half-line}
below to the convex function $\Phi=\beta t-\phi(t),$ which implies
that 
\begin{equation}
\int_{Y}e^{\beta\Psi}\mu_{0}<\infty\iff-\beta<\lim_{t\rightarrow\infty}\frac{d\phi}{dt},\label{eq:slope cond in pf}
\end{equation}
 concluding the proof of formula in question. To prove that $\beta_{N,c}<0$
note that $\phi(t)\rightarrow\infty$ as $t\rightarrow-\infty$ and
$\phi(t)\rightarrow0$ as $t\rightarrow\infty.$ Since $\phi(t)$
is convex if follows that, using either left or right derivatives,
$\lim_{t\rightarrow-\infty}d\phi(t)/dt\leq0$ and $\lim_{t\rightarrow-\infty}d\phi(t)/dt=0.$
But if $\beta_{N,c}=0,$ then, by the previous step, $\lim_{t\rightarrow-\infty}d\phi(t)/dt=0$
and hence, by convexity, $\phi(t)$ is constant. But this can only
happen if $\Psi$ is constant, which is excluded by the assumptions.
Thus $\beta_{N,c}<0,$ as desired. Finally, to prove the last openness
statement we just have to verify that if $\int_{Y}e^{\beta\Psi}\mu_{0}<\infty,$
then there exists $\delta>0$ such that $\int_{Y}e^{(\beta-\delta)\Psi}\mu_{0}<\infty.$
But this follows directly from the strict inequality in the right
hand side of formula \ref{eq:slope cond in pf}.
\end{proof}
\begin{rem}
\label{rem:integrab threshold and openess}In the case when $Y$ is
compact and $\Psi_{0}=0$ (or, equivalently, bounded) the number $c_{Y}(\Psi)$
is called the\emph{ integrability threshold} of $\Psi$ on $Y$ (or
the \emph{complex singularity exponent}) in the complex geometry literature
(whose inverse is the \emph{Arnold multiplicity}) \cite{d-k}. It
follows from Skoda's local integrability inequality that $c_{Y}(\Psi)>0$
for any function which is psh on a neighborhood of $Y$ and not identically
$-\infty.$ Moreover, $\int_{Y}e^{\beta\Psi}d\lambda=\infty$ in the
critical case $\beta=-c_{Y}(\Psi),$ by the resolution of the openness
conjecture in \cite{bern2} (see also \cite{g-z} for the resolution
of the strong openness conjecture). The proof above yields a simplification
of the proof in \cite{bern2} under the symmetry assumption that $\Psi\in PSH(Y)_{\boldsymbol{a}}$
(anyhow, just like \cite{bern2}, it is based on \cite{bern}). 
\end{rem}

In the proof above the following elementary fact was used:
\begin{lem}
\label{lem:exponential integral over half-line}Let $\Phi(t)$ be
a convex function on $]-\infty,0[$ such that $\Phi(t)$ is bounded
as $t\rightarrow0.$ Then 
\[
\int_{-\infty}^{0}e^{-\Phi(t)}dt<\infty
\]
 iff $\lim_{t\rightarrow-\infty}d\Phi(t)/dt<0,$ using either left
or right derivatives.
\end{lem}

\begin{cor}
\label{cor:slope formula for beta c}Consider the Main Assumptions
and assume also that $S(e_{max})=-\infty,$ if $e_{max}<\infty.$
Then, as $e$ increases strictly towards $e_{max}$
\begin{equation}
\beta_{c}=\lim_{e\rightarrow e_{max}}\frac{dS(e\pm)}{de},\label{eq:beta c as slope in Cor text}
\end{equation}
 where $dS(e\pm)/ds$ denotes either the left or the right derivative
of the concave function $S(e_{\pm}).$
\end{cor}

\begin{proof}
By Theorem \ref{thm:concave S in text} $S(e)$ is concave and continuous
on $[e_{0},e_{max}[.$ Denote by $\tilde{F}$ the usc concave function
defined as $F$ when $\beta\leq0$ and as $-\infty$ when $\beta>0.$
By Prop \ref{prop:thermo equiv under energy appr} $S=(\tilde{F})^{*}$
on $[e_{0},e_{max}[.$ Set $g:=(\tilde{F})^{*}.$ Thus $g$ is constant
for $e\leq e_{0}$ and on $[e_{0},e_{max}[$ it coincides with $S(e)$
(by \ref{prop:thermo equiv under energy appr}). Moreover, $g^{*}=\tilde{F}$
and hence $\overline{\{g^{*}<\infty\}}=[\beta_{c},\infty[.$ Thus,
by formula \ref{eq:gradient im}, 
\[
[\beta_{c},\infty[=\overline{\partial g(\{g>-\infty\})}=\overline{\partial S(]e_{0},e_{max}[}),
\]
 which proves formula \ref{eq:beta c as slope in Cor text}, using
that $dS(e+)/ds\leq dS(e-)/ds$ and $dS(e+)/ds$ and $dS(e-)/ds$
are both decreasing (by concavity). 
\end{proof}

\subsection{\label{subsec:Concrete-expression-for}Concrete expressions in the
homogeneous case}

It seems natural to expect that, under rather general assumptions,
$\beta_{N,c}\rightarrow\beta_{c}$ as $N\rightarrow\infty.$ Here
we will show that this is the case under the Homogeneous Assumptions;
in fact, $\beta_{N,c}=\beta_{c}$ for any $N.$ The starting point
is the following essentially well-known consequence of the Gibbs variational
principle (compare \cite{k,clmp1,be0}):
\begin{lem}
\label{lem:general bounds on the partition function}Let $H^{(N)}$
be a mean field Hamiltonian of the form \ref{eq:Hamilt for W and V intro}.
Then 
\[
Z_{N,\beta}\leq\int_{X}e^{-\beta V(x)}\mu_{0}(x)\left(\int e^{-\beta\left(\frac{1}{2}W(x,y)+V(y)\right)}\mu_{0}(y)\right)^{N-1}
\]
 and 
\begin{equation}
-\frac{1}{N\beta}\log Z_{N,\beta(N-1)N}\leq\inf_{\mu\in\mathcal{P}_{0}(X)}F\left(\beta\right)=:F(\beta)\label{eq:Gibbs}
\end{equation}
As a consequence, $\beta_{c}\leq\limsup_{N\rightarrow\infty}\beta_{N,c}$
and if there exists $\beta_{0}<0$ such that
\begin{equation}
\sup_{x\in X}\int e^{-\beta_{0}\left(\frac{1}{2}W(x,y)+V(y)\right)}\mu_{0}(y)<\infty,\,\,\,\int_{X}e^{-\beta_{0}V}\mu_{0}<\infty\label{eq:uniform integral prop}
\end{equation}
then $\beta_{N,c}<\beta_{0}$ and $\beta_{c}<\beta_{0}.$ 
\end{lem}

\begin{proof}
First observe that it will be enough to consider the case when $V=0$
(otherwise we just replace $\mu_{0}$ with $e^{-\beta V}\mu_{0}$).
Decompose $-\beta H^{(N)}=\frac{1}{N}\sum_{i=1}^{N}f_{i},$ where
$f_{i}$ is the sum of $\frac{1}{2}W(x_{i},x_{j})$ over all $j$
such that $j\neq i.$ The arithmetic-geometric means inequality gives
\[
\int_{X^{N}}e^{-\beta H^{(N)}}\mu_{0}^{\otimes N}\leq\sum_{i=1}^{N}\frac{1}{N}\int_{X^{N}}e^{f_{i}}\mu_{0}^{\otimes N}=\int_{X}\mu_{0}\left(\int e^{-\beta\frac{1}{2}W(x,y)}\mu_{0}(y)\right)^{N-1}.
\]
 Hence, estimating the latter integral over $X$ with the sup over
$X$ proves the first inequality in the proposition. To prove the
second one first note Gibbs variational principle (Jensen's inequality)
gives: for any given $\mu\in\mathcal{P}(X)$ 
\[
-\frac{1}{N\beta}\log Z_{N,\beta}:=\int_{X^{N}}e^{-\beta NE^{(N)}}\mu_{0}^{\otimes N}\leq\beta\int_{X^{N}}E^{(N)}\mu^{\otimes N}-S(\mu),\,\,\,E^{(N)}:=H^{(N)}/N
\]
 as long as the right hand side is well-defined. In the case when
$H^{(N)}$ is of the form in the lemma 
\[
\int_{X^{N}}E^{(N)}\mu^{\otimes N}=\frac{1}{N}\frac{1}{(N-1)}N(N-1)E(\mu)=\frac{N-1}{N}E(\mu),
\]
 which proves \ref{eq:Gibbs}, by taking the infimum over $\mu.$ 
\end{proof}
The following result generalizes the case of the logarithmic interaction
considered in \cite{k,clmp1}. 
\begin{prop}
\label{prop:critical beta as w prime}Under the Homogeneous Assumptions
in $\R^{d}$ (but allowing $d$ to be odd)

\[
\beta_{c}=\beta_{c,N}=\frac{2d}{\dot{w}},\,\,\,\,\dot{\,w}:=\lim_{t\rightarrow-\infty}\frac{dw(e^{t})}{dt}=\lim_{t\rightarrow-\infty}\frac{w(e^{t})}{t}
\]
 if $v$ and $\psi_{0}$ are assumed bounded in a neighborhood of
$0.$ Moreover, $Z_{N,\beta}=\infty$ when $\beta=\frac{4n}{\dot{w}}.$

\end{prop}

\begin{proof}
To simplify the notation we will prove the proposition in the case
when $V=0$ (but the proof in the general case is essentially the
same). First observe that
\begin{equation}
\sup_{X}\int_{X}e^{-\frac{\beta}{2}W(x,y)}\mu_{0}(y)<\infty\iff\int_{0}^{1}e^{-\frac{\beta}{2}w(r)}r^{d}\frac{dr}{r}<\infty\iff\beta>\frac{2d}{\dot{w}}\label{eq:sup finite iff integral finite}
\end{equation}
Indeed, since $w$ is decreasing, $w(r)\leq C$ if $r\geq1$ and hence,
using that $\mu_{0}$ is a probability measure,
\[
\int_{X}e^{-\frac{\beta}{2}W(x,y)}\mu_{0}(y)=\int_{X}e^{-\frac{\beta}{2}w(|x-y|)}\mu_{0}(y)\leq\int_{X\cap\{|x-y|\leq1}e^{-\frac{\beta}{2}w(|x-y|)}\mu_{0}(y)+e^{-\frac{\beta}{2}C}
\]
Changing variables in the integral above and setting $\gamma:=-\beta$
yields
\[
\int_{X\cap\{|x-y|\leq1}e^{\frac{\gamma}{2}w(|x-y|)}\mu_{0}(y)=\int_{\{|z|\leq1\}}e^{\frac{\gamma}{2}W(|z|)}e^{-\psi_{0}(x+z)}d\lambda(z)\leq C'\int_{\{|z|\leq1\}}e^{\frac{\gamma}{2}W(|z|)}d\lambda(z)
\]
using that $\psi_{0}$ is bounded from below. This proves \ref{eq:sup finite iff integral finite},
using Lemma \ref{lem:exponential integral over half-line} in the
last equivalence (by setting $t:=\log r)$. Hence, applying the previous
lemma gives

\begin{equation}
\beta_{N,c}<\frac{2d}{\dot{w}}\label{eq:upper bound on beta N c in pf}
\end{equation}
To prove that $\beta_{N,c}\geq2d/\dot{w}$ we restrict the integration
over $X^{N}$ to a ball $B_{R}$ of radius $R$ centered at the origin
and use that $w$ is decreasing to get 
\[
Z_{N,\beta}\geq\int_{B_{R}^{N}}e^{-\frac{N(N-1)}{2N}w(R)}\mu_{0}^{\otimes N}\geq Ce^{-\frac{\beta N}{2}w(R)}(R^{d})^{N}\geq C'
\]
Setting $R=e^{t}$ thus gives 
\[
(Z_{N,\beta})^{1/N}\geq C^{1/N}e^{-t\left(\frac{\beta}{2}\frac{1}{2t}w(e^{t})-d\right)}
\]
Hence, if $\beta<2d/\dot{w,}$ then as $R\rightarrow0,$ i.e. $t\rightarrow-\infty$
we get $(Z_{N,\beta})^{1/N}\geq C^{1/N}e^{-t\delta}$ for some $\delta>0.$
This means that $Z_{N,\beta}=\infty,$ which proves $\beta_{N,c}=2d/\dot{w}.$
Moreover, if $\beta=2d/\dot{w}$ then the argument shows that the
integral of $e^{-\beta NE^{(N)}}\mu_{0}^{\otimes N}$ over $B_{R}^{N}$
does not tend to zero as $R\rightarrow0.$ Since $\mu_{0}$ does not
charge single points it follows that $Z_{N,\beta}=\infty$ (for $d$
even this is a special case of the last statement in Cor \ref{cor:integrability threshold as slope}). 

Next, thanks to the second inequality in Lemma \ref{lem:general bounds on the partition function}
the inequality \ref{eq:upper bound on beta N c in pf} implies that

\[
\beta_{c}\leq\frac{2d}{\dot{w}}
\]
All that remains is thus to verify the reversed inequality. To this
end fix $\beta$ such that $F(\beta)>-\infty,$ i.e. such that there
exists a constant $C$ such that
\begin{equation}
\beta E(\mu)-S(\mu)\geq-C\label{eq:lower bd on free energy in pf}
\end{equation}
For $\epsilon>0$ set $\nu_{\epsilon}=(T_{\epsilon})_{*}\nu_{0}$
where $\nu_{0}$ is any fixed probability measure such that $S(\nu_{0})>-\infty.$
Then, on the one hand, as $t:=(\log\epsilon)$ tends to $-\infty$
\[
\frac{1}{t}E(\nu_{e^{t}})=\frac{1}{2}\int_{X^{2}}\frac{1}{2t}w(e^{t}|x-y|)\nu_{0}(x)\nu_{0}(y)\rightarrow\frac{1}{2}\dot{w}
\]
 by the monotone convergence theorem (using that the integrand is
monotone in $t,$ by concavity). On the other hand,
\[
S(\nu_{\epsilon})=S(\nu_{0})+d\log\epsilon
\]
Hence, applying the inequality \ref{eq:lower bd on free energy in pf}
to $\nu_{\epsilon}$ and dividing both sides with $t$ implies, by
letting $t\rightarrow-\infty,$ that
\[
\frac{\beta}{2}\dot{w}-d\leq0.
\]
 This shows that $\beta_{c}\geq\frac{2d}{\dot{w}},$ as desired. 
\end{proof}
\begin{rem}
Remarkably, it is always the case that $F(\beta_{c})<\infty$ when
$X$ is a compact domain in $\R^{d}$ and $W(x,y)=-\log(|x-y|$. Indeed,
this follows from Adam's generalization of the Moser-Trudinger inequality
in $\R^{2},$ as discussed in \cite{be} (see also \cite{k,clmp1}$)$.
This finiteness should be contrasted with the general divergence $Z_{N,\beta_{c}}=-\infty$
for any $N$ (see Cor \ref{cor:integrability threshold as slope}). 
\end{rem}

Note that if $X$ is compact and $W$ is finite, then $\beta_{c}=\beta_{c,N}=-\infty,$
but the converse does not hold, as illustrated by an application of
the previous proposition to the case when
\[
W(x,y)=\log(\log1/|x-y|).
\]

\subsection{\label{subsec:Bound-on-the non-isotr}The anisotropic case}

Consider now the case when the Main Assumptions hold and $W$ is translationally
invariant
\begin{equation}
W(z,w)=-\Psi(z-w),\,\,\,\Psi\in PSH_{\boldsymbol{a}}(\C^{n}),\label{eq:transl inv in bounding}
\end{equation}
 but not not necessarily isotropic. More generally, since we will
only be concerned with integrability properties we allow that \ref{eq:conv of weighted energy}
only holds up to a bounded term.
\begin{prop}
\label{prop:transl invariant}Consider the Main Assumptions and assume
moreover that $W$ is translationally invariant (up to a bounded term).
Then there exists a positive number $\gamma$ such that
\[
\max\{\beta_{N,c},\beta\}\leq-\gamma<0
\]
\end{prop}

\begin{proof}
First assume that $V=0.$ Then the first integral appearing in the
uniform integrability property \ref{eq:uniform integral prop} may,
after making the change of variables $z=y-z,$ be estimated as
\begin{equation}
\int e^{\frac{\beta}{2}\Psi(y-x)}e^{-\Psi_{0}(y)}dy=\int e^{\frac{\beta}{2}\Psi(z)}e^{-\Psi_{0}(z+x)}dz\leq\left(\int e^{\frac{p\beta}{2}\Psi(z)}dz\right)^{1/p}\left(\int e^{-q\Psi_{0}(z+x)}dz\right)^{1/q},\label{eq:estimate in pf translational}
\end{equation}
 using Hölder's inequality with conjugate exponents $p$ and $q.$
By the translational invariance of Lebesgue measure the integral in
the second factor is given by the integral of $\int e^{-q\Psi_{0}(y)}dy$
and thus independent of $x.$ Moreover, it follows from the openness
statement in Cor\ref{cor:integrability threshold as slope} that the
integral is finite for $q$ sufficiently close to $1.$ Similarly,
we can then make the integral in the first factor finite by taking
$\beta$ negative, but sufficiently close to $0.$ Finally, in the
case when $V$ is not identically zero we first apply the Cauchy-Schwartz
inequality to estimate
\[
\left(\int e^{-\beta\left(\frac{1}{2}W(x,y)+V(y)\right)}\mu_{0}(y)\right)^{2}\leq\int e^{-2\beta\frac{1}{2}W(x,y)}\mu_{0}(y)\int e^{-2\beta V(y)}\mu_{0}(y)
\]
 and then repeat the previous argument to both integrals appearing
in the right hand side. 
\end{proof}
Next, consider the case when $\Psi$ has an isolated singularity at
the origin, i.e. $\Psi$ is locally bounded on the complement of the
origin. Then one gets the following concrete bound, expressed in terms
of the integrability threshold $c_{0}(\Psi)$ of $\Psi$ on a ball
$B_{\epsilon}$ centered at the origin in $\C^{n}$ of sufficiently
small radius $\epsilon$ (discussed in Remark \ref{rem:integrab threshold and openess}). 
\begin{prop}
Let $X$ be a compact subspace of $\C^{n}$ and assume that $\Psi$
has an isolated singularity at the origin and that $V$ and $\Psi_{0}$
are bounded. Then, for any sufficiently small $\epsilon$ 
\[
\max_{N\geq2}\{\beta_{N,c},\beta_{c}\}=-\frac{1}{2}c_{0}(\Psi)<0.
\]
\end{prop}

\begin{proof}
The assumptions ensure that the bounds \ref{eq:uniform integral prop}
in Lemma \ref{lem:general bounds on the partition function} hold
iff the sup is replaced with an integral i.e. iff $Z_{2,\beta_{0}}<\infty$
iff $\int_{B_{\epsilon}}e^{\beta_{0}\frac{1}{2}\psi}d\lambda<\infty$
(as seen by changing variables as in the first equality in formula
\ref{eq:estimate in pf translational}). Hence, we can conclude using
the very definition of $c_{0}(\Psi).$ 
\end{proof}
In fact, as discussed in Remark \ref{rem:integrab threshold and openess}
it is enough to assume that $PSH(\C^{n}).$ The invariant $c_{0}(\Psi)$
plays a key role in current complex geometry and can be estimated
from below in terms of certain multiplicities (expressed as local
intersection numbers) \cite{d-p}. In the ``algebraic'' case in
formula \ref{eq:algebraic case intro} the integrability threshold
$c_{0}(\Psi)$ coincides with the\emph{ log canonical threshold} at
$0\in\C^{n}$ of the ideal in the polynomial ring $\C[z_{1},..,z_{n}]$
generated by the corresponding polynomials $P_{j}(z)$ \cite{mu}. 
\begin{example}
The log canonical threshold can be computed using algebro-geometric
techniques. For example, when $\Psi(z)=\log\left(|z_{1}|^{2\alpha_{1}}+...+|z_{n}|^{2\alpha_{n}}\right)$
for positive real numbers $\alpha_{i}$ one gets $c_{0}(\psi)=1/\alpha_{1}+...+1/\alpha_{n}$
\cite[Example 1.9]{mu}.
\end{example}

In the simplest case when $\Psi$ is ``algebraic quasi-homogeneous''
of degree $d$ (Example \ref{exa:quasi-homo}) with an isolated singularity
at the origin (i.e. the zero-locus of corresponding polynomials $P_{j}$
only intersect at the origin) we have, by homogeneity, that 
\[
\Psi=d\log|z|^{2}+\varphi(z),
\]
 for a positive number $d$ and a continuous function $\varphi,$
which descends to the compact quotient $(\C^{n+1}+\{0\})/\C_{\boldsymbol{a}}^{*}$
and is thus bounded. In this case it thus follows from Prop \ref{prop:critical beta as w prime}
that 
\[
\beta_{N,c}=\beta_{c}=\frac{4n}{d}.
\]
A wide variety of such $\Psi$ may be obtained by taking $P_{i}=\partial f(z)/\partial z_{i}$
for given quasi-homogeneous polynomial $f$ with an isolated degenerate
zero at the origin in $\C^{n}.$ Then $\Psi(z)$ can be expressed
in terms of a Ginzburg-Landau type potential:
\[
\Psi(z)=\log\left(\sum_{i}|\frac{\partial f}{\partial z_{i}}(z)|^{2}\right),
\]
 so that $W(z,w)$ is the standard logarithmic interaction precisely
when $f$ is proportional $z_{1}^{2}+...+z_{n}^{2}.$ 

\subsection{\label{subsec:Existence-of-maximum Main As}Existence of maximum
entropy measures}

Combining Prop \ref{prop:transl invariant} with  the results in Section
\ref{subsec:Existence-when compact} yields the following existence
result:
\begin{prop}
Consider the Main Assumptions when $X$ is compact. Then, for any
$e\in]e_{min},e_{0}[$ there exists a maximum entropy measure $\mu^{e}.$
If moreover $W(x,y)$ is assumed translationally invariant (up to
a bounded term) then there exists a maximum entropy measure $\mu^{e}$
for any $e\in[e_{0},e_{max}[.$ In particular, this is the case under
the Homogeneous Assumptions.
\end{prop}

Turning to the non-compact case we recall that, under the Main Assumptions,
\[
\mu_{0}=e^{-\Psi_{0}}d\lambda
\]
 for $\Psi_{0}\in PSH_{\boldsymbol{a}}(X).$ As a consequence, if
$\Psi_{0}$ is also assumed to be a continuous exhaustion function
(which is automatically the case if $\Psi_{0}$ is rotationally invariant),
then Prop \ref{prop:existence non-compact} implies the following
\begin{prop}
Consider the Main Assumptions and assume that $\Psi_{0}$ is continuous
exhaustion function and that the growth-assumption \ref{eq:growth assumption on W V}
holds for a $\phi_{0}$ such that $\phi_{0}/\Psi_{0}\rightarrow0$
uniformly as $|x|\rightarrow\infty.$ If $W(x,y)$ is assumed translationally
invariant (up to a bounded term), then there exists a maximum entropy
measure $\mu^{e}$ for any $e\in]e_{min},e_{max}[.$ 
\end{prop}

\begin{proof}
According to Prop Prop \ref{prop:existence non-compact} we just have
to verify that $\int e^{\delta\Psi_{0}}\mu_{0}<\infty$for some $\delta>0.$
But this follows from openness property in Cor \ref{cor:integrability threshold as slope}.
\end{proof}
For example, the previous proposition applies when $X=\R^{2n}$ endowed
with a centered Gaussian measure, $V=0$ and $W$ is of the ``algebraic
quasi-homogeneous'' form in Example \ref{exa:quasi-homo}. 

\section{\label{sec:strict}Strict concavity of $S(e)$}

In this final section we show how to deduce a stronger strict concavity
result for $S(e)$ under the Homogeneous Assumptions, using a uniqueness
result for minimizers of $F_{\beta}$ shown in the companion paper
\cite{be}. The starting point is the following criterion for the
strict concavity of $S(e)$ in the high energy region:
\begin{prop}
\label{pro:differentiability}Assume that $X$ is compact and that
$F_{\beta}$ has a unique minimizer on $\mathcal{P}(X)$ for any $\beta\in]\beta_{c},0[.$
If the energy approximation property holds and $E(\mu_{\beta})\rightarrow e_{max}$
as $\beta\rightarrow\beta_{c},$ then $S(e)$ is strictly concave
on $]e_{0},e_{max}[$ (in particular, this is the case if $E(\mu)$
is continuous on $\mathcal{P}(X)).$ 
\end{prop}

\begin{proof}
As pointed out in the proof of Prop \ref{prop:E convex plus energy approx implies S str conc}
the uniqueness assumption implies that $F(\beta)$ is differentiable.
Thus we can conclude by applying Lemma \ref{subsec:General-setup intro}.
\end{proof}
In the case of the point-vortex model the uniqueness assumption in
the previous proposition (and the energy approximation property) holds
on any simply connected compact domain $X$ \cite{clmp2}. Moreover,
by the concentration/compactness alternative established in \cite{clmp2},
the blow-up property holds iff $\mu_{\beta_{j}}$ converges weakly
towards a Dirac mass (such domains $X$ are called \emph{domains of
the first kind }in \cite{clmp2}). The following result is shown in
\cite{be}:
\begin{thm}
\label{thm:(Uniqueness)-Let-}(Uniqueness) Let $X$ be a ball centered
at the origin in $\R^{2n}$ or all of $\R^{2n}.$ Assume that $W$
and $V$ satisfy the Homogeneous Assumptions and that $v+\beta\psi_{0}$
is strictly concave wrt $\log$r when $r>0$ for a given $\beta<0.$
Then any minimizer of $F_{\beta}(\mu)$ is uniquely determined. If
the latter assumption is replaced by the assumption that $W(x,y)$
is a weakly positive definite kernel and that $w(r)$ is strictly
increasing, then minimizers are uniquely determined modulo translation
when $X=\R^{2n}$ and unique when $X$ is a ball.
\end{thm}

We finally arrive at the following
\begin{thm}
\label{Thm:strict concavity under special}Under the Homogeneous Assumptions,
the  entropy $S(e)$ is concave for $e>E(\mu_{0})$ and strictly concave
if $X$ is a ball and either $v$ is strictly concave wrt $\log$r
or $w$ is strictly increasing for $r\in]0,\infty[.$ If moreover
$W(x,y)$ is a weakly positive definite kernel, then $S(e)$ is strictly
concave on $]e_{min},e_{max}[.$
\end{thm}

\begin{proof}
First consider the case when $X$ is compact and $W,V$ and $\Psi_{0}$
are continuous and $v$ is strictly concave wrt $\log$r, when $r>0.$
Then the strict concavity of $S(e)$ follows directly from combining
the previous theorem  with Lemma \ref{lem:f diff implies S strict concav in gener}
(and similarly if $w$ is strictly increasing), using that  the energy
approximation property holds under the Main Assumptions and hence
also under the Homogeneous Assumptions. Next, if $v$ is not assumed
strictly concave wrt $\log$r, we replace $v$ with $v+\epsilon r.$
Then the corresponding  entropy $S_{\epsilon}(e)$ is concave and
letting $\epsilon\rightarrow0$ reveals that $S(e)$ is also concave.
The general case is then deduced from the previous case using the
approximation arguments employed in the proof of Theorem \ref{thm:concave S in text}.
Finally, if $W(x,y)$ is weakly positive definite, then by Prop \ref{prop:E convex plus energy approx implies S str conc}
$S(e)$ is also strictly concave on $]e_{min},e_{0}[$ and continuous
on $]e_{min},e_{0}].$ This means that $S(e)$ is strictly concave
on both $]e_{min},e_{0}[$ and $]e_{0},e_{max}[.$ Since $S$ is continuous
on $]e_{min},e_{max}[$ it follows that $S$ is strictly concave on
$]e_{min},e_{max}[$ (indeed, otherwise it would be affine on some
open interval in $]e_{min},e_{max}[$ which would contradict the strict
concavity on $]e_{min},e_{0}[$ or $]e_{0},e_{max}[.$
\end{proof}

\section{Appendix}

In this appendix we provide, for the convenience of the reader, some
proofs of essentially well-known results stated in Section \ref{sec:Preliminaries-and-notation}
and the proof of of Lemma \ref{lem:psh in disc}.

\subsection{Proof of Lemma \ref{lem:diff implies dual strc conc}}

Assume that $f$ is \emph{not} strictly concave in the interior of
$[y_{0},y_{1}].$ Then there exists an open interval $I\subset]y_{0},y_{1}[$
such that $f^{*}$ is affine on $I.$ In particular, there exists
a number $a$ such that $f'(y)\equiv a$ on $I.$ Note that $a\in[x_{0},x_{1}].$
Indeed, since $f$ is concave $(\partial f)(]-\infty,x_{0}])\Subset\{y\geq f'(x_{0})\}$
and $(\partial f)([x_{1},\infty[)\Subset\{y\leq f'(x_{1})\}.$ Hence,
by \ref{eq:y in gradient iff x in gradient}, $x\in[x_{0},x_{1}]$
and since $a\in(\partial f)(y)$ for any $y\in I$ it follows from
\ref{eq:y in gradient iff x in gradient} that $I\subset(\partial f)(a),$
showing that $f$ is not differentiable at $a.$ 

\subsection{Proof of Lemma \ref{lem:convex envol affine}}

First observe that $f^{**}$ is continuous on $U.$ Indeed, $f^{**}$
is concave (since it is an inf of affine functions) and hence it is
enough to check that $f$ is finite on $U.$ But $f^{**}\geq f$ and
by assumption $f>-\infty$ on $U.$ Moreover, by assumption there
exists a constant $C$ such that $f\leq C.$ Since the constant function
$C$ is a contender for the inf in formula \ref{eq:f star start as concave envelope}
it follows that $f^{**}\leq C,$ showing that $f^{**}$ is finite
and thus continuous on $U.$ As a consequence, $\Omega:=\{f^{**}>f\}\cap U$
is open in $U.$ Now fix a point $x_{0}\in\Omega.$ Since $\Omega$
is open there exists $x_{0,\pm}$ in $\Omega$ such that $x_{0,-}<x_{0}<x_{0,+}.$
Moreover, since $f$ is usc we may assume that the affine function
$a(x)$ on $[x_{0,-},x_{0,+}]$ with prescribed boundary values $a(x_{0,\pm})=f^{**}(x_{0,\pm})$
satisfies $f(x)<a(x)$ on $[x_{0,-},x_{0,+}].$ Hence, the continuous
function $\tilde{f}$ defined as $f$ on the complement of $[x_{0,-},x_{0,+}]$
and as $a(x)$ on $[x_{0,-},x_{0,+}]$ is concave on $\R$ and satisfies
$\tilde{f}\geq f.$ But then it follows that $\tilde{f}=f,$ by formula
\ref{eq:f star start as concave envelope}. 

\subsection{Proof of Lemma \ref{lem:pseudo}}

The function $\rho:=-\log(-\phi)$ on $\Omega$ is, clearly, an exhaustion
function. Moreover, $\rho$ is psh. Indeed, fix $z_{0}\in\Omega$
and $z\in\C^{n}$ and consider the restriction of $\rho$ to the complex
line $w\mapsto z_{0}+wz,$ parametrized by $w\in\C.$ First assume
that $\phi$ is smooth at $z_{0}$ and factorize the Laplacian on
$\C_{w}$ in the standard way, $\Delta=4\partial_{w}\partial_{\bar{w}}.$
We get
\[
\partial_{w}\partial_{\bar{w}}\rho=-\partial_{w}(\frac{\partial_{\bar{w}}(-\phi)}{-\phi})=\partial_{w}(\frac{\partial_{\bar{w}}\phi}{-\phi})=\frac{\partial_{w}\partial_{\bar{w}}\phi}{-\phi}-\partial_{\bar{w}}\phi\partial_{w}(\phi^{-1})=\frac{\partial_{w}\partial_{\bar{w}}\phi}{-\phi}+\partial_{\bar{w}}\phi\partial_{w}\phi\geq0+0.
\]
Hence, $w\mapsto\rho(z_{0}+wz)$ is subharmonic close to $z_{0}.$
The subharmonicity in the general case is shown in a similar way using
either distributional derivatives or a regularization argument. Finally,
denoting by $\rho_{Y}$ a continuous psh exhaustion function of $Y$
the maximum of $\rho$ and $\rho_{Y}$ defines a psh exhaustion function
of $\{\phi<0\}\cap Y.$ Indeed, in general, the maximum of two psh
functions is still psh (as follows from the corresponding standard
result for subharmonic functions). 

\subsection{Proof of Lemma \ref{lem:psh in disc}}

Since $\psi(z,w)$ is locally bounded from above it will be enough
to consider the complement of the diagonal in $D\times D$ (using
that the diagonal is pluripolar \cite[Thm 5.24]{dem}). In this region
$\log(|z-w|^{2})$ is pluriharmonic, i.e its complex Hessian vanishes
(since $\log|\xi|^{2}$ is harmonic when $\xi\neq0$). Hence, the
complex Hessian $\partial\bar{\partial}\psi$ coincides with $\partial\bar{\partial}$
applied to $-\log|1-z\bar{w}|^{2},$ i.e. to $-\log(1-z\bar{w})-\log(1-\bar{z}w).$
Accordingly, a direct computation yields 
\[
\partial\bar{\partial}\psi(z,w)=\left(\begin{array}{cc}
0 & (1-\bar{z}w)^{-2}\\
(1-z\bar{w})^{-2} & 0
\end{array}\right),\partial\bar{\partial}\left(\phi(z)+\phi(w)\right)=\left(\begin{array}{cc}
2(1-z\bar{z})^{-2} & 0\\
0 & 2(1-w\bar{w})^{-2}
\end{array}\right)
\]
In particular, when $\lambda=1/2$ we get 
\[
\partial\bar{\partial}\left(\psi(z,w)+\lambda\left(\phi(z)+\phi(w)\right)\right)=\left(\begin{array}{cc}
(1-z\bar{z})^{-2} & (1-\bar{z}w)^{-2}\\
(1-z\bar{w})^{-2} & (1-w\bar{w})^{-2}
\end{array}\right)
\]
 Since the trace of this Hermitian matrix is manifestly non-negative
the matrix is semi-positive definite iff its determinant is non-negative.
But the determinant is non-negative iff 
\[
(1-z\bar{z})^{2}(1-w\bar{w})^{2}\leq(1-z\bar{w})^{2}(1-\bar{z}w)^{2}\iff(1-z\bar{z})(1-w\bar{w})\leq(1-z\bar{w})(1-\bar{z}w),
\]
which in turn is equivalent to $-z\bar{z}-w\bar{w}\leq-z\bar{w}-z\bar{w}$
and hence also to the trivial inequality $0\leq|(z-w)|^{2}.$ The
same computation also reveals that when $\lambda<1/2$ the determinant
is negative at $(z,w)=(0,0)$ and hence also at any point in the complement
of the diagonal in $D\times D$ which is sufficiently close to $(0,0).$


\begin{thebibliography}{10}
\bibitem{al}Aly, J.J: Thermodynamics of a two-dimensional self-gravitating
system. Phys. Rev. E 49, 3771 (1994)

\bibitem{a-g-s}Ambrosio, L; Gigli, N; Savaree,G: Gradient flows in
metric spaces and in the space of probability measures Lectures in
Mathematics ETH Zurich. Birkhauser Verlag, Basel, 2005.

\bibitem{an}V. A. Antonov, Original inVest. Leningrad Univ.7, 135
(1962); English translation in Dynamics of Star Clusters (IAU Symposium,
Vol. 113), J. Goodman and P. Hut, ed.(Reidel,Dordrecht,1985),p.525.

\bibitem{a-c-p}Arora ,J.S. ,Chahande ,A.I. and Paeng ,J.K.: Multiplier
methods for engineeringoptimization. International J. for Numerical
Methods in Engineering. 32 (1991),1485--1525

\bibitem{b-c-r}Berg,C.,Christensen ,J.P.R; Ressel, P. (1984): Harmonic
Analysis on Semigroups. Graduate Texts in Mathematics100. Springer,
New York. Theory of positive definite andrelated functions.

\bibitem{be0}Berman, R.J: On Large Deviations for Gibbs Measures,
Mean Energy and Gamma-Convergence. Constructive Approximation volume
48, pages 3--30 (2018).

\bibitem{be1}Berman, R.J: The Coulomb gas, potential theory and phase
transitions. preprint arXiv:1811.10249(2018)

\bibitem{be2}Berman, R.J: Concavity of the microcanonical Gibbs entropy
for small systems and 2D-gravity (in preparation).

\bibitem{be}Berman, R.J: The Q-curvature equation on even-dimensional
conical spheres and uniqueness of free energy minimizers (in preparation).

\bibitem{be3}Berman, R.J: A maximum entropy principle for Sasaki-Einstein
metrics and AdS/CFT. (in preparation)

\bibitem{b-b0}Berman, R.J; Berndtsson, B: Real Monge-Ampère equations
and Kähler-Ricci solitons on toric log Fano varieties. Ann. Fac. Sci.
Toulouse Math. (6) 22 (2013), no. 4, 649--711.

\bibitem{b-b}Berman, R.J; Berndtsson, B: Symmetrization of Plurisubharmonic
and Convex Functions. Indiana Univ. Math. J. Vol. 63, No. 2 (2014),
pp. 345-365 

\bibitem{b-c-p}RJ Berman, TC Collins, D Persson: Emergent Sasaki-Einstein
geometry and AdS/CFT. arXiv:2008.12004, 2020

\bibitem{bern}Berndtsson, B:Subharmonicity properties of the Bergman
kernel and some other functions associated to pseudoconvex domains.,
Ann Inst Fourier 56 (2006) pp 1633-1662.

\bibitem{bern2}B. Berndtsson.The openness conjecture and complex
Brunn-Minkowski inequalities. Complex geometry and dynamics, 29--44,
Abel Symp.,10, Springer, Cham, 2015

\bibitem{b-t}J.Binney; S.Tremaine: Galactic Dynamics. Princeton Series
in Astrophysics (2008).

\bibitem{clmp1}Caglioti.E; Lions, P-L; Marchioro.C; Pulvirenti.M:
A special class of stationary flows for two-dimensional Euler equations:
a statistical mechanics description. Communications in Mathematical
Physics (1992) Volume 143, Number 3, 501-525

\bibitem{clmp2}Caglioti.E; Lions, P-L; Marchioro.C; Pulvirenti.M:
A special class of stationary flows for two-dimensional Euler equations:
a statistical mechanics description. Part II. Communications in Mathematical
Physics (1995) Volume 174, 229-260

\bibitem{c-d-r}A Campa, T Dauxois, S Ruffo: Statistical mechanics
and dynamics of solvable models with long-range interactions. Physics
Reports, 2009 - Elsevier

\bibitem{c-g-z}Chafaï, D; Gozlan, N; Zitt, P-A: First-order global
asymptotics for confined particles with singular pair repulsion. Ann.
Appl. Probab. 24 (2014), no. 6, 2371--2413. 

\bibitem{chang}S-Y A. Chang: Conformal invariants and partial differential
equations. Bull. Amer. Math. Soc. 42 (2005), 365-393 

\bibitem{c-k}Chanillo, S.;Kiessling, M. K.-H.: Rotational symmetry
of solutions to some nonlinearproblems in statistical mechanics and
geometry. Comm. Math. Phys.160(1994), 217--238 

\bibitem{ch}Chavanis, P-H: Statistical mechanics of two-dimensional
vortices and stellar systems. Pages 208-292 in \cite{d-r-a-w}.

\bibitem{d-r-a-w}T. Dauxois, S. Ruffo, E. Arimondo, and M. Wilkens,
eds. Dynamics and Thermodynamics of Systems with Long Range Interactions,
vol. 602 of Lecture Notes in Physics (Springer, New York, 2002).

\bibitem{d-z}Dembo, A; Zeitouni O: Large deviation techniques and
applications. Jones and Bartlett Publ. 1993

\bibitem{dem}Demailly, J: Complex Analytic and Differential Geometry.
https://www-fourier.ujf-grenoble.fr/\textasciitilde demailly/manuscripts/agbook.pdf

\bibitem{d-k}Demailly, J-P ; Kollár, J.: Semi-continuity of complex
singularity exponents and Kähler-Einstein metrics on Fano orbifolds.
Annales scientifiques de l'École Normale Supérieure, Série 4, Tome
34 (2001) no. 4, pp. 525-55

\bibitem{d-p}Jean-Pierre Demailly, Hoàng Hiep Pham: A sharp lower
bound for the log canonical threshold. Acta Math. 212 (1) 1 - 9, 2014

\bibitem{d-g-c}P Di Cintio, S Gupta, L Casetti: Dynamical origin
of non-thermal states in galactic filaments. Monthly Notices of the
Royal Astronomical Society, Volume 475, Issue 1, March 2018, Pages
1137--1147

\bibitem{d-f}C.F. Driscoll and K.S. Fine: Experiments on vortex dynamics
in pure electron plasmas. Phys. Fluids B 2 , 1359 (1990)

\bibitem{d-h}J. Dunkel and S. Hilbert: Consistent thermostatistics
forbids negative absolute temperatures. Nature Physics.10,67(2014)

\bibitem{e-s}Eyink, G.L., Spohn, H.: Negative-temperature states
and large-scale, long-lived vortices in two-dimensional turbulence
J. Stat. Phys.70, Nov. 3/4 (1993).

\bibitem{e-h-t}Ellis, R. S.; Haven, K; Turkington, B: Large deviation
principles and complete equivalence and nonequivalence results for
pure and mixed ensembles. J. Statist. Phys. 101 (2000), no. 5-6, 999--1064.

\bibitem{e-t-t}RS Ellis, H Touchette, B Turkington: Thermodynamic
versus statistical nonequivalence of ensembles for the mean-field
Blume--Emery--Griffiths model. Physica A: Statistical Mechanics
and its Applications Volume 335, Issues 3--4, 15 2004, Pages 518-538

\bibitem{ga}A. J. M. Garrett: Maximum Entropy with Nonlinear Constraints:
Physical Examples. In Maximum Entropy and Bayesian Methods pp 243-249,
Fundamental Theories of Physics book series (FTPH, volume 39).

\bibitem{gr}D.H. E. Gross: Microcanonical Thermodynamics: Phase Transitions
in \textquotedbl small\textquotedbl{} Systems. World scientific lecture
notes in physics. Vol 66 (2001).

\bibitem{g-z}Q. Guan, X. Zhou.A proof of Demailly\textquoteright s
strong openness conjecture. Ann. of Math. (2)182(2015), no.2, 605--616

\bibitem{gau}Gauthier et al: Giant vortex clusters in a two-dimensional
quantum fluid. Science. 364 (6447): (2019) 1264--1267. arXiv:1801.06951.

\bibitem{h-p-d}S Hilbert, P Hänggi, J Dunkel: Thermodynamic laws
in isolated systems. Physical Review E, 2014 - APS

\bibitem{hu}Huang, K: Statistical mechanics. J (Wiley, 1987)

\bibitem{ja}Jaynes, E. T.: Information Theory and Statistical Mechanics.
Phys. Rev., 106, 620 (1957)

\bibitem{jo}Johnstone et al: Evolution of large-scale flow from turbulence
in a two-dimensional superfluid . Science. 365 (6447) (2019) 1267--1271.
arXiv:1801.06952. 

\bibitem{k0}Kiessling M.K.H.: On the equilibrium statistical mechanics
of isothermal classical self-gravitating matter. J Stat Phys 55, 203--257
(1989)

\bibitem{k}Kiessling M.K.H.: Statistical mechanics of classical particles
with logarithmic interactions. Comm. Pure Appl. Math. 46 (1993), 27-56.

\bibitem{k2}Kiessling, M. K.-H.: The unbounded 2D guiding center
plasma,J. Plasma Phys.54(1995),11--29.

\bibitem{ki2}Kiessling, Michael K.-H.: Statistical mechanics approach
to some problems in conformal geometry. Statistical mechanics: from
rigorous results to applications. Phys. A 279 (2000), no. 1-4, 353--368. 

\bibitem{lb}D.Lynden-Bell: Negative specific heat in astronomy, physics
and chemistry. Physica A: Statistical Mechanics and its Applications.
Volume 263, Issues 1--4, 1 (1999) Pages 293-304

\bibitem{m-b}A. J. Majda; A. L. Bertozzi: Vorticity and incompressible
flow. Cambridge University Press (2010).

\bibitem{m-p}C. Marchioro, M. Pulvirenti, Mathematical Theory of
Incompressible Nonviscous Fluids, in: Applied Mathematical Sciences,
vol. 96, Springer-Verlag, New York, 1994

\bibitem{mu}Mustata, M; IMPANGA lecture notes on log canonical thresholds.
In Contributions to Algebraic Geometry: Impanga Lecture Notes. Editor:
P.Pragasz, EMS 2012.

\bibitem{ons}Onsager, L.: Statistical hydrodynamics, Nuovo Cim. Suppl.6,
279--287 (1949).

\bibitem{pa}Padmanabhan, T: Statistical mechanics of gravitating
systems. Phys. Rep.188, 285 (1990)

\bibitem{ra}Ramsey, N. F., Thermodynamics and statistical mechanics
at negative absolute temperatures, Phys. Rev.103, 20--28 (1956).

\bibitem{r}Rockafellar, R. T: Convex analysis. Reprint of the 1970
original. Princeton Landmarks in Mathematics. Princeton Paperbacks.
Princeton University Press, Princeton, NJ, 1997.

\bibitem{ru}H.H. Rugh: Microthermodynamic formalism. Phys. Rev. E64,
055101\textasciitilde 2001. 

\bibitem{si-r}C.H.SilvestreaT.M.Rocha Filho: Ergodicity in a two-dimensional
self-gravitating many-body system. Physics Letters A Volume 380, Issue
3, 28 January 2016, Pages 337-348

\bibitem{s-o}R. A. Smith and T. M. O\textquoteright Neil: Nonaxisymmetric
thermalequilibria of a cylindrically bounded guiding center plasma
or discrete vortex system. Phys. Fluids B2(1990), 2961--2975.

\bibitem{t-e-t}Touchette, H.; Ellis, R.S.; Turkington, B: An introduction
to the thermodynamic and macrostate levels of nonequivalent ensembles.
Physica A: Statistical Mechanics and its Applications Volume 340,
Issues 1--3, 1 2004, Pages 138-146

\bibitem{v}Villani, C.: Topics in optimal transportation. Graduate
Studies in Mathematics, 58. American Mathe-matical Society, Providence,
RI, 2003. xvi+370 pp
\end{thebibliography}
\end{document}